%% file: draft.tex
\def\editmode{0}

\def\bibfilenames{refs}

\def\spsformat{0}

\if\spsformat1

\documentclass{article}
\usepackage{spconf}

\else
\documentclass[journal]{IEEEtran}
\fi

\input{not_main/include_v3}

 \usepackage{balance}
\makeatletter
 \newcommand{\leqnomode}{\tagsleft@true}
\newcommand{\reqnomode}{\tagsleft@false}

 \makeatother

  \usepackage{algorithm}
 \usepackage{algorithmic}

\newtheorem{mydef}{Definition}
 \renewcommand{\arev}[1]{{#1}} 

\input{not_main/definitions.tex}

\begin{document}
\title{Fast Graph Filters for \\Decentralized Subspace  Projection}

\author{Daniel~Romero,~\IEEEmembership{Member,~IEEE,} Siavash Mollaebrahim,~\IEEEmembership{Student Member,~IEEE,} \\Baltasar~Beferull-Lozano,~\IEEEmembership{Senior Member,~IEEE}, and C\'esar~Asensio-Marco,~\IEEEmembership{Member,~IEEE.}
\thanks{Daniel Romero is with the Department of Information and Communication Technology, University of Agder, Norway. 
Siavash Mollaebrahim and Baltasar Beferull-Lozano are with the Intelligent Signal Processing and Wireless Networks (WISENET) Center, University of Agder, Norway. 
Cesar Asensio is with the AINIA Technology Center, Spain. 
e-mail:\{daniel.romero, siavash.mollaebrahim, baltasar.beferull\}@uia.no, ceamar@gmail.com}
\thanks{This work was supported in part by the PETROMAKS Smart-Rig grant
244205/E30 and the SFI Offshore Mechatronics grant 237896/O30 from the
Research Council of Norway.
}
\thanks{Part of this work was presented in the Int. Conf. on
Acoustics, Speech, and Signal Processing, Calgary, Canada, 2018~\cite{weerasinghe2018fast}.}
}
\maketitle
\begin{abstract}A number of inference problems with sensor networks involve projecting a
measured signal onto a given subspace. In existing decentralized
approaches, sensors communicate with their local neighbors to obtain a
sequence of iterates that asymptotically converges to the desired
projection. In contrast, the present paper develops methods that
produce these projections in a finite and approximately minimal number
of iterations. Building upon tools from graph signal processing, the
problem is cast as the design of a graph filter which, in turn, is
reduced to the design of a suitable graph shift operator. Exploiting
the eigenstructure of the projection and shift matrices leads to an
objective whose minimization yields approximately minimum-order graph
filters. To cope with the fact that this problem is not convex, the
present work introduces a novel convex relaxation of the number of
distinct eigenvalues of a matrix based on the nuclear norm of a
Kronecker difference. To tackle the case where there exists no graph
filter capable of implementing a certain subspace projection with a
given network topology, a second optimization criterion is presented
to approximate the desired projection while trading the number of
iterations for approximation error. Two algorithms are proposed
to optimize the aforementioned criteria based on the
alternating-direction method of multipliers. An exhaustive simulation
study demonstrates that the obtained filters can effectively obtain
subspace projections markedly faster than existing algorithms.

\end{abstract}

\section{Introduction}
\label{sec:intro}

\cmt{Motivation}
\begin{myitemize}%
\myitem\cmt{monitoring spatial fields vs (possibly W)SNs}A frequent
inference problem in signal processing involves the estimation of a
spatial field using measurements collected by a (possibly wireless)
sensor
network~\cite{raghavendra2006,chen2015ubiquitous,nevat2015heterogeneous,asensio2015energy,nowak2004inhomogeneous}.
The field of interest may quantify magnitudes such as temperature,
electromagnetic radiation, concentration of airborne or liquid
pollutants, flows of gas or liquids in porous soils and rocks such as
oil reservoirs, acoustic pressure, and radioactivity to name a
few. Instead of spatial fields, one may  be alternatively interested in fields
defined on the nodes or edges of a network; see
e.g.~\cite{forero2014dictionary}.
\myitem\cmt{Inference\ra projection}In either case, a number of
common inference tasks such as least squares estimation, denoising,
(weighted) consensus, and decentralized detection can be cast as projecting
the observations onto a given signal subspace; see
e.g.~\cite{kay1,harsanyi1994hyperspectral,behrens1994oblique}.
\myitem\cmt{decentralized}
\begin{myitemize}%
\myitem\cmt{centralized}Such a fundamental task can be implemented in
a centralized fashion,
\begin{myitemize}%
\myitem\cmt{description}where a fusion center gathers and processes
the measurements collected by all sensors.
\myitem\cmt{limitations}Unfortunately, this approach
\begin{myitemize}%
\myitem\cmt{comp}gives rise to i) communication bottlenecks, since those nodes
near the fusion center are required to forward data packets from many
sensors, 
\myitem\cmt{comm}ii) computational challenges, since the load is
concentrated in the fusion center,
\myitem\cmt{robustness}and iii)  vulnerability to attacks or failure of the fusion center.
\end{myitemize}%
\end{myitemize}%
\myitem\cmt{description}For these reasons, the decentralized
paradigm, where there is no central processor and all nodes share the
computational load, is oftentimes preferred~\cite{nedic2018decentralized}.
\myitem\cmt{strengths}These implementations are therefore scalable,
robust, and balance the communication and processing requirements across nodes.
\end{myitemize}%
\myitem\cmt{fast}The present paper capitalizes on the notion of graph
filter~\cite{shuman2013emerging,sandryhaila2014bigdata} to develop
algorithms for computing projections in a decentralized fashion with
an approximately minimal number of iterations.
\end{myitemize}%

\cmt{literature: approaches for subspace proj}
\begin{myitemize}%
\myitem\cmt{general inference}To tackle problems involving
decentralized processing, it is common to define a \emph{communication
graph} where each node represents a sensor and there exists an edge
between two nodes if the corresponding sensors can communicate,
e.g. via a radio link.  One could therefore feel inclined to address
the problem at hand using standard inference tools for data defined on
graphs; see
e.g.~\cite{belkin2004regularization,romero2017spacetime,romero2017multikernel,lauritzen1996graphical,isufi2019forecasting,scarselli2008graph}. However,
these methods are based on exploiting a certain relation between the
data and the graph;
e.g. smoothness~\cite{romero2017multikernel}. Therefore, this
framework fundamentally differs from the one at hand since here the
graph only provides information about how the sensors communicate,
i.e., it does not generally provide information about the spatial
field of interest.\footnote{\arev{If the field is smooth over space
and the graph is also smooth over space, meaning that nodes that are
spatially close have a low geodesic distance in the graph, then both
notions may approximately coincide. However, even in this case,
exploiting spatial information would still be more
accurate. Fortunately, as described in this paper, graph signal
processing tools can still be applied to exploit spatial smoothness
rather than graph smoothness. }}
\myitem\cmt{decentralized optimization}To obtain a decentralized
    subspace projection algorithm, one could instead adopt a
    decentralized optimization standpoint, e.g. via
\begin{myitemize}%
\myitem\cmt{works}
\begin{myitemize}%
    \myitem\cmt{DLMS}the \emph{distributed least mean squares} (DLMS)
    method in~\cite{mateos2012drls}, based on the \emph{alternating
    direction method of multipliers} (ADMM)~\cite{giannakis2016decentralized,boyd2011distributed},
    \myitem\cmt{DGD}or the \emph{decentralized gradient descent} (DGD)
    method in~\cite{shi2017distributed}, which builds upon gradient descent.
\end{myitemize}%
\myitem\cmt{limitations}Although these algorithms can accommodate
    general objective functions, their convergence is only
    asymptotic and can be significantly improved by exploiting the
    structure of the subspace projection problem; see
    Sec.~\ref{sec:sim}. 
\end{myitemize}%
\myitem\cmt{general projections}%
\begin{myitemize}%
\myitem\cmt{works}%
\begin{myitemize}%
\myitem\cmt{barbarossa and insausti}For this reason, a method tailored to computing projections in a
decentralized fashion was proposed in~\cite{barbarossa2009projection},
later extended in~\cite{insausti2012in-network} and~\cite{dilorenzo2020innetwork}, where every node
obtains each iterate by linearly combining its previous iterate with
the previous iterate of its neighbors.  The combination weights are
adjusted to achieve a fast asymptotic convergence.
\end{myitemize}%
\myitem\cmt{limitations \ra asymptotic}The main strength of this 
approach is its simplicity, since each node simply repeats the same
operation over and over. The price to be paid is that convergence is
asymptotic, which means that a large number of data packages need to
be exchanged to attain a prescribed projection accuracy. In addition,
these algorithms can only accommodate a limited set of
topologies~\cite{camaro2013reducing}.
\end{myitemize}%
\myitem\cmt{finte-time consensus}A special case of the subspace
projection problem is average consensus, where the signal subspace
comprises the vectors whose entries are all equal.
\begin{myitemize}%
\myitem\cmt{works no GSP}For this special case, the
algorithms in~\cite{kibangou2012graph,safavi2015nulling} produce
projections with a finite number of communication rounds. 
\myitem\cmt{works GSP}A graph signal
processing~\cite{shuman2013emerging,sandryhaila2014bigdata}
perspective to tackle this special case was adopted
in~\cite{sandryhaila2014finitetime}.
\myitem\cmt{Limitations}Unfortunately, these schemes cannot be applied
to the general subspace projection problem.
\end{myitemize}%
\end{myitemize}%
\begin{myitemize}%
\myitem\cmt{GSP}%
\begin{myitemize}%
\myitem\cmt{works}%
\begin{myitemize}%
\myitem\cmt{segarra}A more general framework is proposed
in~\cite{segarra2017graphfilterdesign}, which would allow
implementation of a projection with a graph filter if one were given
a \emph{shift matrix} such that a subset of its eigenvectors spans the
subspace of interest. \cmt{limitations}Unfortunately, this framework
does not include any method to find such a shift matrix in a general case unless the
target subspace is of dimension 1. 
\myitem\cmt{coutinho}An even more general setup is presented
in~\cite{coutino2019dgf}, which approximates an arbitrary linear
 transformation at the expense of more complex node computations
 through the notion of ``edge-variant'' graph filters.
\cmt{limitations}%
\begin{myitemize}%
\myitem\cmt{Non-convex}Unfortunately, the non-convex~\cite{boyd}
 nature of the optimization problem involved therein
\begin{myitemize}%
\myitem\cmt{guarantees}yields no guarantees that a projection filter can be
found even if it exists
\myitem\cmt{unpredictable behavior}and may lead to unpredictable
behavior if the network topology is modified and the filter weights
need to be updated. 
\myitem\cmt{complexity}
\end{myitemize}%
\myitem\cmt{higher complexity \ra Each node needs to store more coefficients}
\myitem\cmt{speed}Besides, even in the unlikely event that  the optimization algorithm finds
a global optimum, the resulting filter is not necessarily
implementable in a small number of iterations. 
\end{myitemize}%
\myitem\cmt{freq response}Further graph-filter design  schemes abound,
but they typically seek implementing a given frequency
    response \cite{isufi2017autoregressive,tay2014near,narang2012perfect}
    relative to a given shift matrix.
\end{myitemize}%
\end{myitemize}%
\end{myitemize}%

\cmt{contribution}
\begin{myitemize}%
\myitem\cmt{gap to fill}To sum up, there is no decentralized algorithm for computing
    general subspace projections in a finite number of iterations. 
\myitem\cmt{exact}The present paper fills this gap by suitably
designing a shift matrix and the graph filter coefficients. The sought
filter is of approximately minimal order, which implies that the
number of data exchanges among nodes is approximately minimized. 
\myitem\cmt{convex surrogate for the multiplicity}The minimal order is
seen to depend on the  multiplicity of the eigenvalues of the shift
matrix. Since maximizing this multiplicity would lead to a non-convex
problem, a novel convex relaxation technique is developed relying on
a nuclear norm functional of the shift matrix. 
\myitem\cmt{convexity}To solve the resulting optimization problem, a
solver based on ADMM is also developed. 
\myitem\cmt{approximate}For those scenarios where there exists no
graph filter  that can implement the desired
projection on the given topology, a method is proposed to approximate such a projection
while trading approximation error for filter order.
\myitem\cmt{solvers}Another ADMM solver is developed for this case.

\end{myitemize}%

\cmt{conference}\arev{
The conference precursor~\cite{weerasinghe2018fast} of this work
  contains~\thref{prop:disjointevals} and the key steps
  leading to \eqref{prob:exactrelaxed}. Most of the analysis,
  simulations, the approximate projection method, and the ADMM solvers
  are presented here anew. We also published a related subgradient
  method in~\cite{mollaebrahim2018large} but it is not contained in the
  present manuscript. }

\cmt{paper structure}The paper is structured as
follows. Sec.~\ref{prob} formulates the problem, reviews common
applications, and lies some background on graph
filters. Secs.~\ref{sec:exactproj} and~\ref{sec:approximate}
respectively propose methods for exact and approximate projection
implementation. Finally, Sec.~\ref{sec:sim} \arev{presents the simulations
and Sec.~\ref{sec:conclusions} summarizes the main conclusions and provides a discussion. One proof
and the derivations of the ADMM methods are provided in the supplementary
material.}

\emph{Notation}:
\begin{myitemize}%
\myitem\cmt{}Symbol $\define$ denotes equality by definition.
\myitem\cmt{Set notation}%
\begin{myitemize}
\myitem\cmt{}For sets $\mathcal{A}$ and $\mathcal{B}$, the cardinality
of $\mathcal{A}$ is denoted as $|\mathcal{A}|$ whereas
\myitem\cmt{}$\mathcal{A}\subsetneq\mathcal{B}$ indicates that
$\mathcal{A}$ is a proper subset of $\mathcal{B}$. 
\end{myitemize}
\myitem\cmt{matrix}
\begin{myitemize}
\myitem\cmt{}Boldface lowercase (uppercase) letters represent column
vectors (matrices).
\myitem\cmt{}The $\ell_n$ norm of vector $\bm
v$ is denoted as $\|\bm v\|_n$.
\myitem\cmt{}With $\bm A$ and $\bm B$ matrices of appropriate dimensions, 
$[\bm A;\bm B]$ and $[\bm A,\bm B]$ respectively denote
their vertical and horizontal concatenation, 
 \myitem\cmt{}$\bm A\transpose$ the transpose of $\bm A$,
 \myitem\cmt{}$\mathsf{cols}(\bm A)$ the set
 of the  columns of $\bm A$,
 \myitem\cmt{}$\diagnb(\bm A)$ a vector comprising the diagonal
 entries of $\bm A$,
 \myitem\cmt{}$\colspan(\bm A)$ the span of the columns of $\bm A$, 
 \myitem\cmt{}$\bm A\otimes\bm B$  the Kronecker product of $\bm A$
 and $\bm B$,
  \myitem\cmt{}$\mathsf{evals}(\bm A)$ the set of  eigenvalues of $\bm A$,
 \myitem\cmt{}$\lambda_i(\bm A)$ the $i$-th largest eigenvalue of $\bm
 A$,
 \myitem\cmt{}$\sigma_{i}(\bm A)$ the
$i$-th largest singular value of $\bm A$,
\myitem\cmt{norms}%
\begin{myitemize}
 \myitem\cmt{}$
 ||\bm A||_2 \define \sigma_1(\bm A) $ 
 the $2$-norm of $\bm A$, and
\myitem\cmt{}$||\bm A||_\star\define\sum_i\sigma_i(\bm A) $
 the nuclear norm of $\bm A$.
\end{myitemize}
\myitem\cmt{}For a subspace $\mathcal{A}$, notation
$\mathcal{A}^\perp$ represents the orthogonal complement.
\end{myitemize}
\myitem\cmt{Probability}Finally,
\begin{myitemize}
\myitem\cmt{}$\expectednb$ denotes expectation and
\myitem\cmt{}$\mathcal{N}$  the normal distribution.
\end{myitemize}
\end{myitemize}%

\section{Preliminaries} \label{prob}
\subsection{The Subspace Projection Problem}
\label{sec:projformulation}
\cmt{graph}
\begin{myitemize}%
\myitem\cmt{def}Let $\graph(\nodeset,\edgeset)$ denote a
 graph with vertex set $\nodeset=\{1,\ldots,\nodenum\}$, where  each vertex 
 corresponds to a sensor or \emph{node}, and edge set
 $\edgeset\subset \nodeset^2$.  Let there be an edge
 $(n,n')$ in $ \edgeset$ if and only if (iff) the nodes $n$ and
 $n'$ can communicate directly, e.g. through their radio
 interface. Thus, it is natural to assume (i) that  $\edgeset$ contains all
 self loops, i.e., $(n,n)\in\edgeset~\forall n\in\nodeset$, and (ii) that $\graph$ is undirected, which means
 that $(n,n')\in \edgeset$ implies that  $(n',n)\in \edgeset$.
\myitem\cmt{neighborhood}The neighborhood of the $n$-th node is defined as $\neighborhood{\nodeind}=\{n'\mid(n,n')\in\edgeset\}$. 
\end{myitemize}%

\cmt{signal}
\begin{myitemize}%
\myitem\cmt{observations}Given  $\obsvec=[\obsentry_1,\ldots,
\obsentry_\nodenum]^\top$, where $ \obsentry_n\in \mathbb{R}$ denotes the observation or measurement
acquired by the $n$-th node,
\myitem\cmt{goal}the goal is to estimate the
signal vector $\desiredvec \in \mathbb{R}^\nodenum$, which quantifies
the phenomenon of interest (e.g. temperature field).
\myitem\cmt{model}
\begin{myitemize}%
\myitem\cmt{observation model}The latter is
related to $\obsvec$ via
\begin{equation}
\label{eq:obsvecdef}
{\obsvec}=\desiredvec+{\noise},
\end{equation}
where  $\noise \in \mathbb{R}^\nodenum$ stands for additive noise.
\myitem\cmt{subspace}Vector $\desiredvec$ is known to
lie in a given subspace $\range{\uparmat}$ of dimension
$\subspacedim < \nodenum$, where the columns of   $\uparmat
\in \rfield^{\nodenum \times \subspacedim} $ are assumed orthonormal without loss of
generality (w.l.o.g.). Hence, $\desiredvec$ can be expressed as $
\desiredvec=\uparmat \coordvec $ for some
$\coordvec\in \rfield^\subspacedim$. 
\end{myitemize}%
\end{myitemize}%

\cmt{subspace projection problem}
\begin{myitemize}%
\myitem\cmt{projection}The
orthogonal projection of $\obsvec$ onto
$\range{\uparmat}$, also known as the least-squares
estimate\footnote{Also the \emph{best linear unbiased
estimator}, \emph{minimum variance unbiased estimator},
and \emph{maximum likelihood estimator} \cite{kay1} under appropriate
assumptions. } of $\desiredvec$, is given
by: \begin{equation}\label{projection}
{\hat{\desiredvec}}\define
{\uparmat\uparmat^{\top}}{{\obsvec}}\mathop{=}^{\Delta} {\projmat\obsvec}\hspace{0.5mm},
\end{equation}
\noindent where $\projmat \in \mathbb{R}^{\nodenum \times \nodenum} $ is the
projection matrix onto $\range{\uparmat}$.
\myitem\cmt{problem}The subspace projection problem is to find $ \hat\desiredvec$ given
$\obsvec$ and $\uparmat$. Vector ${\hat\desiredvec}$ is expected to be a better
estimate of $ \desiredvec$ than $\obsvec$ since the noise 
 is annihilated along $\nodenum-\subspacedim$ dimensions.

\end{myitemize}%

\subsection{The Choice of the Basis}
\label{sec:basis} 

\cmt{Overview}This section discusses specific choices of the basis
$\uset\define\{\uvec_1,\ldots, \uvec_\subspacedim\}$ formed by the
columns of $\uparmat$ in different application scenarios where a
spatial field needs to be monitored. 
\cmt{Notation}
\begin{myitemize}%
\myitem\cmt{space}To this end, let
$\sensorlocvec_\nodeind\in \rfield^\regiondim$ denote the spatial
location of the $\nodeind$-th sensor, where $\regiondim=2$ or
3. Similarly, let $\locvec\in \rfield^\regiondim$ denote an arbitrary
location in the area of interest.
\myitem\cmt{field}Suppose that the goal is to estimate a spatial field
$\signalfun:\rfield^\regiondim\rightarrow\rfield$ given the
measurements in \eqref{eq:obsvecdef}, where $\desiredvec\define
[\signalfun(\sensorlocvec_1),\ldots,\signalfun(\sensorlocvec_\nodenum)]\transpose$.
\end{myitemize}%

\cmt{Parametric model}Oftentimes, the physics of the problem directly
provides a linear parametric expansion for $\signalfun$. 
\begin{myitemize}%
\myitem\cmt{examples}
\begin{myitemize}
\myitem\cmt{diffusion}For example,
in the case of a diffusion field, such as a temperature field, one has
\begin{align}
\label{eq:diffusion}
\signalfun(\locvec) 
=&\sum_{i=1}^{\subspacedim} \frac{ \exp\left\{-  \|\locvec -
\sourcelocvec{i}\|_2^2 /(2 \sigma_i^2)\right\}}{2 \pi \sigma_i^2}
\arev{\tilde\coord_i}
\end{align}
\arev{for some coefficients $\tilde\coord_i$}, where $\sourcelocvec{i}$ is the location of the $i$-th source and
the parameters $\{\sigma_i^2\}_{i=1}^r$ are related to the diffusivity
of the medium.
\myitem\cmt{RSS}In some cases governed by a wave
equation, as occurs in wireless communications (see e.g.~\cite{romero2017spectrummaps}),
 $\signalfun$ may admit an expansion in terms of  Cauchy bells~\cite{barbarossa2009projection}: 
\begin{align}
\label{eq:cauchy}
\signalfun(\locvec) = \sum_{i=1}^{\subspacedim} \frac{1}{1 + { \|\locvec -
\sourcelocvec{i}\|_2^2}/{\sigma_i^2}}
\arev{\tilde\coord_i}.
\end{align}
\end{myitemize}%
\myitem\cmt{Find basis}\arev{To obtain $\uset$, %
evaluate \eqref{eq:diffusion} or \eqref{eq:cauchy} at the
sensor locations and collect the coefficients that multiply each
$\tilde\coord_i$ to form the vector
$\tilde \uvec_i\in \rfield^\nodenum$, $i=1,\ldots,r$. 
This yields the expansion
$\desiredvec=\sum_{i=1}^\subspacedim \tilde\uvec_i\tilde\coord_i$. Finally, 
orthonormalize  $\{\tilde \uvec_i\}_{i=1}^{r}$.}

\end{myitemize}%

\cmt{Approximations model}\arev{This approach applies when   $\signalfun$
satisfies a parametric expansion as in \eqref{eq:diffusion}
or \eqref{eq:cauchy} and this expansion is known. However, it is
often the case that the form of the expansion is known but it contains
unknown parameters, the form of the expansion is not even known,
or the field does not even admit a linear expansion but it
approximately does. In these situations, one may still pursue a linear
inference approach by capitalizing on some form of smoothness that the
target field exhibits across space.}
\begin{myitemize}%
\myitem\cmt{approximations}%
\arev{For instance, $\signalfun$} can be approximately bandlimited, which means
that $\signalfun$ can be reasonably approximated by a reduced number
$\subspacedim$ of Fourier or \emph{discrete cosine transform} (DCT) basis functions.
\begin{myitemize}%
\myitem\cmt{bandlimited}In the latter case, upon letting
$\locvec\define[\loc_1,\loc_2]\transpose$, one can write
\begin{align}
 \signalfun(\locvec) \approx& \sum_{i_1=0}^{\subspacedim_1 -1} \sum_{i_2=0}^{\subspacedim_2 -1} \coord_{i_1,i_2}  & \\
& \times \cos \left ( \frac{\pi}{\length_1} i_1 \left ( \loc_1 + \frac{1}{2} \right )   \right ) \; 
\cos \left ( \frac{\pi}{\length_2} i_2 \left ( \loc_2 + \frac{1}{2} \right )   \right ) \nonumber.
\end{align}
Here, $\length_1$ and
 $\length_2$ denote the length  along the 1st and 2nd dimensions of the region where $\signalfun$ is
 defined. The vector $\coordvec$ defined in
 Sec.~\ref{sec:projformulation} can be recovered by stacking the
 $\subspacedim_1\subspacedim_2$ coefficients $\{\coord_{i_1,i_2}\}$,
 whereas $\uset$ can be found as described earlier in this section.

\myitem\cmt{Others}Besides Fourier or DCT bases, one may pursue
 approximations based on any other collection of 
 basis functions such as conventional polynomials, discrete prolate
 spheroidal functions, and wavelets.
\end{myitemize}%
\myitem\cmt{tradeoff}Note that in any approximation of this kind there
 is a fundamental variance-bias trade-off; see
 e.g.~\cite[Ch. 3.4]{cherkassky2007}. To see this, note that the \emph{signal-to-noise ratio} after
 projection in the model
 \eqref{eq:obsvecdef} is given by
 $\|\projmat\desiredvec\|^2_2/\expectednb[\|\projmat\noise\|_2^2]$. If
 $\noise$ has zero mean and  covariance matrix
 $\expectednb[\noise \noise\transpose]=\noisepow\bm I_\nodenum$, then 
$\expectednb[\|\projmat\noise\|_2^2]=\expectednb[\|\uparmat\uparmat\transpose\noise\|_2^2]=\noisepow\trnb[\projmat]=\noisepow\subspacedim$. Thus, although a basis with a larger $\subspacedim$ may  capture more signal
 energy $\|\projmat\desiredvec\|^2_2$, the power of the noise
 component in $\hat\desiredvec$ is also
 increased.

\end{myitemize}%

\subsection{Graph Filters}
\label{sec:graphfilters}

\cmt{overview}This section briefly reviews the notion of graph
 filters~\cite{shuman2013emerging,sandryhaila2014bigdata}, which
 constitute a central part of the proposed algorithms.
\cmt{Graph signal}In this context, vector $\bm
z\define[z_1,\ldots,z_N]\transpose$ is referred to as a \emph{graph signal},
which emphasizes the fact that the entry $z_n$ is stored at the 
$n$-th node.

\cmt{graph filter}A graph filter involves two steps, as described
next.
\begin{myitemize}%
\myitem\cmt{first: local aggregations}In the first step,
\begin{myitemize}%
\myitem\cmt{overview}a finite
sequence of graph signals
$\{\obsvec\itnot{\itind}\}_{\itind=0}^{\itnum}$, where $\obsvec\itnot{\itind}\define[\obsentry_1\itnot{\itind},\ldots,
\obsentry_\nodenum\itnot{\itind}]^\top$, is collaboratively
obtained by the network through a sequence of $\itnum$ \emph{local
data exchange rounds}, or just \emph{local exchanges} for short, where
$\obsvec\itnot{0}\define\obsvec$ is the graph signal to filter.
\myitem\cmt{l-th round}At the $\itind$-th round,  each node sends
its $\obsentry\itnodenot{\itind-1}{\nodeind}$ to its neighbors and
computes a linear combination of the entries
$\{\obsentry\itnodenot{\itind-1}{\nodeind'}\}_{\nodeind'\in \neighborhood{\nodeind}}$
that it receives from them. Specifically, the next graph signal is
obtained as $\obsentry\itnodenot{\itind}{\nodeind}
= \sum_{\nodeind'\in \neighborhood{\nodeind}}\shiftentry\nodenodenot{\nodeind}{\nodeind'} \obsentry\itnodenot{\itind-1}{\nodeind'},~\nodeind=1,\ldots,\nodenum,$
where  $\shiftentry\nodenodenot{\nodeind}{\nodeind'}$ 
is the
 coefficient corresponding to the linear aggregation that takes place
 between nodes ${\nodeind}$ and ${\nodeind'}$.
 \myitem\cmt{shift operator}
\begin{myitemize}%
\myitem\cmt{matrix form}By letting $\shiftentry\nodenodenot{\nodeind}{\nodeind'}=0$
 whenever $\nodeind'\notin\neighborhood{\nodeind}$, one can
 equivalently write  $\obsentry\itnodenot{\itind}{\nodeind}
 = \sum_{\nodeind'=1}^\nodenum\shiftentry\nodenodenot{\nodeind}{\nodeind'} \obsentry\itnodenot{\itind-1}{\nodeind'}$
 or, in matrix form, $\obsvec\itnot{\itind}
 = \shiftmat\obsvec\itnot{\itind-1}$, where
 $\shiftmat\in \rfield^{\nodenum\times \nodenum}$ is given by
 $(\shiftmat)_{\nodeind,\nodeind'}=\shiftentry_{\nodeind,\nodeind'},~\nodeind,\nodeind'=1,\ldots,\nodenum$.
\myitem\cmt{term}In the graph signal processing literature, 
the matrix $\shiftmat$ is usually referred to as  \emph{shift
 matrix}~\cite{segarra2017graphfilterdesign}.
\myitem\cmt{def}More generally, an $\nodenum\times\nodenum$ matrix
 $\shiftmat$
 is said to be a shift matrix over the graph
 $\graph\define(\nodeset,\edgeset)$ if $(\shiftmat)_{\nodeind,\nodeind'}=0$
 for all $(\nodeind,\nodeind')\notin \edgeset$. The set of all possible shift
 matrices over   $\graph$
 will be denoted as $\shiftmatset$.
\myitem\cmt{examples}Examples of  matrices in  $\shiftmatset$ include the adjacency
 and Laplacian matrices of $\graph$~\cite{segarra2017graphfilterdesign}. 
  \myitem\cmt{operator}Associated with the shift matrix is the \emph{shift
  operator}, defined as the function
  $\obsvec \mapsto \shiftmat \obsvec$.
\end{myitemize}%
\myitem\cmt{explicit form}Notice that, upon recursively applying  $\obsvec\itnot{\itind}
 = \shiftmat\obsvec\itnot{\itind-1}$,  one can write\footnote{Throughout the paper,  $\bm A^0$ for a
 square matrix $\bm A$ denotes the identity matrix of the same size as
 $\bm A$, regardless of whether $\bm A$ is invertible.}
 $\obsvec\itnot{\itind}=\shiftmat^\itind \obsvec$,
 $\itind=0,\ldots, \itnum$.
\end{myitemize}%

\myitem\cmt{second: linear combination}In the second step of the graph filter,
all nodes linearly combine the iterates in the first step.
 Specifically, the following graph signal is computed:
\begin{align}
\label{eq:outputgfilter}
\projobsvec
= \sum_{\itind=0}^{\itnum}\filtercoef_\itind \obsvec\itnot{\itind}
= \sum_{\itind=0}^{\itnum}\filtercoef_\itind
 \shiftmat^{\itind}\obsvec
\end{align}
where $\filtercoef_\itind\in\rfield,~{\itind=0},\ldots,{\itnum},$ are
the so-called filter coefficients.

\myitem\cmt{Operator}%
\begin{myitemize}%
\myitem\cmt{graph filter}The operation in \eqref{eq:outputgfilter} can be
generically expressed as $\obsvec \mapsto \filtermat\obsvec$, where
\begin{equation} \label{filterH}
{\filtermat}\define\sum_{\itind=0}^{\itnum}\filtercoef_\itind{\shiftmat}^{\itind}, 
\end{equation}
and is commonly referred to as an order-$\itnum$ \emph{graph filter}.
\myitem\cmt{cayley-hamilton}An important implication of  the
Cayley-Hamilton Theorem \cite{horn} is that for any order-$\itnum$
graph filter $\filtermat$ with $\itnum\geq \nodenum$, there exists an
order-($\nodenum-1$) graph filter $\filtermat'$ with shift matrix $\shiftmat$
and coefficients $\filtercoef_\itind'$ such that
$\filtermat=\filtermat'$. This establishes an upper bound on the order
and, therefore, the number of local
exchanges required to apply a graph filter. Thus,   one can
assume w.l.o.g. that
$\itnum\leq\nodenum-1$.
\end{myitemize}%
\end{myitemize}%

\subsection{Asymptotic Decentralized Projections}
\label{sec:barbarossa}

\begin{myitemize}%
\myitem\cmt{description}A decentralized scheme for subspace projection
was proposed in~\cite{barbarossa2009projection}. There, a matrix
$\shiftmat$ is found such that (i) $\shiftmat\in \shiftmatset$ and
(ii) $\lim_{\itind\rightarrow\infty} \shiftmat^\itind
= \projmat$. Then, the nodes compute the sequence
$\{\obsvec\itnot{\itind}, \itind=0, 1,\ldots\}$, where
$\obsvec\itnot{\itind}=\shiftmat^\itind \obsvec$. This
constitutes the infinite counterpart of the first step in a graph
filter; cf.  Sec.~\ref{sec:graphfilters}. Due to (ii), it follows that
$\lim_{\itind\rightarrow\infty}  \obsvec\itnot{\itind}
= \projmat\obsvec$, as desired.
\myitem\cmt{strengths}The main strength of this method is its
simplicity, since each node just needs to store one coefficient for
each neighbor and the same operation is repeated over and
over. 
\myitem\cmt{limitations}A limitation is that 
the number of local exchanges required to attain a target error
$||\obsvec\itnot{\itind}- \projmat \obsvec||$ is generally high since this
approach only provides asymptotic convergence. Furthermore, the set of
graphs for which (i) and (ii) can be simultaneously satisfied is
considerably limited; see Sec.~\ref{sec:sim}
and \cite{camaro2013reducing}.

\end{myitemize}%

\section{Exact Projection Filters} \label{sec:exactproj}

\cmt{overview}This section proposes an algorithm to find
a graph filter that yields a subspace projection in an approximately
minimal number of iterations. To this end, Secs.~\ref{sec:minorder}
and~\ref{sec:polfeas}
formalize the problem and 
characterize the set of feasible shift matrices for a given
$\mathcal{E}$ and $\uparmat$. Subsequent sections introduce an
optimization methodology  to approximately minimize the order of the
filter, i.e. the number of communication steps needed to obtain the
projection via graph filtering.

\subsection{Minimum-order Projection Filters}
\label{sec:minorder}

\cmt{feasibility problem}
\begin{myitemize}%
\myitem\cmt{goal}To solve the subspace projection problem formulated
in Sec.~\ref{sec:projformulation} with a graph filter,  one could think of
finding a shift matrix $\shiftmat\in \shiftmatset$ and a set of
coefficients $\{\filtercoef_\itind\}_{\itind=0}^{\itnum}$ such that
$\projmat \obsvec
= \sum_{\itind=0}^{\itnum}\filtercoef_\itind \shiftmat^{\itind}\obsvec$
for all $\obsvec\in\rfield^\nodenum$ or, equivalently, such that
$\projmat
= \sum_{\itind=0}^{\itnum}\filtercoef_\itind \shiftmat^{\itind}$. Since
$\projmat$ is symmetric, it will be assumed that $\shiftmat$ is also
symmetric.
\myitem\cmt{poly feasible}To assist in this quest, consider the following
definition:
\begin{mydef}
\thlabel{def:polfeas}
Let $\dummymat\in\rfield^{\nodenum\times\nodenum}$ be an arbitrary
(not necessarily a projection) matrix.  A symmetric matrix (not
necessarily a shift matrix) $\shiftmat\in{\mathbb{R}}^{\nodenum\times{\nodenum}}$
is polynomially feasible to implement the operator
$\obsvec \mapsto\dummymat\obsvec$ if there exist $\itnum$ and
$\filtercoefvec\define[\filtercoef_0,\ldots,\filtercoef_\itnum]\transpose$
such that $\sum_{\itind=0}^{\itnum}\filtercoef_\itind\shiftmat^{\itind}=\dummymat$.
\end{mydef}
For a given $\dummymat$, the set of all polynomially feasible matrices
$\shiftmat\in \rfield^{\nodenum\times \nodenum}$  will be denoted as
$\polfeasmatset{\dummymat}$.
Except for the tip, this set is a cone since $ \shiftmat\in\polfeasmatset{\dummymat}
$ implies
$\auxconst \shiftmat\in\polfeasmatset{\dummymat} $ for all $\auxconst\neq 0$.

\myitem\cmt{coefficients straightfwd if shift known}Note that given a matrix in
$\polfeasmatset{\dummymat}$, it is straightforward to obtain 
$\filtercoefvec$
such that $\sum_{\itind=0}^{\itnum}\filtercoef_{\itind}\shiftmat^{l}=\dummymat$; see e.g.~\cite{segarra2017graphfilterdesign} and 
Sec.~\ref{sec:polfeas}.
\myitem\cmt{formulation}One may then consider the following \emph{feasibility problem}:
\begin{myitemize}%
\myitem\cmt{given}Given a graph $\graph\define(\nodeset,\edgeset)$ and a
projection matrix $\projmat$, 
\myitem\cmt{find}find a polynomially feasible shift matrix, i.e.,
find $\shiftmat \in \shiftmatset \cap \polfeasmatset{\projmat}$.
\end{myitemize}%
\myitem\cmt{feasibility}When this problem admits a solution, we say
that \emph{there exists an exact projection filter} for implementing
$\projmat$ on $\graph$. Whether this is the case depends on $\projmat$
and $\graph$. For example, if $\graph$ is too sparse, then $\projmat$
will not be computable as a graph filter. In the extreme case where
$\graph$ is fully disconnected, then the only computable projection is
$\projmat = \identity_\nodenum$. Conversely, when $\graph$ is fully
connected, then all projections can be computed as a graph
filter. Furthermore, given  that $ \shiftmatset $ is a
subspace and that $ \shiftmat\in\polfeasmatset{\dummymat} $ implies
$\auxconst \shiftmat\in\polfeasmatset{\dummymat} $ for all
$\auxconst\neq 0$, it is easy to see that such a feasibility problem
either has no solution or has infinitely many.

\end{myitemize}%
\cmt{fast filter}
\begin{myitemize}%
\myitem\cmt{informally}When feasible shifts exist, it is reasonable to
seek the $\shiftmat$ that minimizes the number of local exchanges
$\itnum$.
\myitem\cmt{formally}For arbitrary  matrices $\shiftmat$ and $\dummymat$, define $\orderfun{\dummymat}(\shiftmat)$ 
as the minimum $\itnum$ such that $\dummymat
= \sum_{\itind=0}^{\itnum}\filtercoef_\itind \shiftmat^{\itind}$ for
some $\{\filtercoef_\itind\}_{\itind=0}^\itnum$. In view of the bound
dictated by the Cayley-Hamilton Theorem (Sec.~\ref{sec:graphfilters}),   $\orderfun{\dummymat}(\shiftmat)$ can
be viewed as a function
$\orderfun{\dummymat}:\polfeasmatset{\dummymat}\rightarrow\{0,\ldots,\nodenum-1\}$.
\myitem\cmt{formulation}Given a projection matrix $\projmat$, the
problem of finding the shift matrix associated with the minimum-order
filter can therefore be  formulated as: 
\leqnomode
\begin{align}
\tag{P1}\label{prob:exact}~~\minimize_{\shiftmat}~~&\orderfun{\projmat}(\shiftmat)\\
 \st~~& \shiftmat \in \shiftmatset \cap \polfeasmatset{\projmat}.
 \nonumber
\end{align}
\reqnomode 
\end{myitemize}%
\cmt{approximate subspace proj}Conversely, when
$\shiftmatset \cap \polfeasmatset{\projmat}=\emptyset$, there exists
no graph filter capable of implementing $\projmat$. For those cases,
Sec.~\ref{sec:approximate} describes how to find a graph filter that
approximates $\projmat$.

\subsection{Polynomially Feasible Matrices}
\label{sec:polfeas}

\cmt{overview} To assist in solving \eqref{prob:exact}, this section presents an
algebraic characterization of the set $\polfeasmatset{\projmat}$ of
polynomially feasible matrices. Recall that the matrices in this set
need not be shift matrices, that is, they need not satisfy the
topology constraints determined by the edge set $\mathcal{E}$.

\cmt{Prefeasible shifts}
\begin{myitemize}%
\myitem\cmt{condition}

\begin{mylemma}\thlabel{prop:prefeasible}
Let $\uparmat\in\rfield^{\nodenum\times\subspacedim}$ with orthonormal
columns be given and let $\projmat =\uparmat\uparmat^{\top}$.  Let
also $\upermat\in\rfield^{\nodenum\times{\nodenum-\subspacedim}}$ with
orthonormal columns satisfy
${\colspan(\upermat)=\colspan^{\perp}(\uparmat)}$. If
$\shiftmat\in\polfeasmatset{\projmat}$, then there exist symmetric
matrices $ \fparmat\in\rfield^{{\subspacedim}\times{\subspacedim}}$
and
$ \fpermat\in\rfield^{\nodenum-\subspacedim\times{\nodenum-\subspacedim}}$
such that:
\begin{align}\label{eq:def:cpref}
\shiftmat=\begin {bmatrix}\uparmat& \upermat\end {bmatrix}\left[	\begin{matrix} \fparmat& \bm{0}\\ \bm{0}&\fpermat\end{matrix} \right]\begin {bmatrix}\uparmat^{\top}\\ \upermat^{\top}\end {bmatrix}.
\end{align} 
\end{mylemma}
\begin{proof}
See Appendix~\ref{proof:prefeasible}.
\end{proof}

\myitem\cmt{indep. choice basis}Note that matrices  $\fparmat$
and $\fpermat$ satisfying  \eqref{eq:def:cpref}
exist regardless\footnote{Note that the algorithms in this paper
produce the same filters for all matrices $\uparmat$ that span a given
signal subspace $\range{\uparmat}$; likewise for $\upermat$. Thus, the obtained
filters only depend on the signal subspace and not on the specific
choice of the basis. This property is what bypasses the difficulty
encountered in~\cite[eq. (20)]{segarra2017graphfilterdesign}, which
will typically be infeasible for graphs with more than $\nodenum$
missing edges.  } of the
choice of $\uparmat$ and $\upermat$ as long as the columns of $\uparmat$ and $\upermat$ 
respectively constitute an orthonormal basis for the signal subspace  $\range{\uparmat}$  and its
orthogonal complement $\colspan^{\perp}(\uparmat)$. 
\myitem\cmt{converse}As seen later, the converse of \thref{prop:prefeasible} does not hold. 

\myitem\cmt{implications: eigenstructure}To understand the implications of \thref{prop:prefeasible},
\begin{myitemize}
\myitem\cmt{orth decomposition}rewrite  \eqref{eq:def:cpref} as
\begin{align}\label{eq:def:prefparperp}
\shiftmat = \uparmat\fparmat\uparmat^{\top}
+ \upermat\fpermat\upermat^{\top} = \shiftparmat + \shiftpermat, 
\end{align}
where $\shiftparmat\define \uparmat\fparmat\uparmat\transpose$ and
$\shiftpermat \define \upermat\fpermat\upermat\transpose$ are 
symmetric matrices whose column spans are respectively contained in the signal
subspace and its orthogonal complement. Thus, they clearly satisfy
$\shiftparmat\transpose\shiftpermat=\shiftpermat\transpose\shiftparmat=\bm 0$.
\myitem\cmt{not nec. shifts}Note that even if $\shiftmat$ is a shift
matrix, i.e. $\shiftmat\in\shiftmatset$, matrices
$\shiftparmat$ and $\shiftpermat$ may not be shift matrices. 
\myitem\cmt{eigendecomposition}
Consider now the eigendecompositions
$\fparmat=\fparevecmat\evalparmat\fparevecmat\transpose,\fpermat=\fperevecmat\evalpermat\fperevecmat\transpose$ 
for orthogonal $ \fparevecmat \in \mathbb{R}^{\subspacedim \times
\subspacedim}, \fperevecmat\in \mathbb{R}^{\nodenum-\subspacedim \times \nodenum-\subspacedim}$ and diagonal $\evalparmat \in \mathbb{R}^{\subspacedim \times \subspacedim},\evalpermat \in \mathbb{R}^{\nodenum-\subspacedim \times\nodenum-\subspacedim}$.
Then, \eqref{eq:def:cpref} can be rewritten as:
\begin{align}\label{eq:def:prefevd}
\shiftmat=\begin
{bmatrix}\uparmat\fparevecmat& \upermat\fperevecmat\end
{bmatrix}
\left[	\begin{matrix} \evalparmat& \bm{0}\\ \bm{0}&\evalpermat\end{matrix} \right]
\begin {bmatrix}\fparevecmat\transpose\uparmat^{\top}\\ \fperevecmat\transpose\upermat^{\top}\end {bmatrix}.
\end{align} 
This is clearly an eigenvalue decomposition of $\shiftmat$. It further
shows that
$\mathrm{evals}(\shiftmat)=\mathrm{evals}(\fparmat)\cup\mathrm{evals}(\fpermat)$. In
view of \eqref{eq:def:prefevd},
\thref{prop:prefeasible} establishes that  any 
$\shiftmat\in \polfeasmatset{\projmat}$
has exactly $\subspacedim$ orthogonal eigenvectors (the
columns of $\uparmat\fparevecmat$) in the
signal subspace and $\nodenum-\subspacedim$ (the
columns of $\upermat\fperevecmat$) in its orthogonal complement.

\end{myitemize}%

\myitem\cmt{definition}The following definition builds
upon  \thref{prop:prefeasible} to  introduce 
a necessary condition  for feasibility of a shift matrix 
that will prove instrumental in subsequent sections. 
\begin{mydef}
\thlabel{def:prefeas}
Let $\uparmat$ and $\upermat$ be given. 
If  $\shiftmat\in\rfield^{\nodenum\times{\nodenum}}$ is such that it
satisfies \eqref{eq:def:cpref} for some symmetric
$\fparmat\in{\mathbb{R}}^{{\subspacedim}\times{\subspacedim}}$ and
$\fpermat\in{\mathbb{R}}^{\nodenum-\subspacedim\times{\nodenum-\subspacedim}}$, then $\shiftmat$
is said to be pre-feasible.
\end{mydef}
\cmt{indep choice bases}Note again that this definition is independent
of the choice of  $\uparmat$ and $\upermat$ so long as their columns
respectively form a basis for the signal subspace and its
orthogonal complement.
\myitem\cmt{notation}Given
$\projmat\in \rfield^{\nodenum\times\nodenum}$, the set of all
pre-feasible matrices in $\rfield^{\nodenum\times\nodenum}$ will be
denoted as $\prefeasmatset{\projmat}$.
\end{myitemize}%

\cmt{Feasible shifts}%
\begin{myitemize}
\myitem\cmt{feas $\subset$ prefeas}Observe that all
matrices in  $\polfeasmatset{\projmat}$ are also 
in $\prefeasmatset{\projmat}$.
\myitem\cmt{prefeas $\subsetneq$ feas}However, 
not all matrices in $\prefeasmatset{\projmat}$ are in
$\polfeasmatset{\projmat}$. Trivial examples include $\shiftmat
= \identity_\nodenum$ (recall that $\subspacedim<\nodenum$) and
$\shiftmat=\bm{0}$. The rest of this section will characterize the
matrices in $\prefeasmatset{\projmat}$ that are in
$\polfeasmatset{\projmat}$.
\myitem\cmt{which prefeasible mats are feasible}In particular, it will be seen that any pre-feasible
matrix $\shiftmat$ where $\evalparmat$ and $\evalpermat$ share at
least an
eigenvalue is not polynomially feasible.  To this end, note
from \thref{def:polfeas} and \eqref{eq:def:prefevd} that any
pre-feasible matrix must satisfy
\begin{align}
\label{eq:polfeasprefeas}
\projmat= \uparmat\fparevecmat \Big[\sum_{\itind=0}^{\itnum}\filtercoef_\itind \evalparmat^{l}\Big]\fparevecmat^{\top} \uparmat^{\top} +\upermat \fperevecmat \Big[\sum_{\itind=0}^{\itnum}\filtercoef_\itind \evalpermat^{l}\Big]\fperevecmat^{\top} \upermat^{\top}
\end{align}
for some  $\{\filtercoef_\itind\}_{\itind=0}^{\itnum}$  to be
polynomially feasible.  Multiplying both sides of \eqref{eq:polfeasprefeas} on the left by 
$\uparmat^\top$ and on the right by
$\uparmat$, it follows that $\fparevecmat\big[{\sum_{\itind=0}^{\itnum}\filtercoef_\itind\ \evalparmat^{\itind}}\big]\fparevecmat^\top = \identity_\subspacedim$ or, equivalently,
$\sum_{\itind=0}^{\itnum}\filtercoef_\itind\ \evalparmat^{\itind}=\identity_\subspacedim$. Likewise, multiplying both sides of  \eqref{eq:polfeasprefeas} on the left by $\upermat^\top$ and on the right by $\upermat$, it follows that
${\sum_{\itind=0}^{\itnum}\filtercoef_\itind\ \evalpermat^{\itind}=\bm
0}$. Arranging these two conditions in matrix form yields:
 \begin{equation} \label{eq:vander}
  \begin{bmatrix}\one_\subspacedim \\  \bm 0_{\nodenum-\subspacedim}\end{bmatrix} = 
\begin{bmatrix}
    1       & \lambda_{1}  &\dots & \lambda_{1}^{\itnum} \\
    \vdots  & \vdots      &\ddots & \vdots \\
    1       & \lambda_{\subspacedim} &\dots
    & \lambda_{\subspacedim}^{\itnum} \\
    1       & \lambda_{\subspacedim+1} &\dots & \lambda_{\subspacedim+1}^{\itnum} \\ 
    \vdots & \vdots  & \ddots & \vdots \\
    1       & \lambda_{\nodenum}  &\dots & \lambda_{\nodenum}^{\itnum}
\end{bmatrix}
 \begin{bmatrix}
  \filtercoef_0 \\ \filtercoef_1\\ \vdots\\ c_{\itnum}
 \end{bmatrix},
 \end{equation}
\noindent where $\lambda_1,\ldots,\lambda_\nodenum $ are such that
  $\evalparmat\triangleq \text{diag}\{\lambda_1,\ldots,\lambda_\subspacedim \}$
  and
  $\evalpermat\hspace{-2mm}\triangleq \text{diag}\{\lambda_{\subspacedim+1},\ldots,\lambda_\nodenum \}.$
  Vandermonde systems such as  \eqref{eq:vander}  frequently arise when designing graph
  filters; see
  e.g.~\cite{segarra2017graphfilterdesign,sandryhaila2014finitetime}. With
  the appropriate definitions, it can be expressed in matrix
  form as:
\begin{equation} \label{eq:vander2}
\boldsymbol \lambda_\projmat = \vandermondemat \filtercoefvec ,
\end{equation}
\noindent which provides a means to obtain the coefficients
  $\{\filtercoef_\itind\}_{\itind=0}^{\itnum}$ when \eqref{eq:vander2} admits a solution. 
  
\begin{myitemize}%
\myitem\cmt{}To understand when the latter is the case, assume w.l.o.g. that $\itnum=\nodenum-1$ since the existence
of a solution to \eqref{eq:vander} for some $\itnum$ implies its
existence for $\itnum=\nodenum-1$.
\myitem\cmt{}Since $\vandermondemat$ is a square Vandermonde matrix, any two rows
corresponding to distinct eigenvalues are linearly independent. 
\myitem\cmt{}Looking at the left-hand side of \eqref{eq:vander}, it is easy to
see that the system \eqref{eq:vander} admits a solution iff
$\evalparmat$ and $\evalpermat$ do not share eigenvalues.
\end{myitemize}%

\myitem\cmt{statement}This conclusion can be combined
with \thref{prop:prefeasible} as follows:
\begin{mytheorem}
\thlabel{prop:disjointevals}
Let $\shiftmat\in\rfield^{\nodenum\times{\nodenum}}$ be
symmetric. Then, $\shiftmat\in\polfeasmatset{\projmat}$ iff both
the following two conditions hold:
\leqnomode
\begin{align}
\tag{C1}&\shiftmat\in\prefeasmatset{\projmat}\text{, i.e., it
satisfies \eqref{eq:def:cpref}} \\&\text{ for some symmetric  }\fpermat\text{ and
  }\fparmat,\nonumber\\ \tag{C2} & \mathrm{evals}(\fparmat)\cap\mathrm{evals}(\fpermat)=\emptyset.
\end{align}
\reqnomode

\end{mytheorem}

\end{myitemize}%

\subsection{Filter Order Minimization}
\label{sec:filterorderminimization}

\cmt{overview}\thref{prop:disjointevals} implies
  that \eqref{prob:exact} can be reformulated as the minimization of
$\orderfun{\projmat}(\shiftmat)$ subject to
$\shiftmat \in \shiftmatset$, (C1), and (C2). This section and the
next develop a reformulation more amenable to application of a
numerical solver.

\cmt{objective}The first step is to express the objective function
$\orderfun{\projmat}(\shiftmat)$  more explicitly. To that end,
consider the following result:
\begin{myitemize}%
\myitem\cmt{recall def $\orderfun{\projmat}$}
\myitem\cmt{Order = \# distinct evals}
\begin{mylemma}
\thlabel{prop:minimalorder}If $\shiftmat = \uparmat\fparmat\uparmat\transpose
+ \upermat\fpermat\upermat\transpose\in \polfeasmatset{\projmat}$, then
\begin{align}
\label{eq:minimalorder}
\orderfun{\projmat}(\shiftmat) \leq \noevalsfun(\fparmat)
+\noevalsfun(\fpermat) -1,
\end{align}
where $\noevalsfun(\cdot)$ is the number of distinct eigenvalues of
its argument.
\end{mylemma}

\begin{IEEEproof}
Computing $\orderfun{\projmat}(\shiftmat)$ for a given $\shiftmat$
amounts to determining the minimum $\itnum$ for
which \eqref{eq:vander}, or equivalently its compact
version \eqref{eq:vander2}, admits a solution.  Since
$\shiftmat \in \polfeasmatset{\projmat}$ by hypothesis, one has
that \eqref{eq:vander2} is satisfied for at least one value of
$\itnum$. Let $\itnum_0$ denote the smallest value for
which \eqref{eq:vander2} holds. Since $\vandermondemat$ is
Vandermonde, it has at most $\noevalsfun(\shiftmat)$ linearly
independent rows, which implies that 
$\itnum_0+1\leq \noevalsfun(\shiftmat)$. The proof is completed by noting
from \thref{prop:disjointevals} that $\noevalsfun(\shiftmat)
= \noevalsfun(\fparmat)
+\noevalsfun(\fpermat)$.
\end{IEEEproof}
\myitem\cmt{tightness}In practice, the bound in \eqref{eq:minimalorder}
will hold with equality unless in degenerate cases. For example, when
$\nodenum=4$, $\subspacedim=2$, $\lambda_1=-\lambda_2$ and
$\lambda_3=-\lambda_4\neq |\lambda_1|$, it can be seen that
$\orderfun{\projmat}(\shiftmat)=2$ whereas $ \noevalsfun(\fparmat)
+\noevalsfun(\fpermat)-1=3$. However, in general, such an $\shiftmat$
will only be in $\shiftmatset \cap \polfeasmatset{\projmat}$ if
$\projmat$ and $\graph$ are jointly selected to achieve this end,
which will not occur in a practical application.
\myitem\cmt{interpretation}In words, 
\thref{prop:minimalorder} implies that one may seek the shift matrix
of an approximately minimal-order projection filter as the matrix with
the smallest number of distinct eigenvalues among all matrices in the
feasible set of \eqref{prob:exact}.
\end{myitemize}%

\cmt{Feasible set}Next, this feasible set is
rewritten more explicitly.
\begin{myitemize}%
\myitem\cmt{split into two constraints}To this end, split the constraint
$\shiftmat \in \shiftmatset \cap \polfeasmatset{\projmat}$ 
into the two constraints
$\shiftmat \in \shiftmatset$ and
$\shiftmat \in   \polfeasmatset{\projmat}$.
\begin{myitemize}%
\myitem\cmt{topology constraint}Regarding $\shiftmat \in \shiftmatset$,
\begin{myitemize}%
\myitem\cmt{def shift}recall from the definition of $\shiftmatset$
in Sec.~\ref{sec:graphfilters} that
$\shiftmat \in \shiftmatset$ iff  $(\shiftmat)_{\nodeind,\nodeind'}=0$
for all $(\nodeind,\nodeind') $ such that
 $(\nodeind,\nodeind')\notin \edgeset$.
 \myitem\cmt{symmetry}Because any 
 $\shiftmat$ that is feasible for  \eqref{prob:exact} has to be 
 in $\polfeasmatset{\projmat}$ and all the 
 matrices in this set are symmetric, any feasible $\shiftmat$ is necessarily
 symmetric. Thus, one can just require that  $(\shiftmat)_{\nodeind,\nodeind'}=0$
 only  for those $(\nodeind,\nodeind') $ such that
 $(\nodeind,\nodeind')\notin \edgeset$  and 
 $\nodeind<\nodeind'$.
 \myitem\cmt{canonicalbasis}With $\canbasisvec_\nodeind$ the
 $\nodeind$-th column of the identity matrix $\identity_\nodenum$, it
 follows that   $(\shiftmat)_{\nodeind,\nodeind'}=
\canbasisvec_\nodeind\transpose\shiftmat\canbasisvec_{\nodeind'}=
(\canbasisvec_{\nodeind'}\otimes\canbasisvec_\nodeind)\transpose\vect(\shiftmat)
$. Thus, the constraint $\shiftmat \in \shiftmatset$ can be expressed
as $\topologyconstraintmat\vect(\shiftmat) = \bm 0$, where
$\topologyconstraintmat$ is a matrix whose rows are given by the
vectors
$\{(\canbasisvec_{\nodeind'}\otimes\canbasisvec_\nodeind)\transpose, 
\forall(\nodeind,\nodeind')$ such that
$(\nodeind,\nodeind')\notin \edgeset$ and
 $\nodeind<\nodeind'\}$. As expected, the fewer edges in the graph,
 the more rows $\topologyconstraintmat$ has and, consequently, the
 smaller the feasible set. In the extreme case of a fully disconnected
 graph, only the diagonal matrices satisfy
 $\shiftmat \in \shiftmatset$ (recall that  $\edgeset$ contains all self-loops;
 cf. Sec.~\ref{sec:projformulation}). 
 
\end{myitemize}%

\myitem\cmt{polfeas constraint}On the other hand, the constraint
$\shiftmat \in   \polfeasmatset{\projmat}$ can be easily expressed invoking
\thref{prop:disjointevals} and introducing two auxiliary optimization
variables $\fparmat$ and $\fpermat$. 

\end{myitemize}%

\myitem\cmt{resulting problem}In view of these observations and \thref{prop:minimalorder},
problem \eqref{prob:exact} becomes 
\leqnomode
\begin{align}
\tag{P1'}\label{prob:exactwithconstraints}
\begin{aligned}
~~\minimize_{\shiftmat,\fparmat,\fpermat}~~&
 \noevalsfun(\fparmat)
+\noevalsfun(\fpermat)\\
 \st~~&  \topologyconstraintmat\vect(\shiftmat) = \bm 0\\
& \shiftmat = \uparmat\fparmat\uparmat\transpose
+ \upermat\fpermat\upermat\transpose\\
&\fparmat=\fparmat\transpose,~\fpermat=\fpermat\transpose\\
&\lambda_\nodeind(\fparmat)\neq\lambda_{\nodeind'}(\fpermat)~\forall\nodeind,\nodeind'
\end{aligned}
\end{align}
\reqnomode 
for an arbitrary choice of
$\upermat\in \rfield^{\nodenum\times \nodenum-\subspacedim}$ with
orthonormal columns spanning $\colspan^\perp(\uparmat)$.
\end{myitemize}%

\cmt{two modifications}Two further modifications are in order.
\begin{myitemize}%
\myitem\cmt{invariance to scalings}First, note that
\eqref{prob:exactwithconstraints} is invariant to scalings in the
sense that if $(\shiftmat,\fparmat,\fpermat)$ is feasible, then
 $(\auxconst\shiftmat,\auxconst\fparmat,\auxconst\fpermat)$ is also
 feasible and attains the same objective value $\forall \auxconst\neq
 0$. For this reason, the constraint $\trnb(\fparmat) = \subspacedim$ will be
 introduced w.l.o.g.  to eliminate this ambiguity.

\myitem\cmt{feasible set is not closed}Second, the feasible set
 of \eqref{prob:exactwithconstraints} is not a closed set due to
 the constraint
 $\lambda_\nodeind(\fparmat)\neq\lambda_{\nodeind'}(\fpermat)~\forall\nodeind,\nodeind'$,
 implying that the optimum may not be attained through an iterative algorithm. In practice, this
 constraint holds so long as the eigenvalues of $\fparmat$ differ from
 those of $\fpermat$, even if there is an eigenvalue of $\fparmat$
 arbitrarily close to an eigenvalue of $\fpermat$. But in the latter case,
 the numerical conditioning of \eqref{eq:vander} would be poor, implying
 that the target projection cannot be implemented as a graph filter
 using finite-precision arithmetic. Thus, the aforementioned
 constraint should be replaced with another one that ensures that (i)
 the feasible set is closed, and (ii) the eigenvalues of $\fparmat$
 are sufficiently different from those of $\fpermat$. One natural possibility
 is
 $|\lambda_\nodeind(\fparmat)-\lambda_{\nodeind'}(\fpermat)|\geq \eigenreg~\forall\nodeind,\nodeind'$,
 where $\eigenreg>0$ is a user-selected parameter. Unfortunately, this
 constraint is not convex, but an effective  relaxation will be
 presented in the next section.

\end{myitemize}%

\cmt{resulting problem} To sum up, the optimization problem 
 to be solved is:
\leqnomode
\begin{align}
\tag{P1''}\label{prob:exactclosedfs}
\begin{aligned}
~~\minimize_{\shiftmat,\fparmat,\fpermat}~~&
 \noevalsfun(\fparmat)
+\noevalsfun(\fpermat)\\
 \st~~&  \topologyconstraintmat\vect(\shiftmat) = \bm 0,~\trnb(\fparmat) = \subspacedim\\
& \shiftmat = \uparmat\fparmat\uparmat\transpose
+ \upermat\fpermat\upermat\transpose\\
&\fparmat=\fparmat\transpose,~\fpermat=\fpermat\transpose\\
&|\lambda_\nodeind(\fparmat)-\lambda_{\nodeind'}(\fpermat)|\geq \eigenreg~\forall\nodeind,\nodeind'.
\end{aligned}
\end{align}
\reqnomode

\cmt{relation to barbarossa's feasible set}Remarkably, the set of
 topologies for which \eqref{prob:exactclosedfs} is feasible is
strictly larger than the set of topologies for which the
method proposed in~\cite{barbarossa2009projection} (and reviewed in
Sec.~\ref{sec:barbarossa}) can be applied. It can be easily seen that
any feasible matrix for the problem in~\cite{barbarossa2009projection}
must be in $ \shiftmatset \cap \polfeasmatset{\projmat}$ and,
additionally, $\nodenum-\subspacedim$ of its eigenvalues must be equal
to 1 whereas the rest must be less than 1. It follows that the
feasible set therein is strictly contained in the feasible set of \eqref{prob:exactwithconstraints}.

\subsection{Exact Projection Filters via Convex Relaxation}
\label{sec:convex}

\cmt{overview}Problem \eqref{prob:exactclosedfs}  is non-convex because
\begin{myitemize}%
\myitem\cmt{}(i) the objective function is not convex and
\myitem\cmt{}(ii) the last constraint, which does not define a convex set.
\end{myitemize}%
This section proposes a convex problem to approximate the solution
to \eqref{prob:exactclosedfs}.

\cmt{objective}To address (i), recall  that $ \noevalsfun(\fparmat)$
equals the number of distinct eigenvalues of $\fparmat$ and 
\begin{myitemize}%
\myitem\cmt{transforming objective}note that the larger  $ \noevalsfun(\fparmat)$, the
larger the number of non-zero elements of the vector
$[\lambda_1-\lambda_2,\lambda_1-\lambda_3,\ldots,\lambda_1-\lambda_r,\lambda_2-\lambda_3,\ldots,\lambda_{r-1}-\lambda_r]\transpose$. A
similar observation applies to $ \noevalsfun(\fpermat)$.  This
suggests replacing the objective in \eqref{prob:exactclosedfs} with
\begin{align}
\label{eq:surrobj}
 \sum_{\nodeind=1}^\subspacedim
\sum_{\nodeind'=1}^\subspacedim||\lambda_\nodeind(\fparmat)-\lambda_{\nodeind'}(\fparmat)||_0
+\\
\sum_{\nodeind=1}^{\nodenum-\subspacedim}
\sum_{\nodeind'=1}^{\nodenum-\subspacedim}||\lambda_\nodeind(\fpermat)-\lambda_{\nodeind'}(\fpermat)||_0\nonumber,
\end{align}
where $\|\auxvec\|_0$ is the so-called \emph{zero norm} or number of
non-zero elements of vector $\auxvec$.
\myitem\cmt{obj. surrogate}A
typical convex surrogate of the zero norm is the
$\itind_1$-norm \cite{candes2008introduction}. However, just replacing
$||\cdot ||_0$ in \eqref{eq:surrobj} with an $l_1$-norm would still
give rise to a non-convex objective since it involves the functions
$\lambda_\nodeind(\cdot)$. 
\myitem\cmt{Kronecker}One of the key ideas in this paper is to apply
the following result:
\begin{mylemma}
\thlabel{prop:evalsnucleardif}
Let $\linfunmat$ be an $\nodenum\times \nodenum$ matrix with eigenvalues $\lambda_1, \lambda_2, \ldots, \lambda_\nodenum$. Then,
 \begin{align}
 {{\left\| \linfunmat\otimes {\identity_{\nodenum}}-{\identity_{\nodenum}}\otimes \linfunmat\right\|}_{*}}=\sum\limits_{\nodeind=1}^{\nodenum}{\sum\limits_{\nodeind'=1}^{\nodenum}{\left| {{\lambda }_{\nodeind}}-{{\lambda }_{\nodeind'}} \right|}}.
 \end{align}
\end{mylemma}
\begin{IEEEproof}
See Appendix~\ref{proof:evalsnucleardif}.
\end{IEEEproof}
Applying \thref{prop:evalsnucleardif} to \eqref{eq:surrobj} suggests
replacing the objective in \eqref{prob:exactclosedfs} with $ ||\fparmat\otimes \identity_{\subspacedim}
-\identity_{\subspacedim} \otimes \fparmat ||_\star + 
||\fpermat\otimes \identity_{\nodenum-\subspacedim}-\identity_{\nodenum-\subspacedim} \otimes
\fpermat||_\star$. \arev{The implications of this relaxation are
further discussed in Appendix~\ref{sec:tightness}.}


\end{myitemize}%

\cmt{feasible set}Regarding (ii), note that 
\begin{myitemize}%
\myitem\cmt{relax feasible set}constraint
$|\lambda_\nodeind(\fparmat)-\lambda_{\nodeind'}(\fpermat)|\geq \eigenreg~\forall\nodeind,\nodeind'$
renders the feasible set non-convex due to (ii-1) the functions
$\lambda_\nodeind(\cdot)$ and (ii-2) the absolute value.
\begin{myitemize}%
\myitem\cmt{relax lambdas}To deal with (ii-1), a sensible approach is to relax the constraints $|\lambda_\nodeind(\fparmat)-\lambda_{\nodeind'}(\fpermat)|\geq \eigenreg~\forall\nodeind,\nodeind'$
into a single constraint which requires only that the average of the
eigenvalues of $\fparmat$ must differ sufficiently from the average of
the eigenvalues of $\fpermat$, that is
\begin{align}
\left|\frac{1}{\subspacedim}\sum_{\nodeind=1}^\subspacedim\lambda_\nodeind(\fparmat)-
\frac{1}{\nodenum-\subspacedim}\sum_{\nodeind=1}^{\nodenum-\subspacedim}
\lambda_{\nodeind}(\fpermat)\right|\geq \eigenreg\nonumber
\end{align}
for some user-selected $\eigenreg>0$. Clearly, this constraint is
equivalent to    $\left|\trnb(\fparmat)/{\subspacedim}-
\trnb(\fpermat)/({\nodenum-\subspacedim})\right|\geq \eigenreg$
and, using the constraint $\trnb(\fpermat)=\subspacedim$ introduced
in \eqref{prob:exactclosedfs},  becomes equivalent to 
  $\left|1-
\trnb(\fpermat)/({\nodenum-\subspacedim})\right|\geq \eigenreg$.
\myitem\cmt{absolute value}To deal with (ii-2), one can now use the definition of absolute
value to conclude that this constraint is satisfied if either
$\trnb(\fpermat) \leq (1-\eigenreg)({\nodenum-\subspacedim})$ or
$
\trnb(\fpermat)\geq (1+\eigenreg)({\nodenum-\subspacedim})$. Both of
these inequalities are affine and therefore convex, however the OR
condition renders the feasible set 
non-convex. However, a solution can be readily found by 
first solving the problem with the constraint $\trnb(\fpermat) \leq
(1-\eigenreg)({\nodenum-\subspacedim})$, then solving it 
with the constraint
$\trnb(\fpermat)\geq (1+\eigenreg)({\nodenum-\subspacedim})$ instead, and
finally comparing the objective values achieved at the optimum of both these
sub-problems.
\myitem\cmt{ineq \ra eq}To simplify the task of solving the resulting
problem, note that the solution in both subproblems will
satisfy the corresponding constraint with equality. To see this, note
that if the inequality constraint is removed, the optimum of the
resulting convex problem becomes
$(\shiftmat,\fparmat,\fpermat)=(\identity_\nodenum,\identity_\subspacedim,
\identity_{\nodenum-\subspacedim})$. Due to convexity and since this optimum does not satisfy the removed constraint, such a
constraint will become necessarily active when introduced.

\myitem\cmt{plus minus}In view of these observations, the last
constraint in \eqref{prob:exactclosedfs} will be replaced with  $ 
\trnb(\fpermat)= (1\pm\eigenreg)({\nodenum-\subspacedim})$, where the
$\pm$ sign indicates that both subproblems must be solved separately.
\myitem\cmt{Regularizer}Except for degenerate cases, one expects that
these relaxed constraints will suffice to ensure that $\fpermat$ and
$\fparmat$ do not share eigenvalues. However, although this was indeed
observed in our experiments, sometimes an eigenvalue of $\fpermat$ may
become sufficiently close to an eigenvalue of $\fparmat$ in such a way
that the Vandermonde system \eqref{eq:vander} becomes poorly
conditioned. To alleviate this situation, a penalty proportional to
$||\fpermat ||_\text{F}^{2}$ will be added to the objective to
``push'' the eigenvalues of $\fpermat$ towards zero.

\end{myitemize}%

\cmt{Convex problem}With all these considerations, the resulting relaxed problem becomes:
\leqnomode
\begin{align}
\nonumber~~\minimize_{\shiftmat,\fparmat,\fpermat}~~&
 \regnucpar||\fparmat\otimes \identity_\subspacedim
-\identity_\subspacedim \otimes \fparmat ||_\star \\&+ \regnucper
||\fpermat\otimes \identity_{\nodenum-\subspacedim}-\identity_{\nodenum-\subspacedim}\otimes
\fpermat||_\star+||\fpermat ||_\text{F}^{2}\nonumber\\
 \st~~&  \topologyconstraintmat\vect(\shiftmat) = \bm 0,\nonumber\\
\tag{P1-R}\label{prob:exactrelaxed}& \shiftmat = \uparmat\fparmat\uparmat\transpose
+ \upermat\fpermat\upermat\transpose,\\
&\fparmat=\fparmat\transpose,~\fpermat=\fpermat\transpose\nonumber,\\
&\trnb(\fpermat)= (1\pm\eigenreg)({\nodenum-\subspacedim}),
\trnb(\fparmat) = \subspacedim\nonumber,
\end{align}
\reqnomode 
where $\regnucpar>0$ and $\regnucper>0$ are user-specified parameters
to control the relative weight of each term in the objective. Having
two parameters rather than just one that scales the term $||\fpermat
||_\text{F}^{2}$ empowers  the user with more flexibility to alleviate
possible numerical issues, as described earlier.
\cmt{relaxation}\arev{A study of the ability
of \eqref{prob:exactrelaxed} to approximate the solution
of \eqref{prob:exact} is presented in Appendix~\ref{sec:tightness}.
}

\cmt{solver}An iterative solver for \eqref{prob:exactrelaxed} is
proposed in Appendix~\ref{sec:exactsolver} in the supplementary
 material based on ADMM, therefore inheriting its solid convergence
 guarantees.

\end{myitemize}%

\section{Approximate Projection Filters}
\label{sec:approximate}

\cmt{Overview}When the graph is too sparse for \eqref{prob:exact} to
admit a solution, one may  instead seek  a low-order graph filter that
approximates $\projmat$ reasonably well. Before presenting an
optimization criterion to obtain the corresponding shift matrix, the next
section analyzes in which situations there exists an exact projection
filter.

\subsection{Feasibility of Exact Projection Filters}
\label{sec:feasibility}

\cmt{Relate topology an $\subspace$ to existence}The following result provides a necessary condition for the existence
of an exact projection shift matrix. 
\begin{myitemize}%
\myitem\cmt{reduced edge set}
Let $\reducededgeset\define\{
(\nodeind,\nodeind')\in\edgeset:~\nodeind\leq\nodeind'
\}=\{(\nodeind_1,\nodeind'_1),\ldots,(\nodeind_\reducededgenum,\nodeind'_\reducededgenum)\}$
denote the reduced edge set, where each of the $\reducededgenum$
undirected  edges shows up only once.
\myitem\cmt{single edge shift}Consider also the symmetric shift matrices
$\singleedgeshiftmat_i\define \canbasisvec_{n_{i}}\canbasisvec_{n_{i}'}^{\top}+\canbasisvec_{n_{i}'}\canbasisvec_{n_{i}},~i=1,\ldots,\reducededgenum$,  where only a single
edge is used.
\myitem\cmt{basis for $\shiftmatset$}Clearly, all feasible shift matrices are
linear combinations of these matrices.
\end{myitemize}%
\begin{mytheorem}
\thlabel{prop:necessaryfeasibility}
Let
$\singleedgeshiftmat \define[\vect(\singleedgeshiftmat_1),\ldots,\vect(\singleedgeshiftmat_\reducededgenum)]$. It
holds that
\begin{align}
\label{eq:necessaryfeasibility}
\dim(\shiftmatset\cap \prefeasmatset{\projmat})
=\reducededgenum -
 \mathrm{rank}((
\uparmat^{\top}\otimes\upermat^{\top})\singleedgeshiftmat ).
\end{align}

\end{mytheorem}
\begin{IEEEproof}
See Appendix~\ref{proof:necessaryfeasibility}. 
\end{IEEEproof}

\cmt{Implications}
\begin{myitemize}%
\myitem\cmt{rank condition}Due to the presence of the scaled identity matrices
 in $\shiftmatset\cap \prefeasmatset{\projmat}$ (see Sec.~\ref{sec:polfeas}), 
a necessary condition for a feasible projection filter to exist,
i.e. 
$\shiftmatset \cap \polfeasmatset{\projmat}\neq \emptyset$, is that
$\dim(\shiftmatset\cap \prefeasmatset{\projmat})>1$. From \eqref{eq:necessaryfeasibility},
this condition becomes $ \mathrm{rank}((
\uparmat^{\top}\otimes\upermat^{\top})\singleedgeshiftmat )\leq
\reducededgenum-2$.
For future reference, this is summarized as follows:
\begin{mycorollary}
\thlabel{prop:necessaryfullfeas}
If there exists an exact projection filter, i.e. $\shiftmatset \cap \polfeasmatset{\projmat}\neq \emptyset$, 
then
\begin{align}
\label{eq:necessaryfeasibility2}
 \mathrm{rank}((
\uparmat^{\top}\otimes\upermat^{\top})\singleedgeshiftmat )\leq
\reducededgenum-2.
\end{align}
\end{mycorollary}
Note that this provides a condition that can be easily checked before
attempting to solve \eqref{prob:exactrelaxed}.
\myitem\cmt{guiding relation}It also provides a guideline
on the minimum number of edges required to exactly implement a
projection with a given $\subspacedim$.  Since
$(\uparmat^{\top}\otimes\upermat^{\top})\singleedgeshiftmat \in{\mathbb{R}}^{\subspacedim(\nodenum-\subspacedim)\times\reducededgenum}$,
it follows that $\mathrm{rank}((
\uparmat^{\top}\otimes\upermat^{\top})\singleedgeshiftmat)\leq
\min(\subspacedim(\nodenum-\subspacedim),\reducededgenum)
$. To obtain a reference for the  number of
edges required for the existence of exact projection filters, suppose that $(
\uparmat^{\top}\otimes\upermat^{\top})\singleedgeshiftmat $ is full rank,
as often is the case. Then, \eqref{eq:necessaryfeasibility2}
becomes
$\min(\subspacedim(\nodenum-\subspacedim),\reducededgenum)
\leq
\reducededgenum-2$, which is equivalent to 
$\subspacedim(\nodenum-\subspacedim)\leq
\reducededgenum-2$. Thus, the number of edges required for an exact
projection filter to exist is in the order of
$\subspacedim(\nodenum-\subspacedim)$. 
This agrees with intuition because the difficulty to implement a
projection must depend equally on $\subspacedim$ and
$\nodenum-\subspacedim$. This follows by noting that (i)
due to the term
$\filtercoef_0\bm I_\nodenum$ in \eqref{filterH},
 implementing $\projmat$
with a graph filter is
equally difficult as implementing $\bm I_{\nodenum}-\projmat$; and
(ii)  $\projmat$ has rank  $\subspacedim$ whereas $\bm
I_{\nodenum}-\projmat$ has rank
$\nodenum-\subspacedim$. \arev{In any case, note that this is just a
guideline on the number of edges. The existence of an exact projection
filter will not only depend on the number of edges but also on their
 locations.}

\cmt{when POLF is empty}Note that \thref{prop:necessaryfullfeas}
provides a necessary condition for the existence of an exact
projection filter. Obtaining a sufficient condition, in turn, is much
more challenging. However, one expects that 
$\shiftmatset\cap\polfeasmatset{\projmat}\neq\emptyset$ whenever
$\dim(\shiftmatset\cap \prefeasmatset{\projmat})>1$. In fact we
conjecture that, given a topology and generating $\uparmat$ by
orthonormalizing $\subspacedim$ vectors in $\rfield^\nodenum$ drawn
i.i.d. from a continuous distribution, the event where
$\dim(\shiftmatset\cap \prefeasmatset{\projmat})>1$ and
$\shiftmatset\cap\polfeasmatset{\projmat}=\emptyset$ has zero
probability. In practice, however, whether
$\shiftmatset\cap\polfeasmatset{\projmat}$ is empty or not is not as
relevant as it may seem.
\cmt{not as important}To see this, remember
from \thref{prop:disjointevals} that
$\shiftmatset\cap\polfeasmatset{\projmat}\neq \emptyset$ when there
exists a pre-feasible shift matrix for which the eigenvalues of $\fparmat$ differ
from those of $\fpermat$. However, even when this is the case, if one
of the eigenvalues of $\fparmat$ lies too close to an eigenvalue of
$\fpermat$, the corresponding filter cannot be  implemented  due to poor conditioning
of \eqref{eq:vander}. In short, even when an exact projection filter
may exist from a theoretical perspective, such a filter may not be
implementable in practice. Instead, having a sufficiently large
$\dim(\shiftmatset\cap \polfeasmatset{\projmat})$ seems more important
since that would allow the user to choose a shift that
leads to a good conditioning of \eqref{eq:vander}. In this sense,
\thref{prop:necessaryfeasibility} suggests that the \emph{margin}
$\reducededgenum-
 \mathrm{rank}((
\uparmat^{\top}\otimes\upermat^{\top})\singleedgeshiftmat )$ would be
 a reasonable indicator of how easy it  is to obtain an exact
 projection filter.

\end{myitemize}%

\subsection{Approximate Projection Criterion}
\label{sec:approximatecriterion}

\cmt{overview}The method proposed in Sec.~\ref{sec:convex} to obtain an exact
 projection filter relies on solving \eqref{prob:exactrelaxed}. Using
 this problem as a starting point, the present section develops an
 optimization criterion that yields low-order graph
 filters that approximate a given projection operator.

\cmt{rewrite \eqref{prob:exactrelaxed} in terms of $\shiftmat$}
\begin{myitemize}%
\myitem\cmt{change vars}To this end,  note from the second
 constraint in \eqref{prob:exactrelaxed} and the orthogonality of
 $\uparmat$ and $\upermat$ that 
$\fparmat= \uparmat^{\top}\shiftmat\uparmat$ and
$\fpermat= \upermat^{\top}\shiftmat\upermat$.
\myitem\cmt{symmetry constraints}It is also easy to see that 
  $\shiftmat$ is symmetric iff $\fparmat$ and $\fpermat$ are
  symmetric\arev{; cf. Appendix~\ref{proof:prefeasible}}.
\myitem\cmt{separate eigenspaces}Additionally, one can also easily
  prove that any  $\shiftmat$ can be expressed as $\shiftmat = \uparmat\fparmat\uparmat\transpose
+ \upermat\fpermat\upermat\transpose$ for some $\fparmat$ and
$\fpermat$ iff $\upermat\transpose\shiftmat\uparmat = \bm 0$. In
words, the latter condition states that each eigenvector of
$\shiftmat$ must be either in the signal subspace or in its orthogonal
complement.

\myitem\cmt{problem}Then, problem \eqref{prob:exactrelaxed} can be equivalently
expressed in terms of $\shiftmat$ as
\leqnomode
\begin{align}
~~\minimize_{\shiftmat}~~&
 \regnucpar||\uparmat\transpose\shiftmat\uparmat\otimes \identity_\subspacedim
-\identity_\subspacedim \otimes \uparmat\transpose\shiftmat\uparmat ||_\star +\nonumber\\& \regnucper
||\upermat\transpose\shiftmat\upermat\otimes \identity_{\nodenum-\subspacedim}-\identity_{\nodenum-\subspacedim} \otimes
\upermat\transpose\shiftmat\upermat||_\star+\nonumber\\&||\upermat\transpose\shiftmat\upermat ||_\text{F}^{2}
\tag{P1-R'}\label{prob:exactrelaxedS}
\\\nonumber
 \st~~
&\topologyconstraintmat\vect(\shiftmat) = \bm 0,~\shiftmat = \shiftmat\transpose,~\upermat\transpose\shiftmat\uparmat = \bm 0,\\\nonumber
&\trnb(\uparmat\transpose\shiftmat\uparmat) = \subspacedim, \\&\trnb(\upermat\transpose\shiftmat\upermat)= (1\pm\eigenreg)({\nodenum-\subspacedim}).\nonumber
\end{align}
\reqnomode
\end{myitemize}%

\cmt{now relax}The constraint that renders \eqref{prob:exactrelaxedS}
infeasible when an exact projection filter does not exist is precisely
$\upermat\transpose\shiftmat\uparmat = \bm 0$. This suggests finding
a shift for an approximate projection filter by solving
\leqnomode
\begin{align}
~~\minimize_{\shiftmat}~~&
 \regnucpar||\uparmat\transpose\shiftmat\uparmat\otimes \identity_\subspacedim
-\identity_\subspacedim \otimes \uparmat\transpose\shiftmat\uparmat ||_\star +\nonumber\\& \regnucper
||\upermat\transpose\shiftmat\upermat\otimes \identity_{\nodenum-\subspacedim}-\identity_{\nodenum-\subspacedim} \otimes
\upermat\transpose\shiftmat\upermat||_\star+\nonumber\\&||\upermat\transpose\shiftmat\upermat ||_\text{F}^{2}+\regpar|| \upermat\transpose\shiftmat\uparmat||_\text{F}^2
\tag{P2-R}\label{prob:approxrelaxed}
\\\nonumber
 \st~~
&\topologyconstraintmat\vect(\shiftmat) = \bm 0,~\shiftmat = \shiftmat\transpose,\\\nonumber
&\trnb(\uparmat\transpose\shiftmat\uparmat) = \subspacedim, \\&\trnb(\upermat\transpose\shiftmat\upermat)= (1\pm\eigenreg)({\nodenum-\subspacedim}),
\nonumber
\end{align}
\reqnomode
where the separation of eigenspaces \emph{enforced} by the constraint
$\upermat\transpose\shiftmat\uparmat = \bm 0$
in \eqref{prob:exactrelaxedS} is now simply \emph{promoted} through
the last term in the objective. The parameter $\regpar>0$ is selected by
the user to balance the trade off between approximation error and
filter order. 
\cmt{solver}\arev{An iterative solver for \eqref{prob:approxrelaxed} is
proposed in Appendix~\ref{sec:approxsolver} of the supplementary
 material based on ADMM, therefore inheriting its solid convergence
 guarantees.}

\section{Numerical Experiments}
\label{sec:sim}

     \begin{figure}[t]
     \centering
         \includegraphics[width=0.4\textwidth]{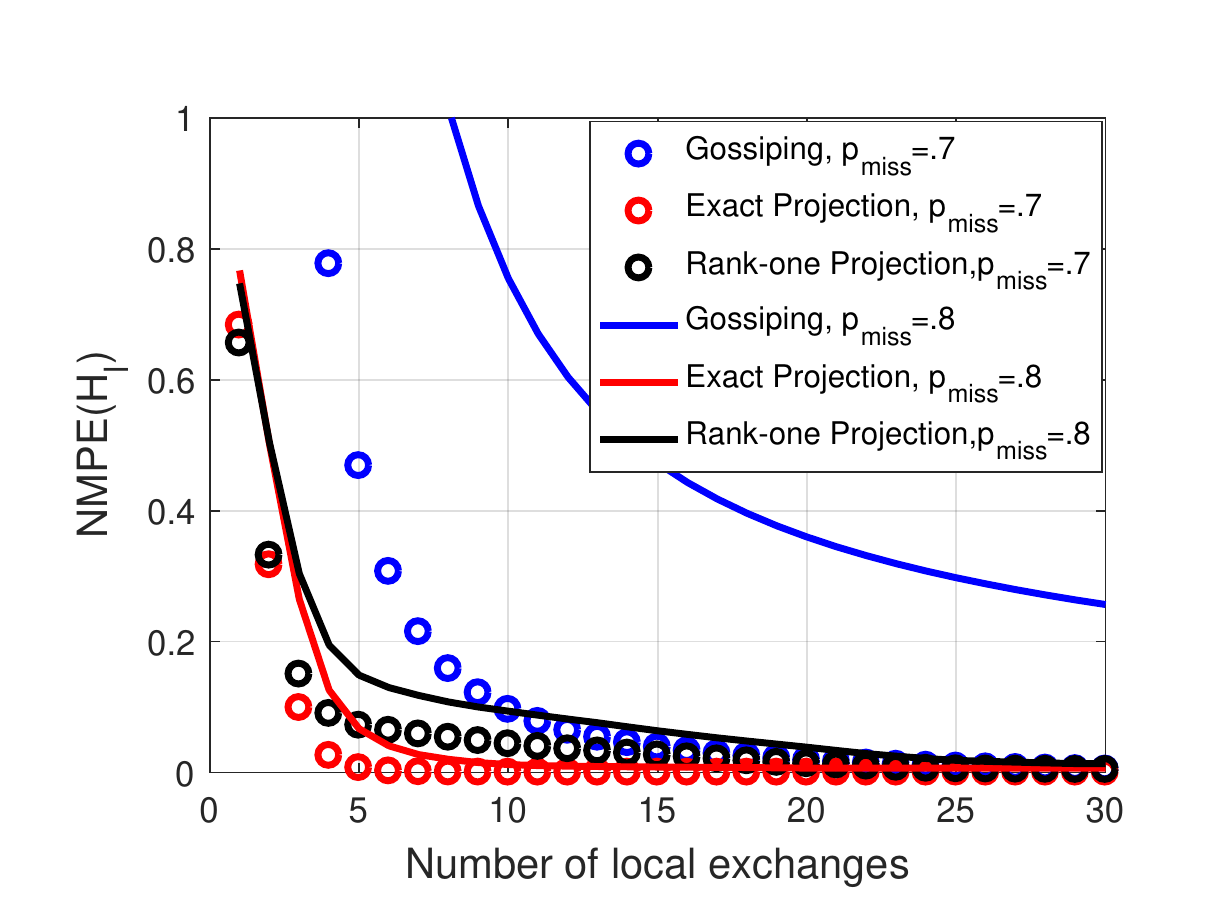}
   \caption{\small{NMPE as a function of the number of
  communications performed per node  for the Erdős–Rényi
         networks ($ \nodenum=30, \subspacedim=1$, $\rho=0.1$, $I_{\text{max}}=1000$,
$\regnucper=0.9, \regnucpar=0.1, \eigenreg=0.1$).}}
         \label{rank-one}
     \end{figure}

     \begin{figure}[t]
         \centering
         \includegraphics[width=0.4\textwidth]{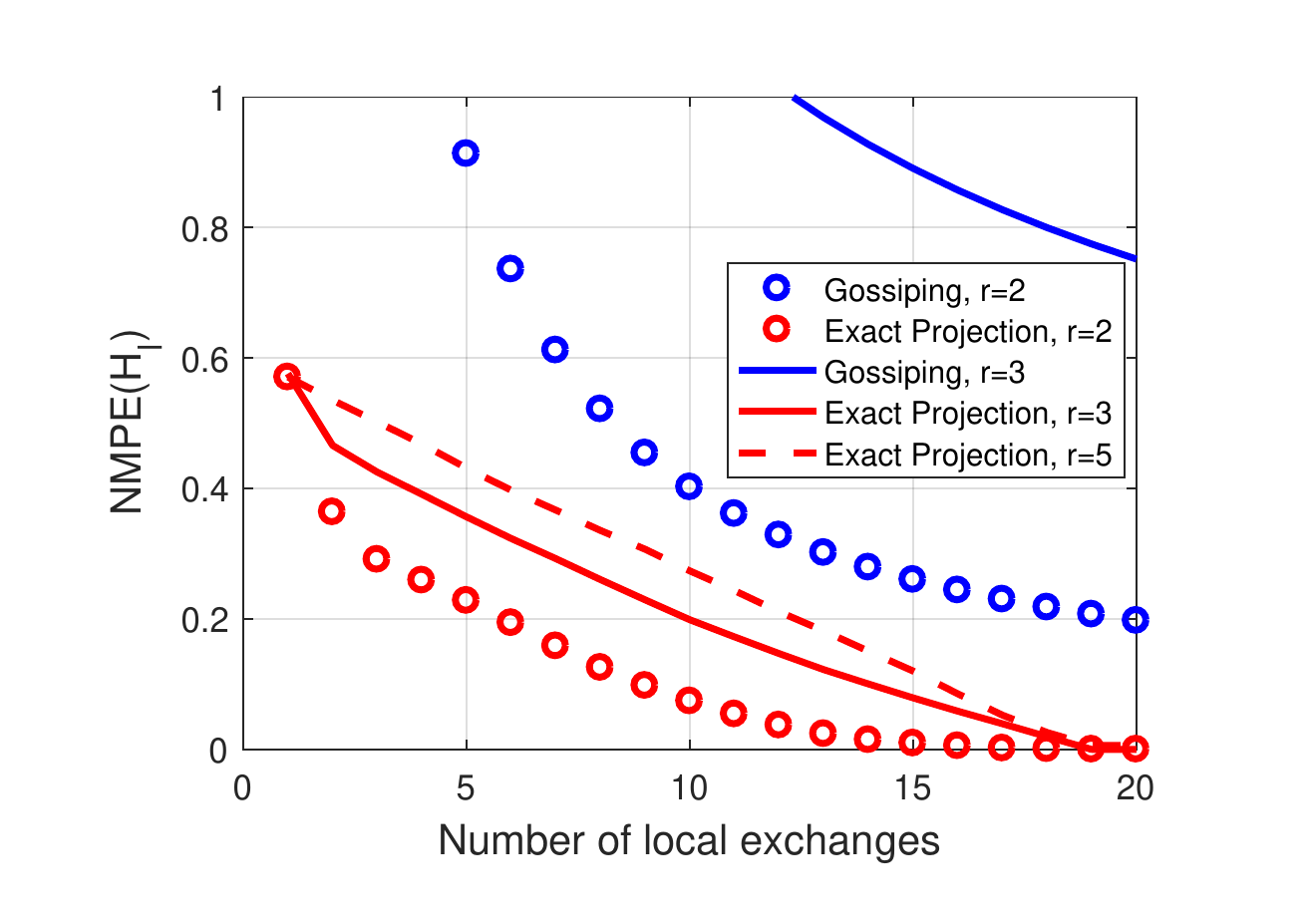}
         \caption{\small{NMPE as a function of the number of
  communications performed per node  for the Erdős–Rényi
         networks ($ \nodenum=20$, $p_{\text{miss}}=0.6$, $\rho=0.1$, $I_{\text{max}}=1000$,
$\regnucper=0.9, \regnucpar=0.1, \eigenreg=0.1$).}}
\label{NMPE-2}
     \end{figure}

     \begin{figure}[t]
         \centering
         \includegraphics[width=0.4\textwidth]{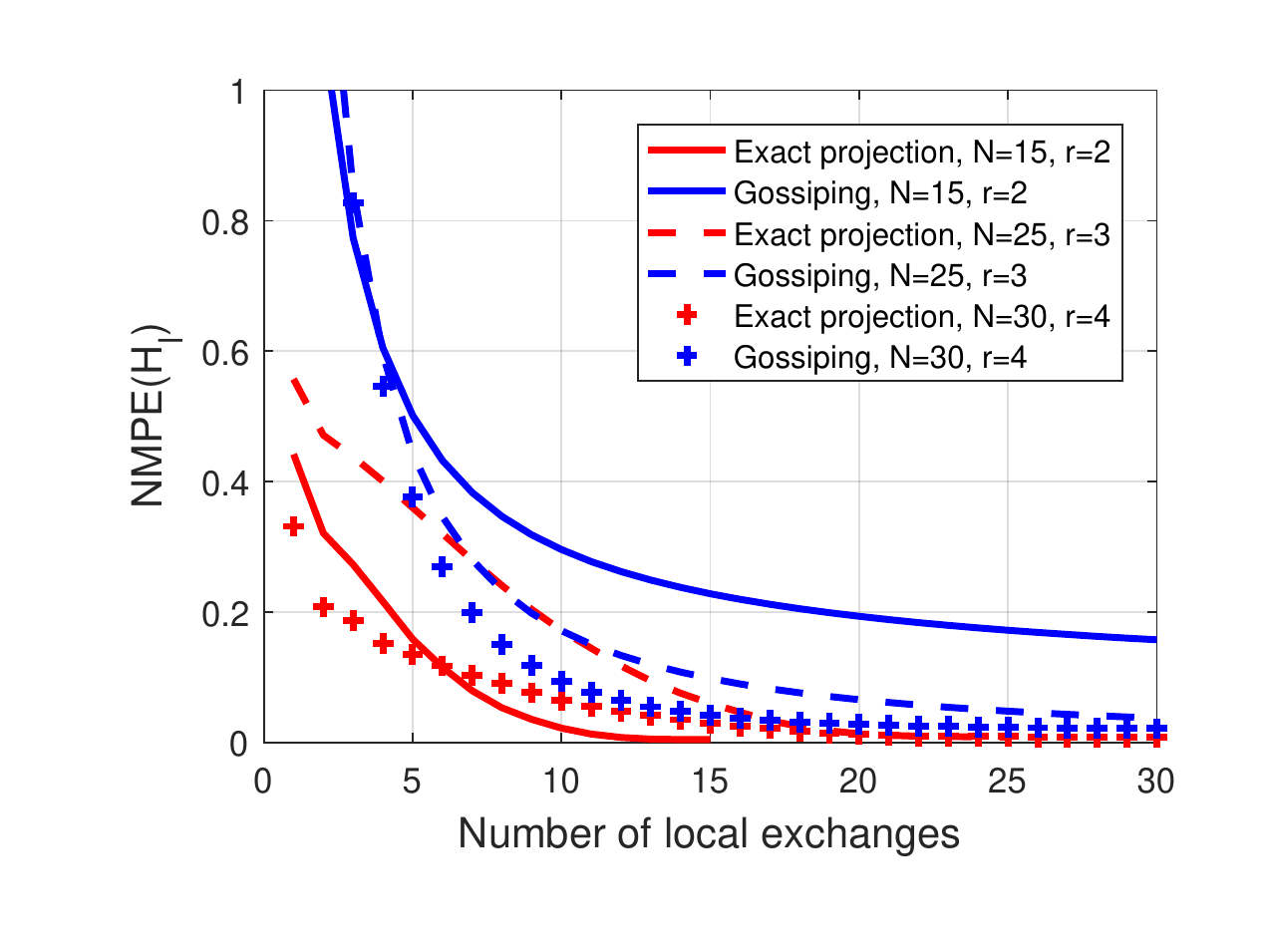}
         \caption{\small{NMPE as a function of the number of
  communications performed per node for WSN ($ d_{max}=.6, \rho=0.1$, $I_{\text{max}}=1000$,
$\regnucper=0.9, \regnucpar=0.1, \eigenreg=0.1$).}}
\label{NMPE-3}
     \end{figure}

     \begin{figure}[t]
         \centering
         \includegraphics[width=0.4\textwidth]{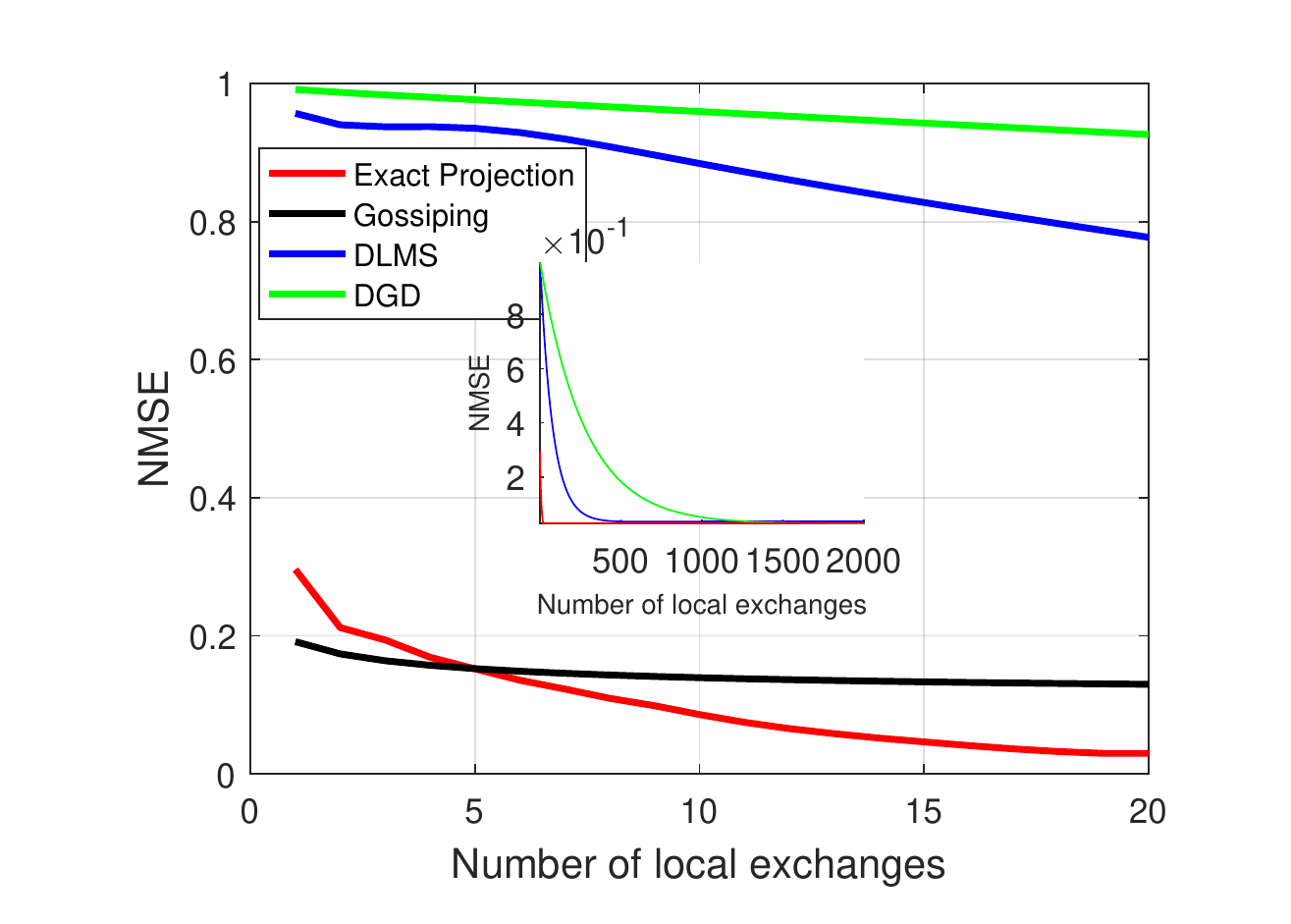}
         \caption{\small{NMSE as a function of the number of
  communications performed per node ($ \nodenum=20$, $\subspacedim=3, \beta=5,$ Erdős–Rényi graph with $p_{\text{miss}}=0.6$, $\rho=0.1$, $I_{\text{max}}=1000$,
$\regnucper=0.9, \regnucpar=0.1, \eigenreg=0.1$).}}
         \label{NMSE-4}
     \end{figure}

     \begin{figure}[t]
         \centering
         \includegraphics[width=0.4\textwidth]{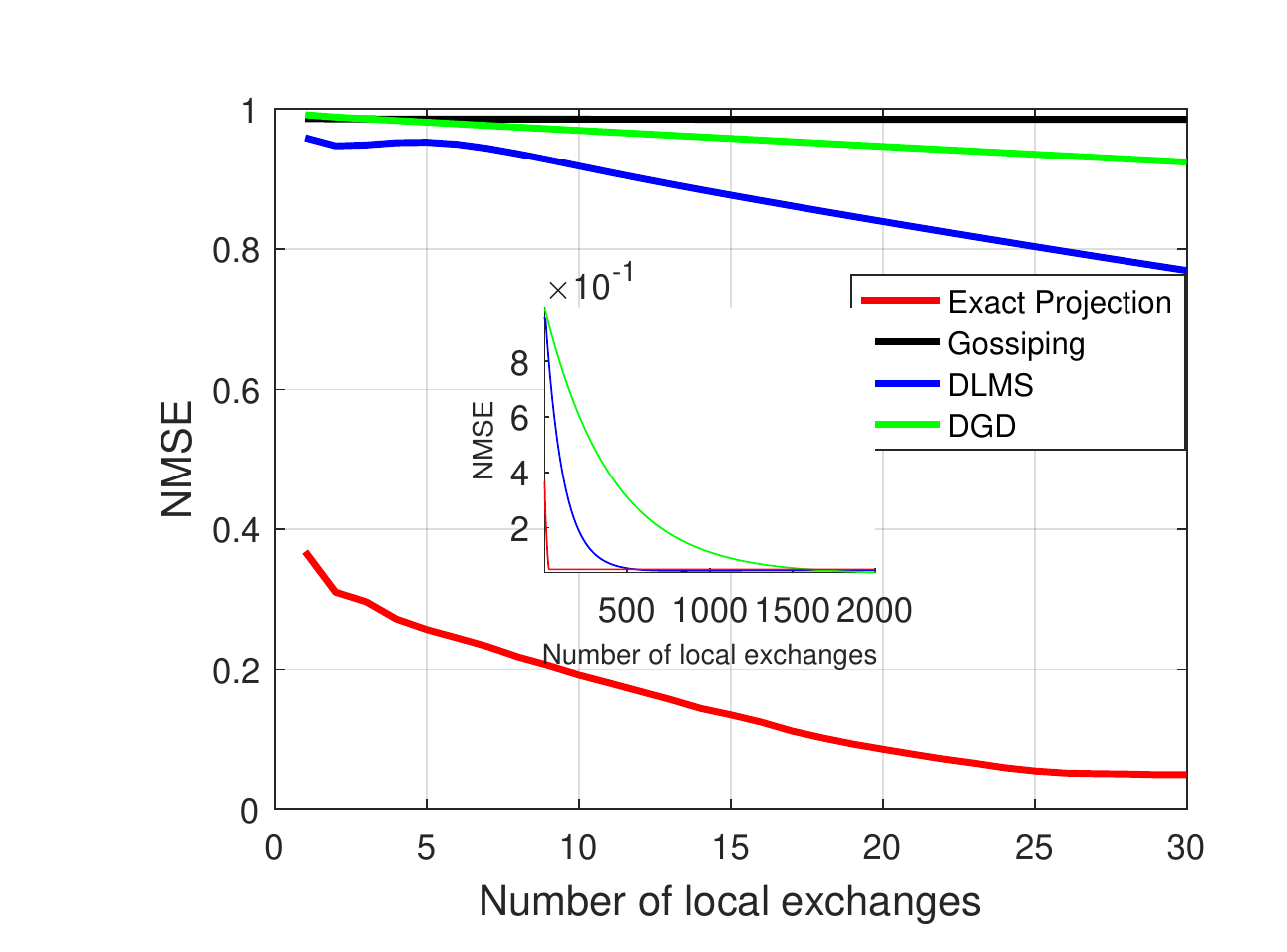}
         \caption{\small{NMSE as a function of the number of
  communications performed per node ($ \nodenum=30, \subspacedim=5$, $\beta=5,$ Erdős–Rényi graph with $p_{\text{miss}}=0.7$, $\rho=0.1$, $I_{\text{max}}=1000$,
$\regnucper=0.9, \regnucpar=0.1, \eigenreg=0.1$).}}
         \label{NMSE-5}
     \end{figure}

     \begin{figure}[t]
         \centering
         \includegraphics[width=0.4\textwidth]{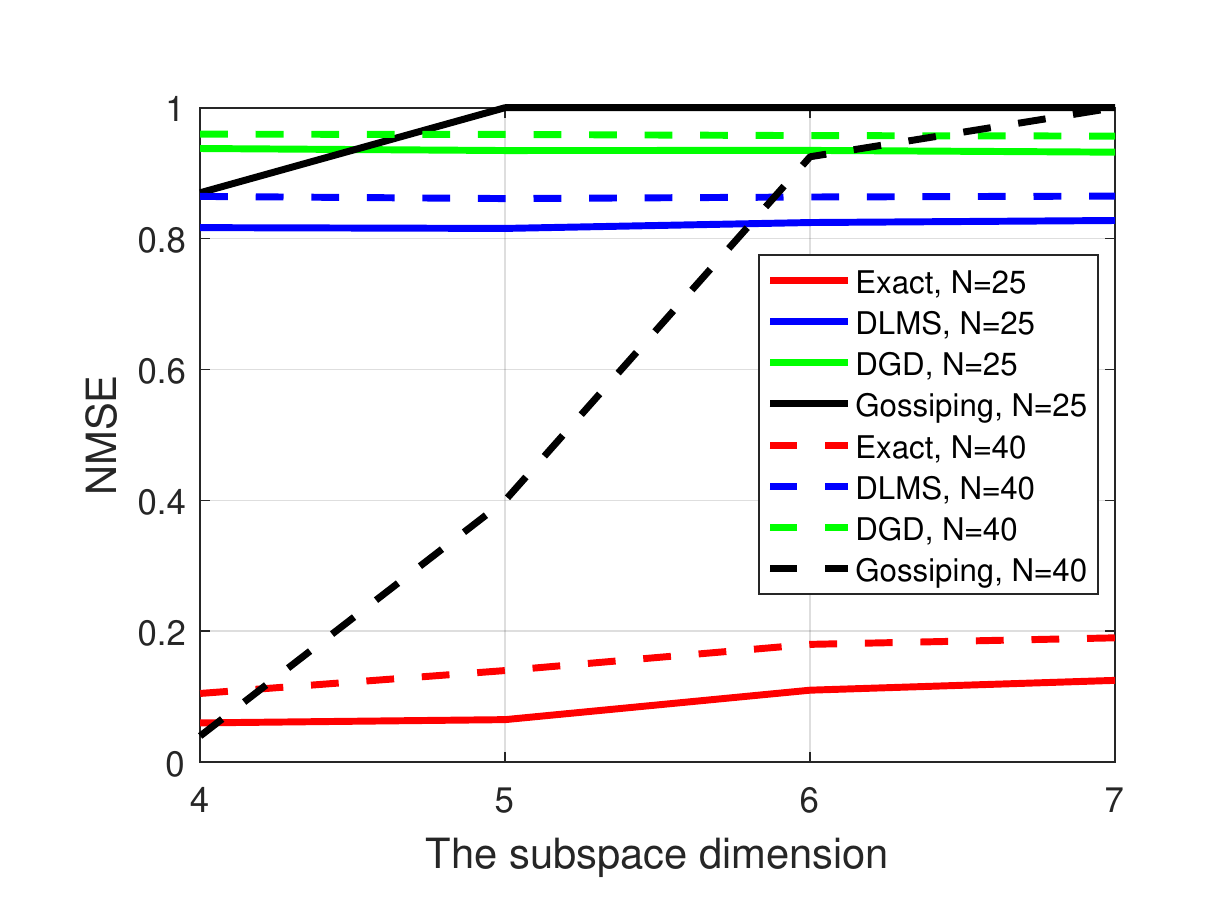}
     \caption{\small{NMSE as a function of the subspace dimension (Erdős–Rényi network
     with $p_{\text{miss}}=.7$, $\beta=5,  L=20$, $\rho=0.1$, $I_{\text{max}}=1000$,
$\regnucper=0.9, \regnucpar=0.1, \eigenreg=0.1$).}}
        \label{NMSE-6}
     \end{figure}

\begin{figure}[t]
         \centering
         \includegraphics[width=0.4\textwidth]{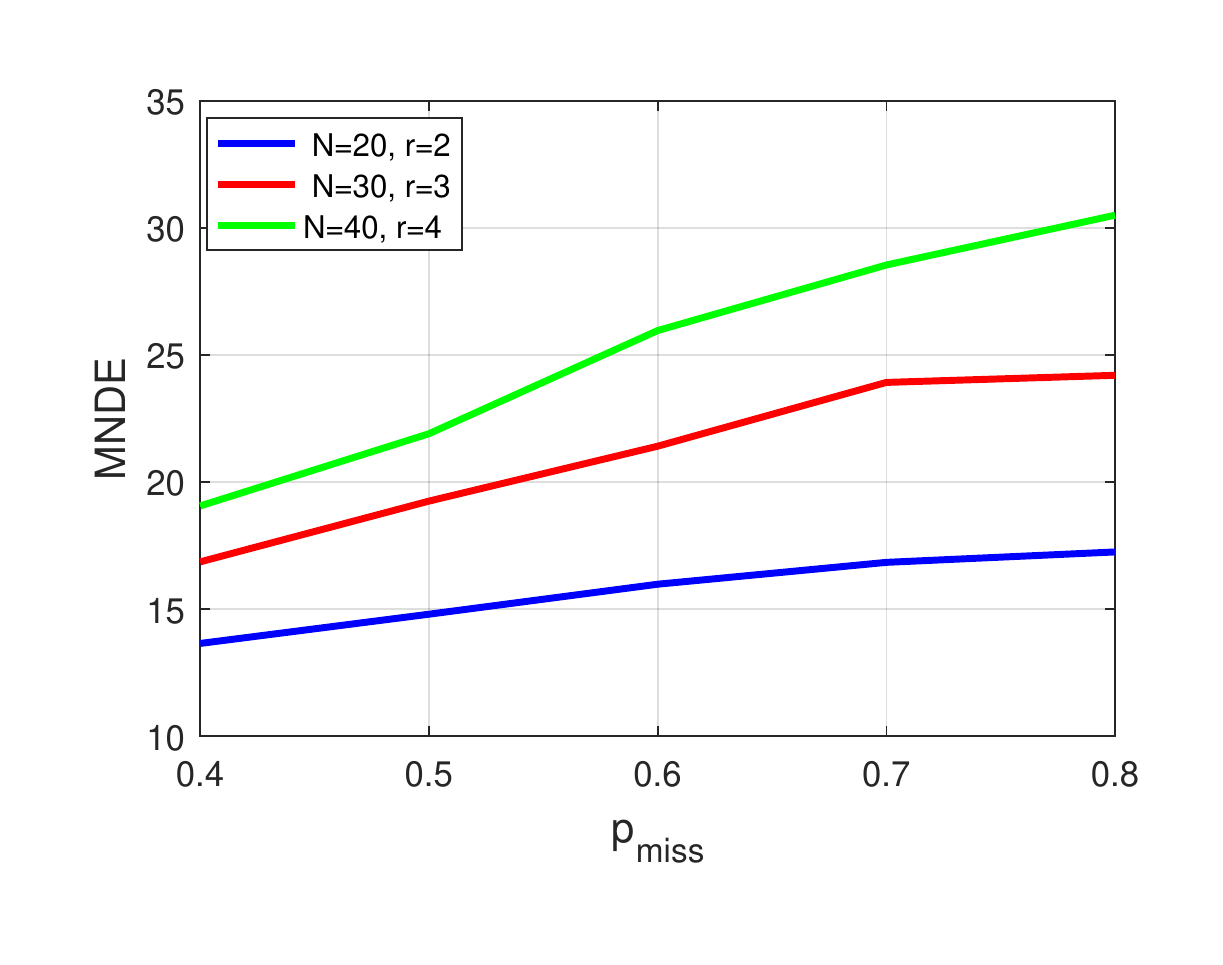}
    \caption{\small{MNDE for the exact projection method as a function of $p_{\text{miss}}$ ($\rho=0.1$, $I_{\text{max}}=1000$,
$\regnucper=0.9, \regnucpar=0.1, \eigenreg=0.1$, $\tau=0.005$, Erdős–Rényi network).}}
         \label{MNDE}
     \end{figure}

\begin{figure}[t]
\centering
\includegraphics[width=0.4\textwidth]{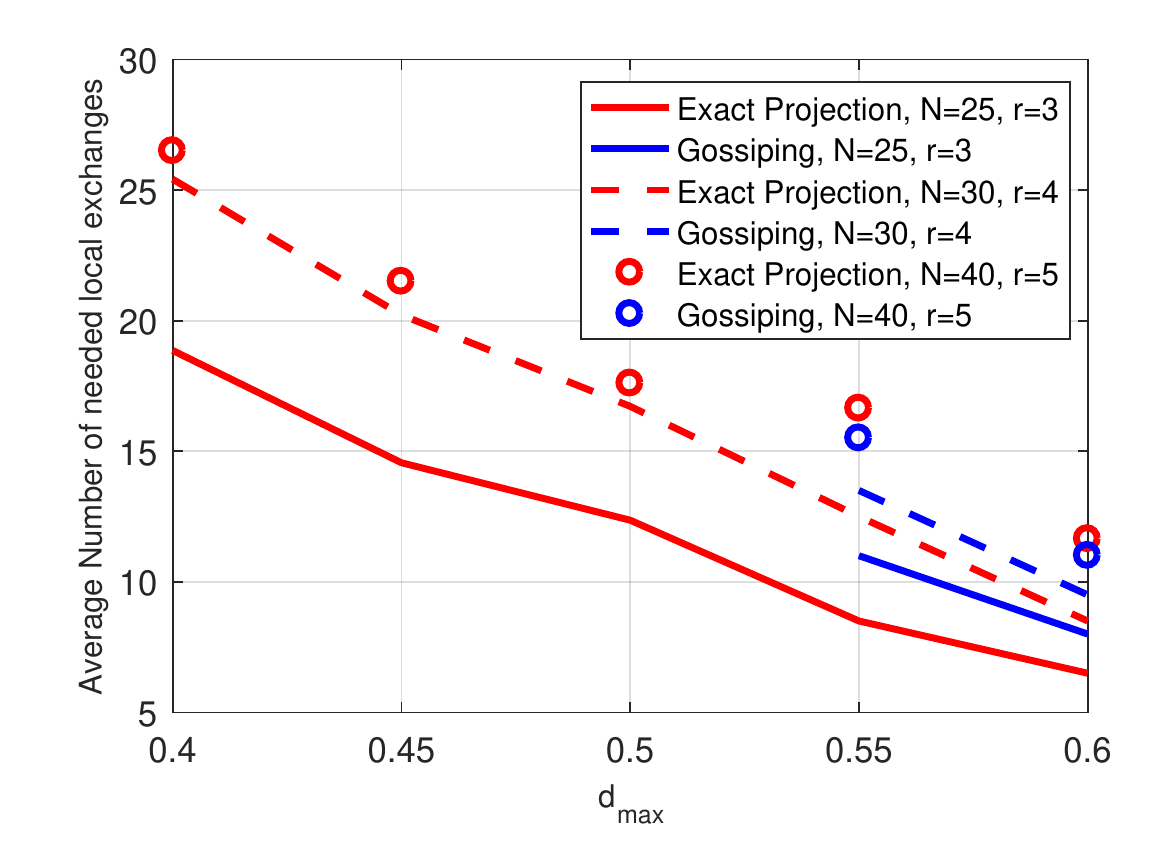}
       \caption{\small{Average number of  local exchanges required to satisfy
       $\text{NMPE}<\gamma_\text{target}=.1$ vs. the WSN graph  parameter
       $d_{\text{max}}$  ($\rho=0.1$, $I_{\text{max}}=1000$, $\regnucper=0.9, \regnucpar=0.1, \eigenreg=0.1$.)}}
\label{WSN-1}
\end{figure}

\begin{figure}[t]
         \centering
         \includegraphics[width=0.4\textwidth]{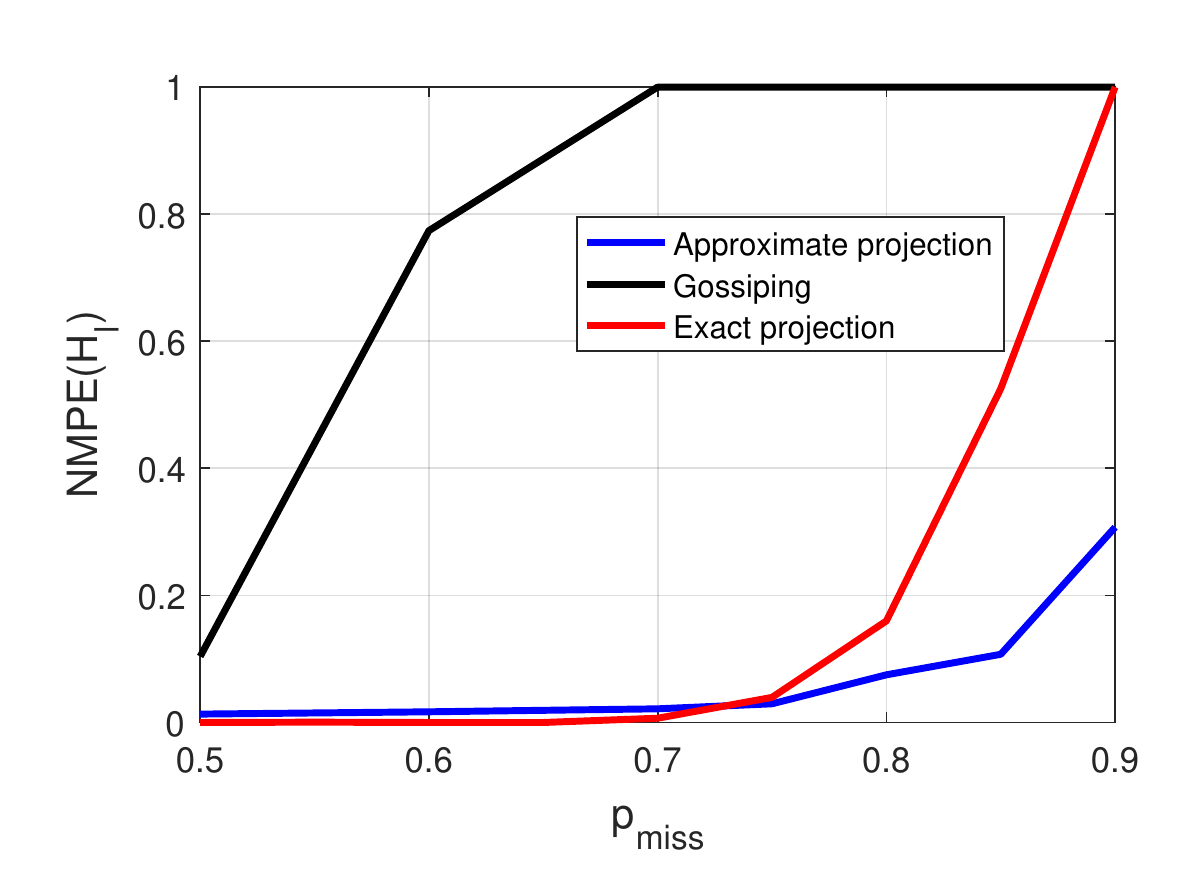}
         \caption{\small{NMPE as a function of $p_{\text{miss}}$ for the
         Erdős–Rényi networks ($\itnum=\nodenum-1$, $\nodenum=20, \subspacedim=3$, $\rho=0.1$, $I_{\text{max}}=1000$,
$\regnucper=0.9, \regnucpar=0.1, \eigenreg=0.2, \regpar=10$).}}
         \label{approx-new1}
\end{figure}

     \begin{figure}
         \centering
         \includegraphics[width=0.4\textwidth]{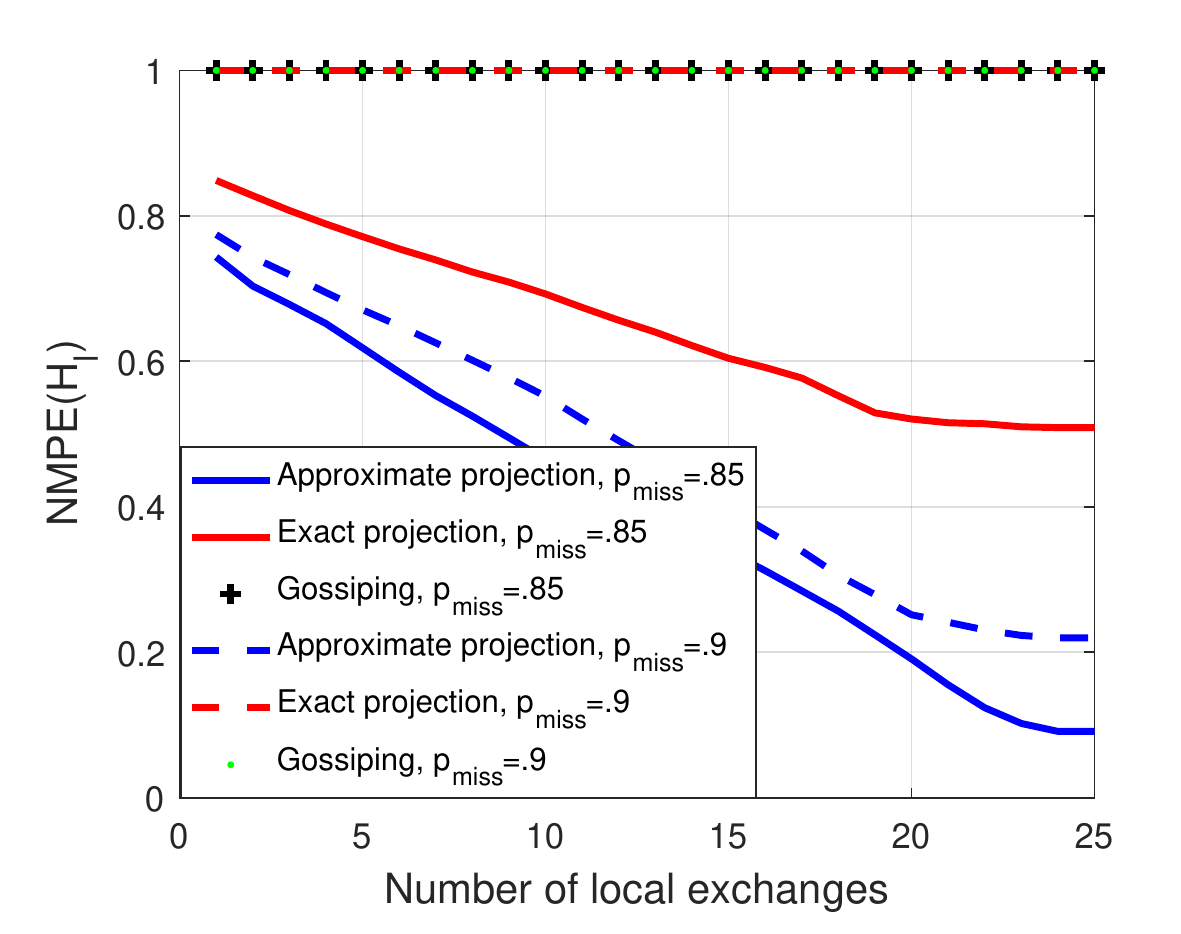}
     \caption{\small{NMPE as a function of the number of
 communications performed per node  for the Erdős–Rényi networks ($\nodenum=25, \subspacedim=3$, $\rho=0.1$, $I_{\text{max}}=1000$,
$\regnucper=0.9, \regnucpar=0.1, \eigenreg=0.2, \regpar=10$).}}
        \label{NMPE-approx-er}
     \end{figure}

     \begin{figure}[t]
         \centering
         \includegraphics[width=0.4\textwidth]{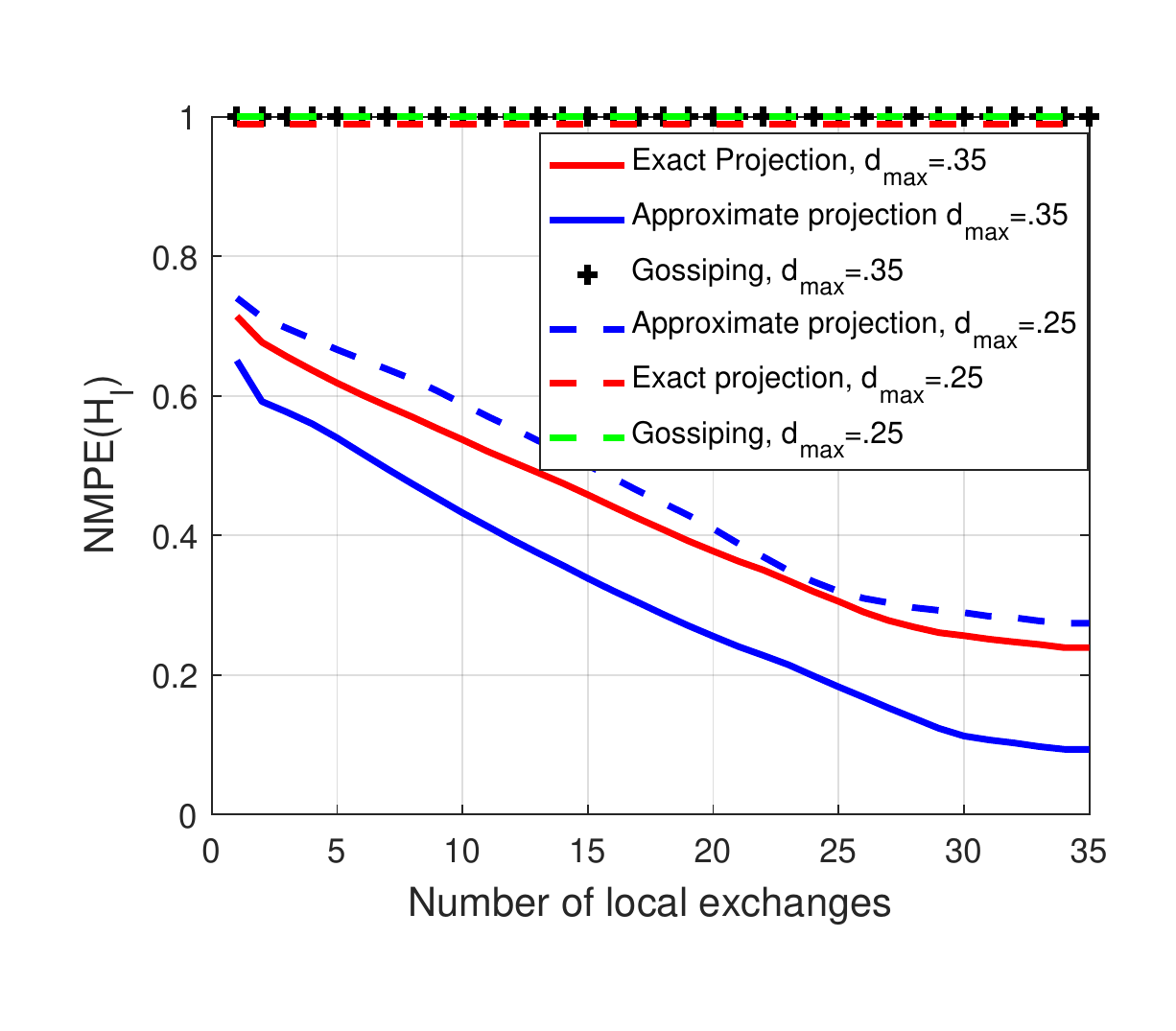}
    \caption{\small{NMPE as a function of the number of
  communications performed per node for the WSN networks ($\nodenum=35, \subspacedim=5$, $\rho=0.1$, $I_{\text{max}}=1000$,
$\regnucper=0.9, \regnucpar=0.1, \eigenreg=0.2, \regpar=10$).}}
         \label{NMPE-approx-wsn}
     \end{figure}

     \begin{figure}[t]
         \centering
         \includegraphics[width=0.4\textwidth]{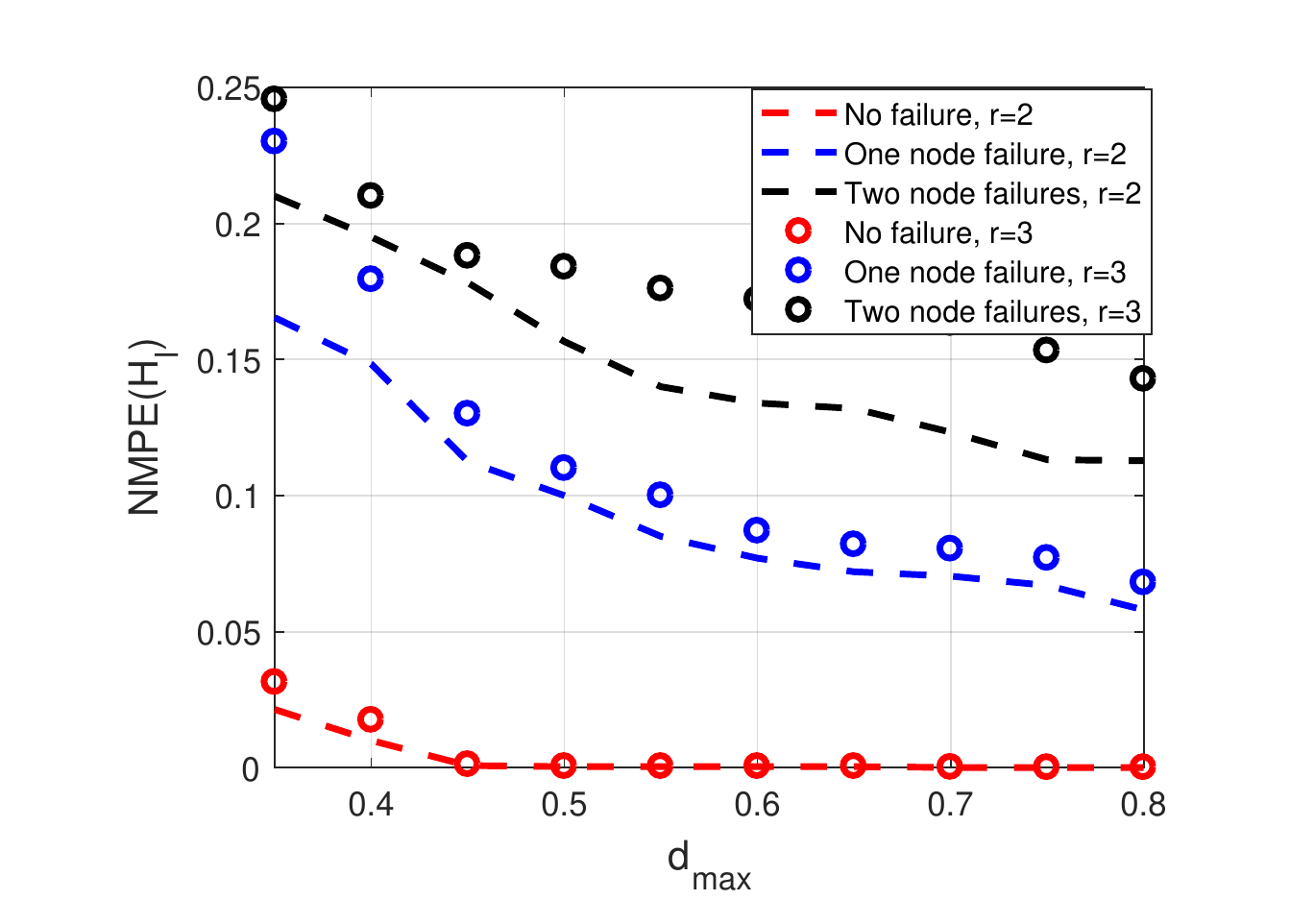}
    \caption{\small{NMPE as a function of $d_{\text{max}}$ for the WSN networks ($\nodenum=20$, $\rho=0.1$, $I_{\text{max}}=1000$,
$\regnucper=0.9, \regnucpar=0.1, \eigenreg=0.1, \regpar=10$).}}
         \label{extra}
     \end{figure}

\cmt{overview}This section validates the performance of the proposed algorithms by
means of numerical experiments.\footnote{The code necessary to reproduce the
experiments is available at \texttt{github.com/uiano/fast\_decentralized\_projections}.}

\cmt{Data Generation}The data generation process is as follows.
\begin{myitemize}%
\myitem\cmt{subspace}Matrix $\uparmat$ is obtained by orthonormalizing an
$\nodenum \times {\subspacedim}$ matrix with i.i.d. standard Gaussian
entries. \arev{This basis is assumed exactly known, meaning that the
error due to the selection of the basis is disregarded; see
Sec.~\ref{sec:basis}.}
\myitem\cmt{observations}To generate the observations
$\obsvec=\desiredvec+\noise$, the noise is drawn as
$\noise\sim\mathcal{N}(\bm 0_\nodenum,\bm I_\nodenum)$ and the signal
as
$\desiredvec=\uparmat\coordvec$, where  $\coordvec$ is obtained as 
$\coordvec=\sqrt{\beta\nodenum/\subspacedim}~\coordvec_0$
with $\coordvec_0\sim\mathcal{N}(\bm 0_\subspacedim,\bm
I_\subspacedim)$ and with  $\beta\define \expectednb{\|\desiredvec\|_2^2}/\expectednb{\|\noise\|_2^2}$ the \emph{signal-to-noise ratio} (SNR).

\end{myitemize}%

\cmt{Graph Models}The following kinds of networks are generated:
\begin{myitemize}
\myitem\cmt{Erdos-Renyi}(i) Erdős–Rényi  graphs, where the presence of each undirected edge is an
i.i.d. Bernoulli($1-p_{\text{miss}}$) random variable with
$p_{\text{miss}}$ the missing edge probability. 
\myitem\cmt{WSN}(ii) wireless sensor network (WSN)  graphs,  generated  by deploying the nodes uniformly at random
over a square area of unit side length and connecting them with an
edge if the internode distance is smaller than $d_{\text{max}}$.
\end{myitemize}%
\cmt{ensure connectivity}If the generated graph in (i) and (ii) is not connected, then additional edges are 
introduced to ensure connectivity. More specifically, if the graph has $c$
components, $c-1$ eges are added between nodes chosen at random from
each component.

\cmt{The compared methods}%
\begin{myitemize}
\myitem\cmt{least squares}Among the compared methods, those for
decentralized optimization 
obtain the projection by solving the least squares problem
$\argmin_{\coordvec}\|\obsvec-\uparmat\coordvec\|^2$. This includes
    \begin{myitemize}
    \myitem\cmt{DLMS}(ii) the \emph{distributed least mean squares} (DLMS)
    method in~\cite{mateos2012drls} with augmented Lagrangian
    parameter $\dlmsauglagpar$ and step size
    $ \dlmsstepsize$, which builds upon ADMM,
    \myitem\cmt{DGD}and (ii) the \emph{decentralized gradient descent} (DGD)
    method in~\cite{shi2017distributed} with step size
    $\mu_{\text{DGD}}$, which is based on gradient descent.
    \end{myitemize}
    \myitem\cmt{Shifts}Other methods iteratively apply a shift matrix:
    \begin{myitemize}
    \myitem\cmt{barbarossa}(iii) The \emph{gossiping} scheme
    in~\cite{barbarossa2009projection} obtains the shift matrix
    $\shiftmat$ that provides fastest convergence of
    $\shiftmat^\itind$ to $\projmat$ as $l\rightarrow \infty$
    according to a certain criterion. Then, the nodes collaboratively
    obtain $\obsvec\itnot{\itind}=\filtermat_\itind \obsvec$ with
    $\filtermat_\itind=\shiftmat^\itind$ at the $\itind$-th exchange
    round; see also Sec.~\ref{sec:graphfilters}. Signal
    $\obsvec\itnot{\itind}$ asymptotically converges to the desired
    projection.
    \myitem\cmt{Segarra}(iv) The \emph{rank-1} method
    in~\cite{segarra2017graphfilterdesign} obtains a shift matrix
    $\shiftmat$ when $\projmat$ is of rank 1 and then applies a graph
    filter. The resulting graph filter is generally of
    maximum order 
    $\nodenum-1$.
    \end{myitemize}%
    \end{myitemize}%

    \cmt{Implementation details}Some implementation details follow.
    \begin{myitemize}
    \myitem\cmt{filter coefficients}To alleviate numerical issues, the method
    in~\cite{segarra2017graphfilterdesign} and the proposed algorithms
    use node-dependent
    coefficients~\cite[Sec.~II-B]{segarra2017graphfilterdesign}. For
    comparison purposes, an order $\itind$ filter $\filtermat_\itind$
    is obtained for each number $\itind$ of local exchanges by fitting
    the node-dependent coefficients to minimize
    $\|\projmat-\filtermat_\itind\|_\text{F}^2$. 
    \myitem\cmt{feasibility}Regarding feasibility (see Sec.~\ref{sec:approximate}), the solver
     for \eqref{prob:exactrelaxed} in Appendix~\ref{sec:exactsolver}
    declares the problem as infeasible if the convergence criterion is
    not met after $I_{\text{max}}$ updates. 
    \myitem\cmt{Parameters}Both  proposed  methods use the
    same parameters in most experiments. This illustrates that a
     single set of parameters works reasonably well in a wide range of
     scenarios, which facilitates parameter tuning. For fairness,
     the competing methods use the same parameters in all experiments:
     $\dlmsauglagpar=0.001, \dlmsstepsize=1$,
     $\mu_{\text{DGD}}=0.1$. These values were adjusted to
     approximately yield the best performance in these scenarios.
    \end{myitemize}%

    \cmt{Metrics}Two main performance metrics will be adopted.
    \begin{myitemize}
    \myitem\cmt{NMSE}To quantify estimation error, some experiments
    obtain the \emph{normalized mean square error} (NMSE), defined as
    $\text{NMSE}(\filtermat_\itind)\define \mathbb{E}\left[||\boldsymbol
    {\xi}-\filtermat_\itind \obsvec ||_2^2 \right]
    / \mathbb{E}\left[||\desiredvec ||_2^2\right] $, where the
    expectation is taken over $\graph$,
    $\uparmat, \coordvec,$ and $\noise$.
    \myitem\cmt{NMPE}For the methods based on graph shift operators, the
     error between the obtained $\filtermat_\itind$ and the target
     $\projmat$ is measured through the \emph{normalized mean
     projection error}
     $\text{NMPE}(\filtermat_{\itind})\triangleq 
    \mathbb{E}\|\projmat -\filtermat_\itind\|_\text{F}^2/
       \subspacedim$, where $\mathbb{E}$ averages over $\graph$,
     and $\uparmat$. The normalization factor $\subspacedim$ was chosen so
     that $\text{NMPE}(\filtermat_{\itind})$ equals the NMSE when
     $\bm \xi=\projmat \bm z'$ with $\bm z'\sim\mathcal{N}(\bm
     0_\nodenum,\bm I_\nodenum)$.
     \myitem\cmt{MC}\arev{Both the NMSE and NMPE will be estimated by
    averaging across 500 Monte Carlo realizations.}
    \myitem\cmt{infeasibility}The optimization problems solved by some of these methods, namely the exact
     projection and the gossiping method, may be infeasible for
     certain realizations of $\graph$ and $\uparmat$; see
     Sec.~\ref{sec:approximate}. In that case, $\filtermat_\itind$ is
 set to $\bm 0$, which would penalize the tested algorithm by setting its  NMSE or NMPE closer to 1.

\end{myitemize}%

\subsection{Exact Projection Filters}

\cmt{Shift-based approaches}Figs.~\ref{rank-one}-\ref{NMPE-3} 
 depict the NMPE for those algorithms that rely on shift
 matrices. Whereas Figs.~\ref{rank-one}-\ref{NMPE-2} adopt an
 Erdős–Rényi random graph, a WSN is used in Fig.~\ref{NMPE-3}. The
 method in~\cite{segarra2017graphfilterdesign} is not displayed in
 Figs.~\ref{NMPE-2}-\ref{NMPE-3} because it cannot be applied when
 $\subspacedim>1$.
\begin{myitemize}%
\myitem\cmt{overall performance}Relative to the competing
 alternatives, the proposed exact projection method is seen to
 generally require a significantly smaller number of local exchange
 rounds to obtain an NMPE close to 0.
\myitem\cmt{finite iterations}%
\begin{myitemize}
\myitem\cmt{proposed}The aforementioned figures also reveal that the proposed
 method provides an exact projection in a finite number of
 iterations. As predicted by the Cayley-Hamilton Theorem, the order is
 never greater than $\nodenum-1$. However, most of the times, the
 actual order is much smaller than this upper bound. This phenomenon
 may not be clear at first glance from Figs.~\ref{rank-one}-\ref{NMPE-3} because they reflect \emph{average} behavior across
 a large number of Monte Carlo iterations. For this reason,
     $\text{NMPE}(\filtermat_{\itind})$ is generally positive for all
     $\itind<\nodenum-1$ when one or more realizations yield a
     filter with maximum order $\itind=\nodenum-1$.
\myitem\cmt{rank-1}On the other hand, the rank-1 method generally yields exact
 projection filters with order $\nodenum-1$ since it is not designed
 to minimize the order.  \myitem\cmt{gossiping}The gossiping method,
 however, does not necessarily produce an exact projection in a finite
 number of iterations. In exchange
 its implementation is simpler; cf. Sec.~\ref{sec:barbarossa}.
\end{myitemize}%

\end{myitemize}%

\cmt{Non-shift-based approaches}To compare with the competing
algorithms that are not based on graph filtering,
\begin{myitemize}
\myitem\cmt{vs. num exchanges}Figs.~\ref{NMSE-4}-\ref{NMSE-5} depict     $\text{NMSE}(\filtermat_\itind)$
vs.  the number $\itind$  of local exchanges for different
scenarios. 
\begin{myitemize}%
\myitem\cmt{NMSE}Observe that the NMSE of the exact projection  method does not converge to 0 due to
the observation noise $\noise$.
\myitem\cmt{scale}The miniatures demonstrate the different convergence time scales of
 DGD, DLMS, and the exact projection method. The cause for the slower
 convergence of DGD and DLMS is twofold: first, DGD and DLMS are
 general-purpose methods whereas the proposed method is tailored to
 the subspace projection problem. Second, DGD and DLMS require the
 selection of step size parameters, which in the simulations here were
 set to ensure convergence in virtually all Monte Carlo
 realizations. For certain specific realizations, though, one could
 find step sizes that yield a faster convergence.
\end{myitemize}%
\myitem\cmt{vs. rank}An alternative perspective to the comparison in
 Figs.~\ref{NMSE-4}-\ref{NMSE-5} is offered by Fig.~\ref{NMSE-6},
which depicts the NMSE vs. the subspace dimension $\subspacedim$ when
the number of local exchange rounds is fixed to $\itnum$. This
analysis would be necessary e.g. in real-time applications. It is
interesting to observe that the sensitivity to $\subspacedim$ is higher for the
gossiping algorithm as compared to the other methods.

\end{myitemize}

\cmt{graph sparsity}
\begin{myitemize}
\myitem\cmt{MNDE}To analyze the impact of the graph sparsity in the
required number of local exchanges to implement a projection exactly
with the proposed method, Fig.~\ref{MNDE} shows the \emph{mean number
of distinct eigenvalues} (MNDE) yielded by the exact projection method
vs.  $p_{\text{miss}}$. The NMDE is defined as the mean of the number
of distinct eigenvalues of $\fparmat$ plus the number of distinct
eigenvalues of $\fpermat$, which equals the order of the filter;
cf. Sec.~\ref{sec:filterorderminimization}. A threshold $\tau$ is used
to determine whether two eigenvalues are different. As expected, the
filter order increases as the graph becomes sparser; see
Sec.~\ref{sec:feasibility}. Remarkably, the increase is more
pronounced for larger networks and higher $\subspacedim$.

\myitem\cmt{required local exchanges}In the case of WSN graphs, the
impact of sparsity is studied in Fig.~\ref{WSN-1}, which displays the
number of local exchange rounds required to reach an NMPE below a
target value $\gamma_\text{target}$.  The gossiping method is also
shown for comparison purposes, yet it is not capable of attaining this
target for certain values of $d_\text{max}$ due to the infeasibility
of the optimization problem that it solves in a significant fraction
of the Monte Carlo realizations; see explanation after the definition
of NMSE earlier in Sec.~\ref{sec:sim}. As $d_{\text{max}}$ increases,
the network becomes more densely connected and therefore an exact
projection can be obtained in a smaller number of local
exchanges. This agrees with intuition and with Fig.~\ref{MNDE}.

\end{myitemize}

\subsection{Approximate Projection Filters}
\cmt{fixed num. local exchanges}As explained in
Sec.~\ref{sec:approximate}, a graph filter capable of exactly
 implementing a projection may not exist when the graph is highly
 sparsely connected. The approximate projection method in that section provides a
 graph filter capable of approximating such a projection. To
 illustrate how this algorithm complements the exact projection
 method from Sec.~\ref{sec:exactproj}, Fig.~\ref{approx-new1} shows
 NMPE after $\itnum=\nodenum-1$ local exchanges for both methods along
 with the gossiping algorithm.
 \begin{myitemize}
\myitem\cmt{when exact}The exact projection method is seen to provide a filter that
exactly implements the target projection when $p_\text{miss}$ is below
approximately 0.7. For larger $p_\text{miss}$, the NMPE becomes
strictly positive. It is important to emphasize that the reason
is \emph{not} that this method provides filters that do not exactly
implement the target projection. When the problem is feasible, the
filters yield an exact projection. However, as $p_\text{miss}$
increases, a smaller fraction of Monte Carlo realizations give rise to
a feasible problem and this is penalized in the NMSE computation; see
explanation after the definition of NMSE earlier in
Sec.~\ref{sec:sim}.
\myitem\cmt{relative to approx}As expected, 
the exact projection method performs better than its approximate
counterpart when exact projection filters exist. However, for
sufficiently sparse graphs, the NMPE of the exact projection method
explodes, whereas the approximate projection method remains low. 
\end{myitemize}

\cmt{vs. local exchanges}Figs.~\ref{NMPE-approx-er}
and~\ref{NMPE-approx-wsn} present the evolution of the NMPE vs. the
number of local exchange rounds for Erdős–Rényi and WSN graphs
respectively with several degrees of sparsity. These figures showcase
that the proposed approximate projection method can reasonably
approximate a projection with a small number of local exchanges even
if the graph is highly sparse. Note that for some of the sparsity
levels used in these figures, the problems solved by the gossiping
algorithm and the exact projection method become infeasible.

\cmt{Robustness}\arev{To analyze the robustness of the approximate projection
method, Fig.~\ref{extra} depicts its NMPE for different numbers of
node failures. Note that the exact projection method cannot be applied
since the removal of one or more sensors renders the exact projection
problem infeasible except for trivial cases where their measurements
bear no information about the estimated signal. Whenever the network
detects that a node has failed, the graph filter is recomputed with
the updated topology.  As expected, the degradation is more pronounced
when the network is sparser (smaller $d_{\max}$). Due to the
distribution used to generate $\uparmat$, each measurement contains,
on average, a fraction $1/\nodenum$ of the signal energy. This informally
indicates that $n/\nodenum$ lower bounds the error attainable in the
presence of $n$ node failures. Observe that the NMPE is indeed close
to such fundamental limits, thereby establishing that the approximate
projection method is reasonably robust to node failures.  }

\section{\arev{Conclusions and Discussion}}
\label{sec:conclusions}

This paper develops methods to obtain graph filters that can be used
 to compute subspace projections in a decentralized fashion. The
 approach relies on transforming the filter design task into a
 shift-matrix design problem. The first method addresses the latter by
 exploiting the eigenstructure of feasible shift matrices. The second
 method builds upon the first to approximate projections when no
 feasible projection filter exists. An exhaustive simulation analysis
 demonstrates the ability of the designed filters to effectively
 produce subspace projections in a small number of iterations. This
 contrasts with existing methods, whose convergence is generally
 asymptotic.

 \arev{The main strength of graph filters lies in their simplicity. As
 described before, the inference operations considered in this paper can be
 implemented with $\itnum\nodenum$ transmissions. Lower and
 higher values can be attained by alternative communication
 strategies, including centralized processing. The approach adopted in
 practice will depend on factors such as energy and hardware
 constraints.  }

\arev{The filters obtained through the proposed methods also inherit
 the general numerical limitations of graph filters. In turn, these
 limitations stem from the well-known  conditioning issues of
 Vandermonde systems, e.g. \eqref{eq:vander}. This limits the
 application of graph filters to networks whose number of nodes is
 comparable to the ones in Sec.~\ref{sec:sim}. Further work is
 required by the graph signal processing community to enable the
 implementation of graph filters in significantly larger networks.}

\arev{Another direction along which the proposed schemes can be
 extended is by lifting the symmetry constraint on $\shiftmat$. Since
 the target matrix $\projmat$ is symmetric, the gain may not be
 significant. Future work will investigate this possibility.}

\balance
\if\editmode1 
\onecolumn
\printbibliography
\twocolumn
\else
\bibliography{\bibfilenames}
\fi

\appendices

\section{Proof of \thref{prop:prefeasible}}
\label{proof:prefeasible}
Let
$\filtermat(\filtercoefvec,\shiftmat)\define \sum_{l=0}^{L}c_{l}\shiftmat^{l}$
and note from the definition of $\polfeasmatset{\projmat}$ that
$\shiftmat\in\polfeasmatset{\projmat}$ iff there exists
$\filtercoefvec$ such that
\begin{align}
\label{eq:hequaluiu}
\filtermat(\filtercoefvec,\shiftmat)=\uparmat\uparmat^{\top}=\begin
{bmatrix}\uparmat& \upermat\end
{bmatrix}\left[	\begin{matrix} \identity_{\subspacedim}& \bm{0}\\ \bm{0}&\bm{0}\end{matrix} \right]\begin
{bmatrix}\uparmat^{\top}\\ \upermat^{\top}\end {bmatrix}.
\end{align}
Since $\shiftmat$ is symmetric, it admits an  eigenvalue
decomposition
$\shiftmat=\topologyconstraintmat\bm{\Lambda}\topologyconstraintmat^{\top}$. Thus,
\begin{subequations}
\label{decomposition}
\begin{align}
\filtermat(\filtercoefvec,\shiftmat)=&\topologyconstraintmat\left(\sum_{l=0}^{L}c_l\bm{\Lambda}^{l}\right)\topologyconstraintmat^{\top}
\\=&
\label{eq:decompositionb}
\topologyconstraintmat \left[ \begin{matrix} \identity_{r}& \bm{0}\\ \bm{0}&\bm{0}\end{matrix} \right]
\topologyconstraintmat\transpose
\\=&
\label{eq:decompositionc}
\topologyconstraintmat_{\parallel}\topologyconstraintmat_{\parallel}^{\top},
\end{align}
\end{subequations}
where $ \topologyconstraintmat_{\parallel}$ comprises the first
$\subspacedim$ columns of
$\topologyconstraintmat\define[ \topologyconstraintmat_{\parallel}, \topologyconstraintmat_{\perp}]$. The
second and third equalities in  \eqref{decomposition} follow from the
fact that
\eqref{eq:hequaluiu} implies that 
$\filtermat(\filtercoefvec,\shiftmat)$ has an eigenvalue 1 with
multiplicity $\subspacedim$ and an eigenvalue 0 with multiplicity
$\nodenum-\subspacedim$ and,
therefore, \eqref{eq:decompositionb} and \eqref{eq:decompositionc}
must hold for some ordering of the eigenvectors in the
eigendecomposition $\shiftmat=\topologyconstraintmat\bm{\Lambda}\topologyconstraintmat^{\top}$.

From \eqref{eq:hequaluiu} and \eqref{eq:decompositionc}, it follows
that
$\topologyconstraintmat_{\parallel}\topologyconstraintmat_{\parallel}^{\top}=\uparmat\uparmat^{\top}$
and, therefore,
$\range{\topologyconstraintmat_{\parallel}\topologyconstraintmat_{\parallel}^{\top}}=\range{\uparmat\uparmat^{\top}}$,
which in turn is equivalent to
$\range{\topologyconstraintmat_{\parallel}}=\range{\uparmat}$. Consequently,
$\topologyconstraintmat_{\parallel}=\uparmat\dummymat_{\parallel}$ for
some $\dummymat_{\parallel}$. Moreover,
since
$\identity_{r}=\topologyconstraintmat_{\parallel}^{\top}\topologyconstraintmat_{\parallel}=\dummymat_{\parallel}^{\top}\uparmat^{\top}\uparmat\dummymat_{\parallel}=\dummymat_{\parallel}^{\top}\dummymat_{\parallel}$,
one can conclude that $\dummymat_{\parallel}$ is orthonormal. Similarly,
$\range{\topologyconstraintmat_{\perp}}=\rangeperp{\topologyconstraintmat_{\parallel}}=\rangeperp{\uparmat}=\range{\upermat}$,
which implies that 
$\topologyconstraintmat_{\perp}=\upermat\dummymat_{\perp}$ for some orthogonal
$\dummymat_{\perp}$.

Upon letting $\bm{\Lambda}_{\parallel}$ and  $\bm{\Lambda}_{\perp}$ be such that
\begin{align}\label{shift operator}
\shiftmat=\topologyconstraintmat\bm{\Lambda}\topologyconstraintmat^{\top}\define\begin {bmatrix}\topologyconstraintmat_{\parallel}& \topologyconstraintmat_{\perp}\end {bmatrix}\left[	\begin{matrix} \bm{\Lambda}_{\parallel}& \bm{0}\\ \bm{0}&\bm{\Lambda}_{\perp}\end{matrix} \right]\begin {bmatrix}\topologyconstraintmat_{\parallel}^{\top}\\ \topologyconstraintmat_{\perp}^{\top}\end {bmatrix}
\end{align}
and applying the above relations to \eqref{shift operator}, it follows
that 
\begin{align}\label{relation}
\shiftmat=\begin {bmatrix}\uparmat\dummymat_{\parallel}& \upermat\dummymat_{\perp}\end {bmatrix}\left[	\begin{matrix} \bm{\Lambda}_{\parallel}& \bm{0}\\ \bm{0}&\bm{\Lambda}_{\perp}\end{matrix} \right]\begin {bmatrix}\dummymat_{\parallel}^{\top}\uparmat^{\top}\\ \dummymat_{\perp}^{\top}\upermat^{\top}\end {bmatrix}=\nonumber\\\begin {bmatrix}\uparmat& \upermat\end {bmatrix}\left[	\begin{matrix} \fparmat& \bm{0}\\ \bm{0}&\fpermat\end{matrix} \right]\begin {bmatrix}\uparmat^{\top}\\ \upermat^{T}\end {bmatrix},
\end{align}
where $\fparmat\define\dummymat_{\parallel}\bm{\Lambda}_\parallel\dummymat_{\parallel}^{\top}$ and $\fpermat\define\dummymat_{\perp}\bm{\Lambda}_\perp\dummymat_{\perp}^{\top}$.

\section{Proof of \thref{prop:evalsnucleardif}}
\label{proof:evalsnucleardif}
Let $\bm A=\auxevecmat\auxevalmat\auxevecmat\transpose$ be an
 eigenvalue decomposition. Then,
 \begin{align*}
 {{\left\| \bm{A}\otimes {\identity_{\nodenum}}-{\identity_{\nodenum}}\otimes \bm{A} \right\|}_{*}}={{\| \auxevecmat\auxevalmat {{\auxevecmat}^{\top}}\otimes {\identity_{\nodenum}}-{\identity_{\nodenum}}\otimes \auxevecmat\auxevalmat {\auxevecmat^{\top}} \|}_{*}},
\end{align*}
Applying the properties of the Kronecker product and the invariance of
the nuclear norm to orthogonal transformations, 
\begin{align*}
 {{\left\| \bm{A}\otimes {\identity_{\nodenum}}-{\identity_{\nodenum}}\otimes \bm{A} \right\|}_{*}}=
 &{{\|
 (\auxevecmat\otimes \identity_{\nodenum})(\boldsymbol\auxevalmat \otimes
 {\identity_{\nodenum}}){{(\auxevecmat\otimes \identity_{\nodenum})}^{\top}}}}\\&{{-({{\identity_{\nodenum}}^{{}}}\otimes\auxevecmat
 )({\identity_{\nodenum}}\otimes \boldsymbol\auxevalmat
 ){{( \identity_{\nodenum} \otimes \auxevecmat)}^{\top}} \|}_{*}}&\nonumber\\&={{\left\| (\boldsymbol\auxevalmat \otimes {\identity_{\nodenum}})-({\identity_{\nodenum}}\otimes \boldsymbol\auxevalmat ) \right\|}_{*}}.
\end{align*}
From the definition of the nuclear and $\ell_1$ norms,
\begin{subequations}
\begin{align}
\nonumber
{{\|
(\boldsymbol\auxevalmat \otimes
{\identity_{\nodenum}})}}&{{-({\identity_{\nodenum}}\otimes \boldsymbol\auxevalmat
) \|}_{*}}
={{\left\| \mathrm{diag}(\boldsymbol\auxevalmat \otimes
{\identity_{\nodenum}}-{\identity_{\nodenum}}\otimes \boldsymbol\auxevalmat
) \right\|}_{1}}
\\=& {{\left\| \boldsymbol{\lambda} \otimes
{\one_{\nodenum}}-{\one_{\nodenum}}\otimes \boldsymbol{\lambda}  \right\|}_{1}}
\\=&\sum\limits_{i=1}^{\nodenum}{\mid (\boldsymbol{\lambda} \otimes
{\one_{\nodenum}}-{\one_{\nodenum}}\otimes \boldsymbol{\lambda}
{{)}_{i}} \mid}
\\=&\sum\limits_{i=1}^{\nodenum}{\mid (\boldsymbol{\lambda} \otimes
{\one_{\nodenum}}-{\one_{\nodenum}}\otimes \boldsymbol{\lambda}
)\transpose{\boldsymbol{e}_{{{\nodenum}^{2}},i}} \mid},
\label{eq:abssumlambdadif}
\end{align}
\end{subequations}
where
$ \boldsymbol{\lambda}\define \mathrm{diag}(\boldsymbol\auxevalmat)$ and
 $\canbasisvec_{M,i}$ is the $i$-th column of
 $\identity_{M}$. Splitting the summation
 in \eqref{eq:abssumlambdadif} and applying the fact that
 $\boldsymbol{e}_{{{\nodenum}^{2}},\nodenum(j-1)+k}=
 {\boldsymbol{e}_{\nodenum,j}}\otimes
 {\boldsymbol{e}_{\nodenum,k}}~\forall j,k$, it follows that
\begin{align}
 {{\| \bm{A}\otimes
 {\identity_{\nodenum}}}}&{{-{\identity_{\nodenum}}\otimes \bm{A} \|}_{*}}\\
 =&
\sum\limits_{j=1}^{\nodenum}{\sum\limits_{k=1}^{\nodenum}{\mid(\boldsymbol{\lambda} \otimes
 {\one_{\nodenum}}-{\one_{\nodenum}}\otimes \boldsymbol{\lambda}
 )\transpose{\boldsymbol{e}_{{{\nodenum}^{2}},\nodenum(j-1)+k}} \mid}}\nonumber\\
 =&\sum\limits_{j=1}^{\nodenum}{\sum\limits_{k=1}^{\nodenum}{\mid(\boldsymbol{\lambda} \otimes {\one_{N}}-{\one_{\nodenum}}\otimes \boldsymbol{\lambda} )\transpose({\boldsymbol{e}_{\nodenum,j}}\otimes {\boldsymbol{e}_{\nodenum,k}}) \mid}}\nonumber
\end{align}
Finally, from the properties of the Kronecker product,
\begin{subequations}
\begin{align}
 {{\| \bm{A}\otimes
 {\identity_{\nodenum}}}}&{{-{\identity_{\nodenum}}\otimes \bm{A} \|}_{*}}\\
 =&
\sum\limits_{j=1}^{\nodenum}{\sum\limits_{k=1}^{\nodenum}{\left|
 {{\eval
 }^{\top}}{\boldsymbol{e}_{N,j}}\otimes \one_{\nodenum}^{\top}{\boldsymbol{e}_{N,k}}-\one_{\nodenum}^{\top}{\boldsymbol{e}_{\nodenum,j}}\otimes
 {{\boldsymbol{\lambda}
 }^{\top}}{\boldsymbol{e}_{N,k}} \right|}}\nonumber\\
 =&\sum\limits_{j=1}^{\nodenum}{\sum\limits_{k=1}^{\nodenum}{\left|
 {{{\lambda} }_{j}}\otimes {1}-{1}\otimes {{{\lambda}
 }_{k}} \right|}}
 =\sum\limits_{j=1}^{N}{\sum\limits_{k=1}^{\nodenum}{\left| {{{\lambda} }_{j}}-{{{\lambda} }_{k}} \right|}}. \nonumber
\end{align}
\end{subequations}


\section{Proof of \thref{prop:necessaryfeasibility}}
\label{proof:necessaryfeasibility}

Consider first the following auxiliary result:
\begin{mylemma}
\thlabel{prop:prefshiftconds}
$\shiftmat \in \shiftmatset \cap \prefeasmatset{\projmat}$ iff the following
three conditions simultaneously hold
\begin{subequations}
\label{eq:prefshiftconds}
\begin{align}
&\topologyconstraintmat\vect(\shiftmat) = \bm 0\label{eq:prefshiftcondstopology},\\
&\upermat\transpose\shiftmat\uparmat = \bm
0\label{eq:prefshiftcondsseparate},\\
&\shiftmat = \shiftmat\transpose\label{eq:prefshiftcondssymmetry}.
\end{align}
\end{subequations}
\end{mylemma}

\begin{IEEEproof}
\begin{myitemize}
\myitem\cmt{$\Rightarrow$}The first step is proving that
$\shiftmat \in \shiftmatset \cap \prefeasmatset{\projmat}$
implies \eqref{eq:prefshiftconds}.
\begin{myitemize}%
\myitem\cmt{1,3}Conditions \eqref{eq:prefshiftcondstopology}
and \eqref{eq:prefshiftcondssymmetry}  follow from the
definition of $\shiftmatset$ and \thref{def:prefeas}.
\myitem\cmt{2}To verify \eqref{eq:prefshiftcondsseparate}, note
from  \thref{def:prefeas} that
$\shiftmat\in\prefeasmatset{\projmat}$ iff  $\shiftmat$ satisfies
 \eqref{eq:def:cpref} for some symmetric $\fparmat$ and
$\fpermat$. Multiplying \eqref{eq:def:cpref}  on the right by
$\umat\define [\uparmat,\upermat]$ and on the left by
$\umat\transpose$ yields
\begin{align}\label{eq:blockequations}
\left[	\begin{matrix}\uparmat^{\top}\shiftmat\uparmat& \uparmat^{\top}\shiftmat\upermat\\ \upermat^{\top}\shiftmat\uparmat&\upermat^{\top}\shiftmat\upermat\end{matrix} \right]=\left[	\begin{matrix} \fparmat& \bm{0}\\ \bm{0}&\fpermat\end{matrix} \right]
\end{align}
Condition \eqref{eq:prefshiftcondsseparate} corresponds to the block
$(2,1)$ in this equality.
\end{myitemize}%

\myitem\cmt{$\Leftarrow$}Conversely,
\begin{myitemize}%
\myitem\cmt{shift}if $\shiftmat$
satisfies \eqref{eq:prefshiftconds}, then it follows
from \eqref{eq:prefshiftcondstopology}
and\footnote{Please keep in mind that $\shiftmat\in \shiftmatset$ 
does not imply that $\shiftmat$ is symmetric and
\eqref{eq:prefshiftcondstopology} alone only imposes support constraints on
the upper-triangular entries of $\shiftmat$; cf. Sec.~\ref{sec:filterorderminimization}.} \eqref{eq:prefshiftcondssymmetry} that
$\shiftmat\in \shiftmatset$.
\myitem\cmt{prefeas}To show that
$\shiftmatset\in \prefeasmatset{\projmat}$, one needs to find
symmetric $\fparmat$ and $\fpermat$ such that \eqref{eq:def:cpref} or,
equivalently, \eqref{eq:blockequations} holds. This can be easily
accomplished just by setting $\fparmat=\uparmat^{\top}\shiftmat\uparmat$
and  $\fpermat = \upermat^{\top}\shiftmat\upermat$ since the (1,2) and
(2,1) blocks of  equality \eqref{eq:blockequations} will automatically
hold due to \eqref{eq:prefshiftcondsseparate} and \eqref{eq:prefshiftcondssymmetry}.
\end{myitemize}%


\end{myitemize}%

\end{IEEEproof}

\begin{myitemize}%
\myitem\cmt{matrices satisfying \eqref{eq:prefshiftcondssymmetry}
and \eqref{eq:prefshiftcondstopology}}Let
$\symmatset\define \{\shiftmat\in \rfield^{\nodenum\times \nodenum}:\shiftmat=\shiftmat\transpose\}$
and note
from   \thref{def:prefeas} that
$ \prefeasmatset{\projmat}\subset \symmatset$. Then,
$\shiftmatset \cap \prefeasmatset{\projmat}=
\shiftmatset \cap (\symmatset\cap \prefeasmatset{\projmat})=
(\shiftmatset \cap \symmatset)\cap \prefeasmatset{\projmat}$.

Now, consider the following
parameterization of $\shiftmatset \cap \symmatset$:
\begin{align}\label{topf}
\shiftmatset \cap \symmatset=\left\{\shiftmat=\sum_{i=1}^{\reducededgenum}\auxconst_{i}(\overbrace{\canbasisvec_{n_{i}}\canbasisvec_{n_{i}'}^{\top}+\canbasisvec_{n_{i}'}\canbasisvec_{n_{i}}\transpose}^{\singleedgeshiftmat_i}),\auxconst_i\in{\mathbb{R}}\right\},
\end{align}
which  can be expressed in vector form as
\begin{align}
\label{eq:symmetricshiftspar}
\vect(\shiftmatset \cap \symmatset)&=\left\{\vect(\shiftmat)=\sum_{i=1}^{\reducededgenum}\auxconst_{i}\vect(\singleedgeshiftmat
_{i})\right\}
\\&=\left\{[\vect(\singleedgeshiftmat
_{1}),\cdots,\vect(\singleedgeshiftmat
_{\reducededgenum})]\boldsymbol{\auxconst},{\boldsymbol{\auxconst}}\in{\mathbb{R}}^{\reducededgenum}\right\}.
\nonumber
\end{align}
\myitem\cmt{eqs in \eqref{eq:prefshiftconds}
combined}Combining \eqref{eq:prefshiftconds}
and \eqref{eq:symmetricshiftspar} yields
\begin{align}
\shiftmatset \cap \prefeasmatset{\projmat}=
\left\{\mathrm{vec^{-1}}(\singleedgeshiftmat \boldsymbol{\auxconst})~ \forall
\boldsymbol{ \auxconst}:
\upermat\transpose\mathrm{vec^{-1}}(\singleedgeshiftmat \boldsymbol{\auxconst}) \uparmat
= \bm 0
\right\}.
\end{align}

\myitem\cmt{dimension}Since the columns of $\singleedgeshiftmat $ are linearly independent, 
\begin{align}
\dim(\shiftmatset \cap \prefeasmatset{\projmat})
=& \mathrm{dim}\{\boldsymbol{\auxconst}:\upermat^{\top}\mathrm{vec^{-1}}(\singleedgeshiftmat \boldsymbol{\auxconst})\uparmat=\bm{0}\}\nonumber\\=&\mathrm{dim}\{\boldsymbol{\auxconst}:(\uparmat^{\top}\otimes\upermat^{\top})\singleedgeshiftmat \boldsymbol{\auxconst}=\bm{0}\} \nonumber.
\end{align}
Since
$(\uparmat^{\top}\otimes\upermat^{\top})\singleedgeshiftmat \in{\mathbb{R}}^{\subspacedim(\nodenum-\subspacedim)\times\reducededgenum}$,
it follows that
$\mathrm{dim}(\shiftmatset\cap\prefeasmatset{\projmat})=\reducededgenum-\mathrm{rank}((
\uparmat^{\top}\otimes\upermat^{\top})\singleedgeshiftmat )$.

\end{myitemize}%

\section{The Tightness  of the Relaxed Solution}
\label{sec:tightness}
{

\cmt{overview}This appendix further justifies why the objective of the
relaxed problem \eqref{prob:exactrelaxed} is a reasonable surrogate
for the objective of the original problem \eqref{prob:exact}.
\cmt{Notation}To simplify notation, some of the symbols used earlier will be reused here.

     \begin{figure}[t]
     \centering
         \includegraphics[width=0.45\textwidth]{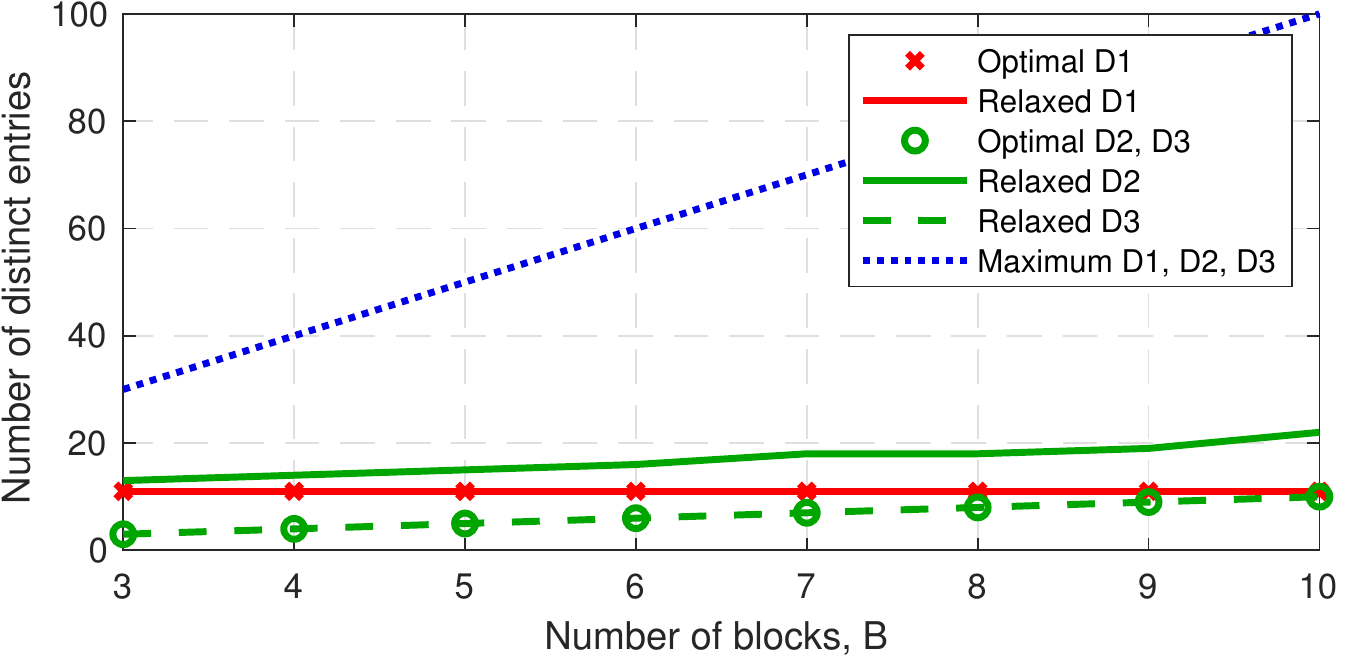}
   \caption{\arev{Comparison between the solution of \eqref{eq:kronl1} and
     that of $ \inf \{
\noevalsfun(\bm x):\bm A \bm x = \bm b \}$  for $M=N=10$. }}
         \label{fig:tightness}
     \end{figure}

\cmt{problem}To focus on the fundamental aspects, consider 
the optimization in \eqref{prob:exactrelaxed} for a fixed
$\fpermat$. The resulting problem is a special case of
\begin{subequations}
\label{prob:frelaxed}
\begin{align}
\minimize_{\fparmat}~~&||\fparmat\otimes \identity_\subspacedim
-\identity_\subspacedim \otimes \fparmat ||_\star\\
\st~~&\bm A_\parallel\vect({\fparmat}) = \bm b_\parallel
\end{align}
\end{subequations}
for some $\bm A_\parallel$ and $\bm b_\parallel$. Ideally, one would like to see how well
the solution to \eqref{prob:frelaxed} approximates the solution to 
\begin{subequations}
\label{prob:fnonrelaxed}
\begin{align}
\minimize_{\fparmat}~~&\noevalsfun(\fparmat)\\
\st~~&\bm A_\parallel\vect({\fparmat}) = \bm b_\parallel.
\end{align}
\end{subequations}
To this end, apply the eigendecomposition
$\fparmat=\fparevecmat\evalparmat\fparevecmat\transpose$ to rewrite
the solution of \eqref{prob:frelaxed} as
\begin{subequations}
\label{prob:frelaxedeigen2}
\begin{align}
\inf_{\text{orthogonal }\fparevecmat} \left[
\begin{array}{cc}
\displaystyle\inf_{\text{diagonal }\evalparmat}||\evalparmat\otimes \identity_\subspacedim
-\identity_\subspacedim \otimes \evalparmat ||_\star\\
\st ~\bm
A_\parallel\vect(\fparevecmat\evalparmat\fparevecmat\transpose) = \bm
b_\parallel\\
\end{array}
\right].
\end{align}
\end{subequations}
The inner problem, which captures the essence of the question to be addressed
here, is a special case of
\begin{subequations}
\label{eq:kronl1}
\begin{align}
\minimize_{\bm x}~~&\|\bm x \otimes \bm 1 - \bm 1 \otimes \bm x\|_1\\
\st~~&\bm A \bm x = \bm b
\end{align}
\end{subequations}
for some $\bm A$ and $\bm b$.  Thus, the ability
of \eqref{prob:frelaxed} to promote solutions with a reduced number of
distinct eigenvalues can be understood by analyzing how
well \eqref{eq:kronl1} promotes solutions with a reduced number of
distinct entries. With $\bm x^*$ denoting a minimizer
of \eqref{eq:kronl1} and with $\noevalsfun(\bm x)$ denoting the number
of distinct entries of $\bm x$, one is therefore interested in
analyzing how far $\noevalsfun(\bm x^*)$ is from 
$\noevalsfun^*\define \inf \{
\noevalsfun(\bm x):\bm A \bm x = \bm b \}$ for given $\bm A$ and $\bm
b$. Unfortunately, this is a challenging comparison since obtaining
$\noevalsfun^*$ generally entails  combinatorial complexity. However,
$\noevalsfun^*$ can be obtained exactly for certain families of $(\bm
A, \bm b)$. Specifically, the rest of this section will compare
$\noevalsfun(\bm x^*)$ and $\noevalsfun^*$ when 
$\bm A$  is drawn from 3 probability distributions.

Let $\bm A$ comprise $B$ block columns:
\begin{align}
\bm A \define
\left[
\begin{array}{c c c}
\bm A_1&\ldots&\bm A_B\\
\bm E_1&\ldots&\bm E_B\\
\bm 1\transpose &\ldots&\bm 1\transpose
\end{array}
\right].
\end{align}
Let also $\bm A_i\define[\bm a_{i,1},\ldots, \bm
  a_{i,N}]\in \rfield^{M \times N}$ and consider the following distributions:
  \begin{itemize}
  \item Distribution \textbf{$D_1$}:
 $\{\bm E_i\}_{i=1}^B$ are empty  and
  $\{\bm a_{i,j}\}_{i,j}$ are drawn independently from a
  continuous distribution.\footnote{Formally, a continuous
  distribution in this context is a distribution that is absolutely
  continuous with respect to Lebesgue
  measure~\cite{billingsley}.}
  \item Distribution \textbf{$D_2$}:
  $\{\bm a_{i,j}\}_{i}$ are drawn independently from a
  continuous distribution for $j=1,\ldots,N-1$ and $\bm
  a_{i,N}\define-\sum_{j=1}^{N-1} \bm a_{i,j}$.
  On the other hand,
  $\bm E_i=(-1/i)\bm e_i \bm e_1\transpose \in \rfield^{B-1 \times
  N}$ for $i=1,\ldots, B-1$, where $\bm e_i$ is, with an abuse of notation, the $i$-th
  column of the identity matrix with appropriate dimension. For $i=B$,
  $\bm E_B= (1/i)\bm 1\bm
  e_1\transpose$.

\item Distribution \textbf{$D_3$}:
 $\{\bm a_{i,j}\}_{i,j}$ are as in $D_2$. Matrix $\bm E_i$ is $\bm
   E_i=(-1/i)\bm e_i \bm 1\transpose \in \rfield^{B-1 \times N}$ for
   $i=1,\ldots, B-1$ and $\bm E_B= (1/i)\bm 1\bm
   1\transpose$ for $i=B$. 
  \end{itemize}

With these distributions, we have the following:
\begin{mytheorem}
\thlabel{th:tightness}
Let $\noevalsfun^*\define \inf \{
\noevalsfun(\bm x):\bm A \bm x = \bm b \}$ and let
$\bm b$ be a vector of the appropriate dimensions with all zeros
  except for the last entry, which equals $N\sum_{i=1}^B i =
  NB(B+1)/2$.
  \begin{enumerate}[label=(\alph*)]
  \item If $\bm A\sim D_1$ with $NB > M$, then
  $\noevalsfun^*=M+1$ with probability 1.
  \item If $\bm A\sim D_i$, $i\in {2,3}$,  with $M \geq B$, then
  $\noevalsfun^*=B$ with probability 1. 
  \end{enumerate}
\end{mytheorem}

\begin{IEEEproof}
See Appendix~\ref{sec:proof:tightness} in the supplementary material.
\end{IEEEproof}

Fig.~\ref{fig:tightness} compares  $\noevalsfun(\bm x^*)$ and
$\noevalsfun^*\define \inf \{
\noevalsfun(\bm x):\bm A \bm x = \bm b \}$, where $\bm x^*$ is
found by applying a convex solver to \eqref{eq:kronl1}. Each point is
obtained for a realization of $\bm A$. Although we have not proved it
formally, it seems that the relaxed solution $\bm x^*$ coincides with
the exact one with probability 1 when $\bm A$ is drawn from $D_1$ or
$D_3$. For $D_2$, the number of distinct entries of $\bm x^*$ is
higher than optimal, but in any case it is considerably lower than the
number $NB$ of entries of $\bm x$,  shown as ``Maximum''.

To sum up, the solution to the relaxed problem \eqref{eq:kronl1} is
optimal in some cases and close to the optimal in other tested
cases. This supports the choice of  $||\fparmat\otimes \identity_\subspacedim
-\identity_\subspacedim \otimes \fparmat ||_\star$ as a convex
surrogate for the number of distinct eigenvalues of~$\fparmat$. 

}

\clearpage

\textbf{SUPPLEMENTARY MATERIAL}

\section{Iterative Solver for \eqref{prob:exactrelaxed}}
\label{sec:exactsolver}
\input{not_main/admm-exact.tex}

\section{Iterative Solver for \eqref{prob:approxrelaxed}}
\label{sec:approxsolver}
\input{not_main/admm-approximate.tex}

\section{Proof of \thref{th:tightness}}
\label{sec:proof:tightness}

{

Given any vector $\bm x\define [x_1,\ldots, x_{NB}]\transpose$, one
can collect its $L\define\noevalsfun(\bm x)$ distinct entries into the
vector $\bbm z\define [\bar z_1,\ldots, \bar z_L]\transpose$ and
construct a partition map $\ell:\{1,\ldots,
NB\} \rightarrow \{1,\ldots, L\}$ such that $ x_i=
\bar z_{\ell(i)}$. With this notation and letting $\bm a_i$ denote the
$i$-th column of $[\bm A_1,\ldots, \bm A_B]$, one can express the
first block row of $\bm A \bm x$ as $ \sum_{l=1}^{L} \bbm a_l \bar
z_l$, where $\bbm a_l\define \sum_{i:\ell(i)=l} \bm a_i,~l=1,\ldots,L$. Thus, the
system $\bm A\bm x = \bm b$ admits a solution iff there exist $\bm x$,
$\bbm z$, and $\ell$ satisfying the aforementioned relations such that
$ \bbm A \bbm z=\bm 0$, $\sum_i \bm E_i \bm x_i = \bm 0$ and $\bm
1\transpose \bm x=NB(B+1)/2$, where $\bbm A\define [\bbm
a_1,\ldots, \bbm a_L]$ and $\bm x_i\in \rfield^N$, $i=1,\ldots,B$, is
such that $\bm x=[\bm x_1\transpose,\ldots,\bm
x_B\transpose]\transpose$.

To prove part (a), observe that, regardless of the partition $\ell$,
the fact that the columns of $\bm A_i$ are independently drawn from a
continuous distribution implies that $\{\bbm a_l\}_l$ also adhere to a
continuous distribution and are independent of each other. Thus
$\rank(\bbm A) = \min(M,L)$ with probability 1. A vector $\bm x$
satisfying $\bm A\bm x=\bm b$ exists iff the homogeneous system $ \bbm
A \bbm z=\bm 0$ admits at least a non-zero solution. This is because
any non-zero solution  $\bbm z$ can be normalized to satisfy $\bm
1\transpose \bm x=NB(B+1)/2$ and thus $\bm A\bm x=\bm b$. Given
that $\rank(\bbm A) = \min(M,L)$ with probability 1, the system $ \bbm
A \bbm z=\bm 0$ admits a non-zero solution iff $L>M$. The proof is
concluded by noting that the minimum $L$ satisfying this condition is
$M+1$. The hypothesis $NB>M$ is necessary because $L\leq NB$.

To prove part (b) let $\cbm A\define[\bm a_{1,1},\ldots, \bm
a_{1,N-1},\bm a_{2,1},\ldots,\bm a_{B,N-1}]$ collect the independent
vectors in $\{\bm a_{i,j}\}_{i,j}$ and note that $[\bm A_1,\ldots, \bm
A_B] = \cbm A \bm T$, where
\begin{align*}
\bm T \define
\left[
\begin{array}{ccccccc}
\bm I_{N-1}&\bm -\bm 1&\bm 0&\bm 0&\ldots&\bm 0&\bm 0\\
\bm 0&\bm 0&\bm I_{N-1}&\bm -\bm 1&\ldots&\bm 0&\bm 0\\
\vdots&\vdots&\vdots&\vdots&\ddots&\vdots&\vdots\\
\bm 0&\bm 0&\bm 0&\bm 0&\ldots&\bm I_{N-1}&\bm -\bm 1
\end{array}
\right].
\end{align*}
Thus, the system $\bm A\bm x = \bm b$ admits a solution iff there
exists $\bm x$ such that $ \cbm A \bm T \bm x=\bm 0$, $\sum_i \bm
E_i \bm x_i = \bm 0$ and $\bm 1\transpose \bm x=NB(B+1)/2$.  Observe
that, by setting $\bm x_i = i\bm 1~\forall i$, the resulting $\bm x$
satisfies these three conditions, which establishes that
$\noevalsfun^*\leq \noevalsfun(\bm x) = B$.  Besides, it can be easily
seen that this solution is the only that satisfies $ \bm T \bm x=\bm
0$, $\sum_i \bm E_i \bm x_i = \bm 0$, and $\bm 1\transpose \bm x=NB(B+1)/2$. In particular, this implies that
$ \cbm A \bm T \bm x=\bm 0$ holds regardless of the value of $\cbm
A$. It remains only to show that, with probability 1, there exists no
solution $\bm x'$ with $\noevalsfun(\bm x') < B$.

To this end, note that the set of realizations of $\cbm A$ for which
such an $\bm x'$ exists is the union over partitions $\ell$ of the
sets of realizations of $\cbm A$ for which a solution $\bm x'$ with
partition $\ell$ exists. Since there are only finitely many
partitions, it suffices to show that the probability of finding such
an $\bm x'$ is 0 for a
single generic partition. Let $\ell$ denote such a partition and
observe that $\bm x$ can be expressed as $\bm x=\bm S_\ell \bbm z$ for
some $NB\times L$ matrix $\bm S_\ell $ with ones and zeros. The
probability that there exists a solution is $
\mathbb{P}[\exists \bbm z \in \mathcal{Z}_\ell: \cbm A \bm T \bm 
S_\ell \bbm z = \bm 0] $ where $\mathcal{Z}_\ell\define \{\bbm z:~\bm
x=\bm S_\ell \bbm z$ satisfies $
\sum_i \bm
E_i \bm x_i = \bm 0$ and $\bm 1\transpose \bm x=NB(B+1)/2\}$.  If
$L<B$, the product $ \bm T \bm S_\ell \bbm z$ is necessarily non-zero
for $\bbm z \in \mathcal{Z}_\ell$,
as described earlier. Thus, the aforementioned probability equals
\begin{subequations}
\begin{align}
&\mathbb{P}[\exists \bbm z \in \mathcal{Z}_\ell: \cbm A \bm T \bm 
S_\ell \bbm z = \bm 0, \bm T \bm 
S_\ell \bbm z \neq \bm 0 ] \\
&\leq \mathbb{P}[\exists \bbm t \in \range{ \bm T \bm 
S_\ell}: \cbm A\bbm t = \bm 0, \bbm t \neq \bm 0 ]\\
&= \mathbb{P}[\exists \bm v \in \rfield^r: \cbm A\bm Q_\ell \bm v = \bm 0, \bm v \neq \bm 0 ]\\
&= \mathbb{P}[\rank[ \cbm A\bm Q_\ell]< r ],
\end{align}
\end{subequations}
where $\bm Q_\ell$ is a matrix whose  $r$  linearly independent columns
constitute a basis for $\range{ \bm T \bm 
S_\ell}$. Noting that $r\leq L<B$ and, by hypothesis, $M\geq B$, it
follows that $\rank[ \cbm A\bm Q_\ell]< r$ is only possible if the $M$
rows of $\cbm A$ lie in a proper subspace of  $\range{ \bm T \bm 
S_\ell}$. Since proper subspaces have zero Lebesgue measure and $\cbm
A$ obeys a continuous probability distribution, it follows that 
$\mathbb{P}[\rank[ \cbm A\bm Q_\ell]< r ]=0$ and, consequently, 
$
\mathbb{P}[\exists \bbm z \in \mathcal{Z}_\ell: \cbm A \bm T \bm 
S_\ell \bbm z = \bm 0]=0$.

}

\end{document}

%% file: not_main/include_v3.tex
\usepackage{savesym} 
\savesymbol{labelindent}
\usepackage{enumitem} 
\newcommand{\arev}[1]{\textcolor{blue}{#1}} 

\if\editmode1  
\usepackage[backend=bibtex,style=alphabetic,sorting=debug,maxbibnames=99]{biblatex}
\DeclareFieldFormat{labelalpha}{\thefield{entrykey}}
\DeclareFieldFormat{extraalpha}{}
\bibliography{\bibfilenames}
\newcommand{\cmt}[1]{\noindent\textcolor{lightgreen}{\underline{[#1]}}} 
\newcommand{\hc}[1]{\textcolor{blue}{#1}} 

\setlistdepth{9}
\newlist{mydeepitemize}{enumerate}{9}
\setlist[mydeepitemize,1]{label=$\bullet$}
\setlist[mydeepitemize,2]{label=$\diamond$}
\setlist[mydeepitemize,3]{label=$\rightarrow$}
\setlist[mydeepitemize,4]{label=$\circ$}
\setlist[mydeepitemize,5]{label=$-$}
\setlist[mydeepitemize,6]{label=$\Delta$}
\setlist[mydeepitemize,7]{label=$\star$}
\setlist[mydeepitemize,8]{label=$\square$}
\setlist[mydeepitemize,9]{label=$\checkmark$}
\newenvironment{myitemize}{\begin{mydeepitemize}}{\end{mydeepitemize}}

\newcommand{\myitem}{\item}

\else
\usepackage{cite}
\newcommand{\cmt}[1]{} 
\newcommand{\hc}[1]{\textcolor{black}{#1}} 
\newenvironment{myitemize}{}{}
\newcommand{\myitem}{}

\fi

\newcommand{\printmybibliography}{
\if\editmode1 
\printbibliography
\else
\bibliography{\bibfilenames}
\fi
}

\usepackage{fixltx2e}

\usepackage{graphicx}

\usepackage{subcaption}              

\usepackage[utf8]{inputenc}

\usepackage{amsfonts}
\usepackage{amsmath}
\usepackage{mathtools}

\usepackage{amssymb}

\usepackage{bm}

\usepackage{color,verbatim}
\usepackage{multirow}
\usepackage{accents}

\usepackage{theoremref}

\usepackage{url}

\newcounter{rulecounter}
\newcommand{\resetrule}{ \setcounter{rulecounter}{0}}
\resetrule

\newsavebox{\selvestebox}
\newenvironment{colbox}[1]
  {\newcommand\colboxcolor{#1}%
   \begin{lrbox}{\selvestebox}%
   \begin{minipage}{\dimexpr\columnwidth-2\fboxsep\relax}}
  {\end{minipage}\end{lrbox}%
   \begin{center}
   \colorbox{\colboxcolor}{\usebox{\selvestebox}}
   \end{center}}

\definecolor{orange}{rgb}{1,0.8,0}
\definecolor{gray}{rgb}{.9,0.9,0.9}
\definecolor{darkgray}{rgb}{.3,0.3,0.3}
\definecolor{darkblue}{rgb}{.1,0.0,0.3}
\definecolor{lightblue}{rgb}{0.7,0.7,1}
\definecolor{lightred}{rgb}{1,0.7,.7}
\definecolor{purple}{RGB}{204,153,255}
\definecolor{lightgray}{rgb}{.95,0.95,0.95}
\definecolor{lightgreen}{rgb}{0.3,0.5,0.3}
\definecolor{darkgreen}{rgb}{0.05,0.3,0.05}

\newcommand{\ra}{$\rightarrow$~}

\newcommand{\bbm}[1]{{\bar{\bm #1}}}

\newcommand{\cbm}[1]{{\check{\bm #1}}}

\newcommand{\inv}{^{-1}}

\newcommand{\rfield}{\mathbb{R}}
\newcommand{\vect}{\mathop{\rm vec}}

\newcommand{\diagnb}{\mathop{\rm diag}}

\newcommand{\rank}{\mathop{\rm rank}}

\newcommand{\range}[1]{\mathcal{R}\brackets{#1}}

\newcommand{\trnb}{\mathop{\rm Tr}}

\newcommand{\transpose}{^\top}
 \newcommand{\define}{:=}

\newcommand{\expectednb}{\mathop{\textrm{E}}\nolimits}

\newcommand{\minimize}{\mathop{\text{minimize}}}

\newcommand{\st}{\mathop{\text{s.t.}}}

\DeclareMathOperator*{\argmin}{arg\,min}

\newtheorem{myproposition}{Proposition}
\newtheorem{myremark}{Remark}
\newtheorem{myproblemstatement}{Problem Statement}
\newtheorem{mylemma}{Lemma}
\newtheorem{mytheorem}{Theorem}
\newtheorem{mydefinition}{Definition}
\newtheorem{mycorollary}{Corollary}

%% file: not_main/definitions.tex
\renewcommand{\expectednb}{\mathbb{\hc{E}}}
\renewcommand{\range}[1]{\mathcal{R}\{#1\}}
\newcommand{\rangeperp}[1]{\mathcal{R}^\perp\{#1\}}

\renewcommand{\vect}{{\rm vec}}
\newcommand{\colspan}{{\hc{\mathcal{R}}}}
\renewcommand{\transpose}{^\top}
\newcommand{\canbasisvec}{{\hc{\bm e}}} 

\renewcommand{\emptyset}{\varnothing}
\newcommand{\graph}{{\hc{\mathcal{G}}}}
\newcommand{\nodeset}{{\hc{\mathcal{V}}}}
\newcommand{\edgeset}{{\hc{\mathcal{E}}}}
\newcommand{\reducededgeset}{{\hc{\bar{\mathcal{E}}}}} 
\newcommand{\subspace}{{\hc{\mathcal{U}}}}

\newcommand{\shiftmatset}{{\hc{\mathcal{S}}_\graph}} 
\newcommand{\polfeasmatset}[1]{{\hc{\mathcal{F}}_{#1}}}
\newcommand{\prefeasmatset}[1]{{\hc{\tilde{\mathcal{F}}}_{#1}}}

\newcommand{\neighborhood}[1]{{\hc{\mathcal{N}}_{#1}}} 
\newcommand{\obsvec}{{\hc{\bm z}}} 
\newcommand{\obsentry}{{\hc{z}}} 
\newcommand{\projobsvec}{{\hc{\bm y}}} 
\newcommand{\coord}{{\hc{\alpha}}}
\newcommand{\coordvec}{{\bm{\coord}}}
\newcommand{\projmat}{{\hc{\bm P}}}
\newcommand{\uparmat}{{\hc{\bm U}}_{\parallel}}
\newcommand{\fparmat}{{\hc{\bm F}}_{\parallel}}
\newcommand{\fparevecmat}{{\hc{\bm Q}}_{\parallel}}
\newcommand{\evalparmat}{{\hc{\bm \Lambda}}_{\parallel}}
\newcommand{\shiftparmat}{{\shiftmat}_{\parallel}}
\newcommand{\upermat}{{\hc{\bm U}}_{\perp}}
\newcommand{\fpermat}{{\hc{\bm F}}_{\perp}}
\newcommand{\fperevecmat}{{\hc{\bm Q}}_{\perp}}
\newcommand{\evalpermat}{{\hc{\bm \Lambda}}_{\perp}}
\newcommand{\shiftpermat}{{\shiftmat}_{\perp}}
\newcommand{\shiftmat}{{\hc{\bm S}}}
\newcommand{\filtermat}{{\hc{\bm H}}}
\newcommand{\shiftentry}{{\hc{s}}}
\newcommand{\vandermondemat}{{\hc{\bm \Psi}}}

\newcommand{\dummymat}{{\hc{\bm Q}}}

\newcommand{\auxconst}{{\hc{\kappa}}}

\newcommand{\auxvec}{{\hc{\bm x}}}

\newcommand{\nodenot}[1]{_{#1}}  
\newcommand{\nodenodenot}[2]{_{#1,#2}}  
\newcommand{\itnot}[1]{^{(#1)}} 
\newcommand{\itnodenot}[2]{\itnot{#1}\nodenot{#2}} 

\newcommand{\itind}{{\hc{l}}} 
\newcommand{\itnum}{{\hc{L}}} 

\newcommand{\nodeind}{{\hc{n}}} 
\newcommand{\nodenum}{{\hc{N}}} 
\newcommand{\reducededgenum}{{\hc{\bar{E}}}} 

\newcommand{\filtercoef}{{\hc{c}}} 
\newcommand{\filtercoefvec}{{\hc{\bm c}}} 

\newcommand{\orderfun}[1]{{\hc{\mathcal{O}}_{#1}}} 
\newcommand{\noevalsfun}{{\hc{\mathcal{L}}}} 

\newcommand{\regpar}{{\hc{\lambda}}} 
\newcommand{\subspacedim}{{\hc{r}}} 

\newcommand{\topologyconstraintmat}{{\hc{\bm{W}}}}

\newcommand{\exactpfeasmat}{{\hc{\bm C}}} 
\newcommand{\exactpfeasvec}{{\hc{\bm b}'}}  

\newcommand{\dlmsstepsize}{{\hc{\mu}_\text{DLMS}}}%
\newcommand{\dlmsauglagpar}{{\hc{\rho}_\text{DLMS}}}%

\newcommand{\singleedgeshiftmat}{{\hc{\bm \Phi}}}

\newcommand{\linearmat}{{\hc{\bm B}}}
\newcommand{\eval}{{\hc{\bm \lambda}}}
\newcommand{\feasmat}{{\hc{\bm D}}}
\newcommand{\symmat}{{\hc{\bm G}}}
\newcommand{\symmatset}{{\hc{\mathbb{S}}}} 
\newcommand{\umat}{{\hc{\bm U}}}
\newcommand{\uvec}{{\hc{\bm u}}}
\newcommand{\uset}{{\hc{\mathcal{U}}}}
\newcommand{\ADMMreg}{{\hc{\rho}}} 
\newcommand{\eigenreg}{{\hc{\epsilon}}}
\newcommand{\linfunmat}{{\hc{\bm A}}}
\newcommand{\linfunvec}{{\hc{\bm a}}}
\newcommand{\linearmatperp}{{\hc{\bm B}}_{\perp}}

\newcommand{\linfunmatpar}{{\hc{\bm A}}_{\parallel}}
\newcommand{\identity}{{\hc{\bm I}}}
\newcommand{\one}{{\hc{\bm 1}}}
\newcommand{\noise}{{\hc{\bm v}}}
\newcommand{\noisepow}{{\hc{\sigma^2}}}
\newcommand{\desiredvec}{{\hc{\bm \xi}}}

\newcommand{\linearvec}{{\hc{\bm b}}}
\newcommand{\symmetryperp}{{\hc{\bm G}}_{\perp}}
\newcommand{\symmetrypar}{{\hc{\bm G}}_{\parallel}}
\newcommand{\regnucper}{{\hc{ \eta}}_{\perp}}
\newcommand{\regnucpar}{{\hc{ \eta}}_{\parallel}}

\newcommand{\conspermat}{{\hc{\bm T}}_{\perp}}
\newcommand{\consparmat}{{\hc{\bm T}}_{\parallel}}
\newcommand{\consmat}{{\hc{\bm T}}}
\newcommand{\nucpermat}{{\hc{\bm Y}}_{\perp}}
\newcommand{\nucparmat}{{\hc{\bm Y}}_{\parallel}}
\newcommand{\nucmat}{{\hc{\bm Y}}}
\newcommand{\lagfirst}{{\hc{\bm q}}_{1}}
\newcommand{\lagsec}{{\hc{\bm q}}_{2}}
\newcommand{\lagtre}{{\hc{\bm q}}_{3}}
\newcommand{\canbasisperp}{{\hc{\bm t}}} 

\newcommand{\auxevecmat}{{\hc{\bm V}}}
\newcommand{\auxevalmat}{{\hc{\bm \Lambda}}}

\newcommand{\loc}{{\hc{x}}}
\newcommand{\locvec}{{\hc{\bm x}}}
\newcommand{\sensorlocvec}{{\hc{\bm x}}}
\newcommand{\regiondim}{{\hc{d}}}
\newcommand{\signalfun}{{\hc{\xi}}}
\newcommand{\sourcelocvec}[1]{{\hc{\bm x_{\text{s},#1}}}}
\newcommand{\length}{{\hc{X}}}

%% file: not_main/admm-exact.tex

This appendix develops an iterative method to
solve \eqref{prob:exactrelaxed} based on ADMM.  To this end, the first
step is to rewrite the objective and constraints in a form that is amenable to the application of this method. 

\cmt{Objective}To rewrite the objective, let $\canbasisvec_i$ and $\canbasisperp_j$ be respectively the $i$-th and $j$-th columns of $\bm I_\subspacedim$ and $\bm I_{\nodenum-\subspacedim}$. 
The first and  second norms in the objective of \eqref{prob:exactrelaxed} can be expressed as:
\begin{align}
{{\left\| \fparmat\otimes\bm{I}_{\subspacedim}-\bm{I}_\subspacedim\otimes\fparmat \right\|}_{*}}\nonumber &= {{\left\| \mathrm{vec}^{-1}(\linfunmat\mathrm{vec} (\fparmat)) \right\|}_{*}}\\
{{\left\| \fpermat\otimes\bm{I}_{\nodenum-\subspacedim}-\bm{I}_{\nodenum-\subspacedim}\otimes\fpermat \right\|}_{*}}\nonumber & =  {{\left\| \mathrm{vec}^{-1}(\linearmat\mathrm{vec} (\fpermat) )\right\|}_{*}},
\end{align}
  where $\linfunmat\define [\linfunvec_{11}, \linfunvec_{21},...,\linfunvec_{\subspacedim\subspacedim}]$, $\linearmat\define[\linearvec_{11}, \linearvec_{21},...,\linearvec_{\nodenum-\subspacedim,\nodenum-\subspacedim}]$, $\linfunvec_{ij}\define\mathrm{vec}({\canbasisvec_{i}\canbasisvec_{j}^{\top}\otimes\bm{I}_{\subspacedim}-\bm{I}_{\subspacedim}\otimes\canbasisvec_{i}\canbasisvec_{j}^{\top}})$, $\linearvec_{ij}\define\mathrm{vec}({\canbasisperp_{i}\canbasisperp_{j}^{\top}\otimes\bm{I}_{\nodenum-\subspacedim}-\bm{I}_{\nodenum-\subspacedim}\otimes\canbasisperp_{i}\canbasisperp_{j}^{\top}})$.

\cmt{constraints}To rewrite the  constraints of \eqref{prob:exactrelaxed},
\begin{myitemize}%
\myitem\cmt{1-2}invoke  the properties of the Kronecker product to
combine  the  first and  second constraints as $\topologyconstraintmat((\uparmat\otimes\uparmat)\mathrm{vec}(\fparmat)+(\upermat\otimes\upermat)\mathrm{vec}(\fpermat))=\bm{0}$.
\myitem\cmt{3-4}The third and fourth constraints can be rewritten as
$\symmetrypar\mathrm{vec}(\fparmat)=\bm{0}$,
$\symmetryperp\mathrm{vec}(\fpermat)=\bm{0}$ where the rows of
$\symmetrypar\in\mathbb{R}^{{(\subspacedim^{2}-\subspacedim)/2\times{\subspacedim}^{2}}}$
and $\symmetryperp\in\mathbb{R}^{{((\nodenum-\subspacedim)^{2}-(\nodenum-\subspacedim))/2\times{(\nodenum-\subspacedim)}^{2}}}$  are respectively given by $(\canbasisvec_{j}^{\top}\otimes\canbasisvec_{i}^{\top}-\canbasisvec_{i}^{\top}\otimes\canbasisvec_{j}^{\top})$ and $(\canbasisperp_{j}^{\top}\otimes\canbasisperp_{i}^{\top}-\canbasisperp_{i}^{\top}\otimes\canbasisperp_{j}^{\top})$ for $i<j$.
\myitem\cmt{trace}Regarding the trace constraints, rewrite $\mathrm{tr}(\fparmat)=\subspacedim$ as $\mathrm{tr}(\bm{I}_\subspacedim\fparmat)=\subspacedim$,  which in turn can be expressed as $\mathrm{vec}^{\top}(\bm{I}_\subspacedim)\mathrm{vec}(\fparmat)=\subspacedim$. Similarly, one can rewrite $\mathrm{tr}(\fpermat)=(1\pm\epsilon)(\nodenum-\subspacedim)$ as $\mathrm{vec}^{\top}(\bm{I}_{\nodenum-\subspacedim})\mathrm{vec}(\fpermat)=(1\pm\epsilon)(\nodenum-\subspacedim)$.
\end{myitemize}%

\cmt{rewritten problem}Therefore, \eqref{prob:exactrelaxed} can be written as:
 \begin{subequations} \label{eq:app} 
\begin {align}
\nonumber \underset{\nucpermat,\nucparmat,\fparmat,\fpermat}{\text{min}} ~~&
          \regnucpar {{\left\| \nucparmat \right\|}_{*}}
           +\regnucper{{\left\|\nucpermat\right\|}_{*}} +{{\left\|\fpermat \right\|}_\text{F}^{2}}\label{optt1}\\
\text{s. t.}   \;\;\;\;    	&  \consparmat\mathrm{vec}(\fparmat)+\conspermat\mathrm{vec}(\fpermat)=\linearvec\\
&\mathrm{vec}(\nucparmat)=\linfunmat\mathrm{vec} (\fparmat)\\
&\mathrm{vec}(\nucpermat)=\linearmat\mathrm{vec} (\fpermat),
\end{align}
\end{subequations}
where
\begin{subequations}
\begin{align}
\consparmat&\overset{\Delta}{=}[ \topologyconstraintmat(\uparmat\otimes\uparmat);
									\symmetrypar;
                                                           \bm{0};
	                                                    \mathrm{vec^{\top}}(\bm{I}_\subspacedim);
                                                           \bm{0}]
                                                           \\
\conspermat&\overset{\Delta}{=}[
                                                        \topologyconstraintmat(\upermat\otimes\upermat);
                                                         \bm{0};
									\symmetryperp;
	                                                    \bm{0};
                                                           \mathrm{vec^{\top}}(\bm{I}_{\nodenum-\subspacedim})]\\
\linearvec&\overset{\Delta}{=}[
                                                          \bm{0};
                                                         \bm{0};
									\bm{0};
	                                                    \subspacedim;
                                                             (1\pm\epsilon)(\nodenum-\subspacedim)].
\end{align}
\end{subequations}

\cmt{ADMM}The ADMM method
in  scaled form \cite[Sec. 3.1.1]{boyd2011distributed} applied to \eqref{eq:app} obtains the $k$-th iterate as follows:
 \begin{subequations} \label{ADMM1}
\begin{align}
&(\nucparmat^{k+1},\nucpermat^{k+1})\label{ADMM1a}\\& \quad\quad :=\underset{\nucparmat,\nucpermat}{\text{argmin}}~~L_{\rho}(\nucparmat,\nucpermat,\fparmat^{k},\fpermat^{k},\lagfirst^{k},\lagsec^{k},\lagtre^{k})\nonumber\\%
\label{ADMM1ab}&(\fparmat^{k+1},\fpermat^{k+1})\\& \quad\quad :=  \underset{\fparmat,\fpermat}{\text{argmin}}~~ L_{\rho}(\nucparmat^{k+1},\nucpermat^{k+1},\fparmat, \fpermat,\lagfirst^{k},\lagsec^k, \lagtre^{k}) \nonumber\\
&\lagfirst^{k+1}:=\lagfirst^{k}+\consparmat\mathrm{vec}(\fparmat^{k+1})+\conspermat\mathrm{vec}(\fpermat^{k+1})-\linearvec\\
&\lagsec^{k+1}:=\lagsec^{k}+\mathrm{vec}(\nucparmat^{k+1})-\linfunmat\mathrm{vec} (\fparmat^{k+1}) \\
&\lagtre^{k+1}:=\lagtre^{k}+\mathrm{vec}(\nucpermat^{k+1})-\linearmat\mathrm{vec} (\fpermat^{k+1}), 
\end{align}
\end{subequations}
where
\begin{align}
 L_{\rho}&(\nucparmat, \nucpermat,\fparmat,\fpermat,\lagfirst,\lagsec,\lagtre) \nonumber\\\define&
\regnucpar{{\left\| \nucparmat \right\|}_{*}}+\regnucper{{\left\| \nucpermat \right\|}_{*}}+{{\left\| \fpermat \right\|}_\text{F}^{2}}\nonumber\\
 &+(\rho/2){{\left\| \consparmat\mathrm{vec}(\fparmat)+\conspermat\mathrm{vec}(\fpermat)-\linearvec+\lagfirst \right\|}_{2}^{2}}\nonumber\\
&+(\rho/2){{\left\| \mathrm{vec}(\nucparmat)-\linfunmat\mathrm{vec}
 (\fparmat)+\lagsec\right\|}_{2}^{2}}
 \nonumber\\
 & +(\rho/2){{\left\| \mathrm{vec}(\nucpermat)-\linearmat\mathrm{vec}
 (\fpermat)+\lagtre\right\|}_{2}^{2}}\nonumber
\end{align}
is the so-called augmented Lagrangian with user-defined \emph{penalty parameter} $\rho>0$. Variables $\lagfirst$, $\lagsec$ and $\lagtre$ correspond to the Lagrange multipliers of \eqref{eq:app}. 

To evaluate \eqref{ADMM1a}, one can leverage the proximal
operator of the nuclear norm~\cite[Th. 2.1]{cai2010singular}
\begin{align*}
\mathrm{prox}_\tau(\bm Z) \define \argmin_{\bm Y}~||\bm Y||_* + \frac{1}{2\tau}||\bm Y-\bm Z||_\text{F}^2
= D_\tau(\bm Z),
\end{align*}
where the \emph{singular value shrinkage} operator $D_\tau$ is defined for $\bm Z$ with SVD $\bm Z = \bm V_{\bm Z} \bm \Sigma_{\bm Z}\bm V_{\bm Z}\transpose$ as $D_\tau(\bm Z) \define
\bm V_{\bm Z} D_\tau(\bm \Sigma_{\bm Z})\bm V_{\bm Z}\transpose$ with
$ D_\tau ({\bm \Sigma}_{\bm Z})$ a diagonal matrix whose $(i,i)$-th entry
equals $\max(0, (\bm \Sigma_{\bm Z})_{i,i}-\tau )$.  
 With this operator,  \eqref{ADMM1a} becomes
\begin{subequations}
\begin{align}
\nucparmat^{k+1}=\mathrm{prox}_{\regnucpar/\rho}(\mathrm{vec}^{-1}(\linfunmat\mathrm{vec} (\fparmat^{k})-\lagsec^{k})) \label{eq:Y_par}\\
\nucpermat^{k+1}=\mathrm{prox}_{\regnucper/\rho}(\mathrm{vec}^{-1}(\linearmat\mathrm{vec} (\fpermat^{k})-\lagtre^{k})) \label{eq:Y_per}.
\end{align}
\end{subequations}
On the other hand,  \eqref{ADMM1ab} can be obtained in closed-form as
\begin{subequations}
\begin{align}
\label{eq:F_par}
&\mathrm{vec}(\fparmat^{k+1})=[\consparmat^{\top}\consparmat+\linfunmat^{\top}\linfunmat]^{-1}\\&[\consparmat^{\top}(\linearvec-\lagfirst^{k}-\conspermat\mathrm{vec}(\fpermat^{k}))
+\linfunmat^{\top}(\lagsec^{k}+\mathrm{vec}(\nucparmat^{k+1}))],\nonumber\\
\label{eq:F_per}
&\mathrm{vec}(\fpermat^{k+1})=[2\bm{I}+\rho\conspermat^{\top}\conspermat+\rho\linearmat^{\top}\linearmat]^{-1}\\&[\rho\conspermat^{\top}(\linearvec-\lagfirst^{k}-\consparmat\mathrm{vec}(\fparmat^{k}))
+\rho\linearmat^{\top}(\lagtre^{k}+\mathrm{vec}(\nucpermat^{k+1}))]\nonumber.
\end{align}
\end{subequations}
The overall ADMM algorithm is summarized as
Algorithm~\ref{exact-algorithm} and drastically improves the converegence rate of
our previous algorithm, reported in~\cite{mollaebrahim2018large}, since the latter
is based on subgradient descent.

\cmt{complexity}\arev{The computational complexity is dominated either  by \eqref{eq:Y_par} or by \eqref{eq:Y_per}, whose SVDs respectively require $O(\subspacedim^6)$  and  $O((\nodenum-\subspacedim)^6)$  arithmetic operations, depending on which quantity is larger. This complexity is much lower than the one of  general-purpose convex solvers and does not limit the values of $\nodenum$ and $\subspacedim$ to be used in practice given the intrinsic limitations of graph filters; see Sec.~\ref{sec:conclusions}.
}

\begin{algorithm}[t]
\caption{Iterative solver for \eqref{prob:exactrelaxed}.}
\label{exact-algorithm}
\begin{algorithmic}[1]
\REQUIRE $\edgeset$, $\uparmat$, 
 $\ADMMreg$, $\regnucpar$, $\regnucper$, $\eigenreg$.
\STATE{Obtain $\upermat$ s.t. $\upermat\transpose\upermat=\bm
I_{\nodenum-\subspacedim}$ and $\uparmat\transpose\upermat=\bm 0_{\subspacedim\times\nodenum}$.}
\STATE{Initialize
$\fparmat^0,~\fpermat^0,~ \lagfirst^{0},~\lagsec^{0},~\lagtre^{0}, k=0
$.}

\WHILE{stopping\_criterion\_not\_met()}
\STATE{obtain $\nucparmat^{{k+1}}$ via \eqref{eq:Y_par}.}	
\STATE{obtain $\nucpermat^{{k+1}}$ via \eqref{eq:Y_per}.}	
\STATE{obtain $\fparmat^{{k+1}}$ via \eqref{eq:F_par}.}	
\STATE{obtain $\fpermat^{{k+1}}$ via \eqref{eq:F_per}.}	
\STATE{ obtain $\lagfirst^{k+1},\lagsec^{k+1},\lagtre^{k+1}$ via \eqref{ADMM1}.}	
\STATE{ $k\leftarrow k+1$.}	
	\ENDWHILE
\RETURN	$ \shiftmat = \uparmat\fparmat^{k}\uparmat\transpose
+ \upermat\fpermat^{k}\upermat\transpose$.
\end{algorithmic}
\end{algorithm}

%% file: not_main/admm-approximate.tex

\cmt{Overview}This section presents a method to solve
\ref{prob:approxrelaxed} via ADMM.
\begin{myitemize}
\myitem\cmt{Objective}To this end, the first step is to express the objective function in a
suitable form. 
\begin{myitemize}
\myitem\cmt{First and second terms}With $\linfunmat$ and $\linearmat$
defined in Appendix~\ref{sec:exactsolver}, the first and second
terms of the objective of \eqref{prob:approxrelaxed} are respectively
proportional to
\begin{align}
&{{\left\| \vect^{-1}(\linfunmat\vect (\uparmat^{\top}\shiftmat\uparmat)) \right\|}_{*}}=\nonumber\\&{{\left\| \vect^{-1}(\linfunmat(\uparmat^{\top}\otimes\uparmat^{\top})\vect (\shiftmat)) \right\|}_{*}}={{\left\| \vect^{-1}(\linfunmatpar\vect (\shiftmat)) \right\|}_{*}}\nonumber
\end{align}
and
\begin{align}
& {{\left\| \vect^{-1}(\linearmat\vect (\upermat^{\top}\shiftmat\upermat)) \right\|}_{*}}=\nonumber\\&{{\left\| \vect^{-1}(\linearmat(\upermat^{\top}\otimes\upermat^{\top})\vect (\shiftmat)) \right\|}_{*}}={{\left\| \vect^{-1}(\linearmatperp\vect (\shiftmat)) \right\|}_{*}}\nonumber,
\end{align}
  where
  $\linfunmatpar\define\linfunmat(\uparmat^{\top}\otimes\uparmat^{\top})$
  and
  $\linearmat_{\perp}\define\linearmat(\upermat^{\top}\otimes\upermat^{\top})$. Noting
  that the nuclear norm of a block diagonal matrix equals the sum of
  the nuclear norms of each block enables one to compactly express the first two
  terms in \eqref{prob:approxrelaxed} as
\begin{align}
   \regnucpar{{\left\| \vect^{-1}(\linfunmat_{\parallel}\vect(\shiftmat)) \right\|}_{*}}
           +\regnucper{{\left\|
               \vect^{-1}(\linearmat_{\perp}\vect(\shiftmat))\right\|}_{*}}
           \nonumber
           =   {{\left\| \nucmat \right\|}_{*}}
\end{align}
where
\begin{align}
  \label{eq:nucmatconstr}
\nucmat\define	\begin{bmatrix}
                                                         \regnucpar\vect^{-1}(\linfunmat_{\parallel}\vect(\shiftmat)) &\mathbf{0}_{\subspacedim^2\times{(\nodenum-\subspacedim)^2}}\\
                                                             \mathbf{0}_{(\nodenum-\subspacedim)^2\times{\subspacedim^2}}&\regnucper\vect^{-1}(\linearmat_{\perp}\vect(\shiftmat))             
\end{bmatrix}.
\end{align}                                                                        
\myitem\cmt{Third and fourth terms}On the other hand, due to the properties of
the Kronecker product, the third and fourth terms
in the objective of \eqref{prob:approxrelaxed} can be written as ${{\| (\upermat^{\top}\otimes\upermat^{\top})\vect(\shiftmat) \|}_{2}^{2}}+\lambda{{\| (\uparmat^{\top}\otimes\upermat^{\top})\vect(\shiftmat) \|}_{2}^{2}}$.

\end{myitemize}%

\myitem\cmt{constraints}Regarding the constraints, note that
$\shiftmat=\shiftmat\transpose$ can be expressed as
$\symmat\vect(\shiftmat)=\mathbf{0}$, where the rows of
$\symmat\in\mathbb{R}^{{(\nodenum^{2}-\nodenum)/2\times{\nodenum}^{2}}}$
are given by
$(\canbasisvec_{j}^{\top}\otimes\canbasisvec_{i}^{\top}-\canbasisvec_{i}^{\top}\otimes\canbasisvec_{j}^{\top})$
for all $i,j=1,\ldots,\nodenum$ such that $i<j$. Here,
$\canbasisvec_{i}$ denotes the $i$-th column of $\bm I_\nodenum$. From
these considerations and using the definition of
$\topologyconstraintmat$ in Sec.~\ref{sec:filterorderminimization}, problem
\eqref{prob:approxrelaxed} can be written as:
\leqnomode
\begin {align}
  \tag{P2-R'}\label{eq:app1} 
  \begin{aligned}
\minimize_{\shiftmat,\nucmat}~~ &  {{\left\| \nucmat \right\|}_{*}}{{+\left\| (\upermat^{\top}\otimes\upermat^{\top})\vect(\shiftmat) \right\|}_{2}^{2}}\\&+\lambda{{\left\| (\uparmat^{\top}\otimes\upermat^{\top})\vect(\shiftmat) \right\|}_{2}^{2}}\\
\st   \;\;\;\;    	&
 \eqref{eq:nucmatconstr},\\
&\consmat\vect(\shiftmat)=\exactpfeasvec,
  \end{aligned}
\end{align}
\reqnomode
where $\consmat\overset{\Delta}{=}[                           \topologyconstraintmat;
									\symmat;
	                                                    \mathrm{vec^{\top}}(\uparmat\uparmat^{\top});
                                                             \mathrm{vec^{\top}}(\upermat\upermat^{\top})]$
and $\exactpfeasvec\overset{\Delta}{=}[
                                                         \mathbf{0};
									\mathbf{0};
	                                                    \subspacedim;
                                                             (1\pm \eigenreg)(\nodenum-\subspacedim)]$.
\myitem\cmt{$\nucmat$ constraint}
By considering each block separately, \eqref{eq:nucmatconstr} holds iff
\begin{align}
\nonumber
 & \regnucpar\linfunmat_{\parallel}\vect(\shiftmat)= \vect(\consmat_{1}\nucmat\consmat_{1}^{\top})=(\consmat_{1}\otimes\consmat_{1})\vect(\nucmat) ,\\
 & \mathbf{0}= \vect(\consmat_{2}\nucmat\consmat_{1}^{\top})=(\consmat_{1}\otimes\consmat_{2})\vect(\nucmat) \nonumber,
\\ & \mathbf{0}= \vect(\consmat_{1}\nucmat\consmat_{2}^{\top})=(\consmat_{2}\otimes\consmat_{1})\vect(\nucmat) \nonumber,
\\ & \regnucper\linearmat_{\perp}\vect(\shiftmat)= \vect(\consmat_{2}\nucmat\consmat_{2}^{\top})=(\consmat_{2}\otimes\consmat_{2})\vect(\nucmat) \nonumber,
\end{align}
where
$\consmat_{1}\define[\identity_{r^2\times{r}^2},\mathbf{0}_{r^2\times({N-r})^{2}}]$ and
$\consmat_{2}\define[\mathbf{0}_{(N-r)^2\times{r}^2},\identity_{(N-r)^2\times{(N-r)}^2}]$. More
compactly, \eqref{eq:nucmatconstr} holds iff 
\begin{align}\label{cast}
&\exactpfeasmat\vect(\nucmat)=\feasmat\vect(\shiftmat)
\end{align}
where  $\exactpfeasmat=[\consmat_{1}\otimes\consmat_{1};
  \consmat_{1}\otimes\consmat_{2};
  \consmat_{2}\otimes\consmat_{1};\consmat_{2}\otimes\consmat_{2}]$ and $\feasmat=[\regnucpar\linfunmat_{\parallel}; \mathbf{0}; \mathbf{0}; \regnucper\linearmat_{\perp}]$. 
Noting that $\exactpfeasmat$ is orthogonal enables one to rewrite
\eqref{eq:app1}  as
\leqnomode
\begin {align}
  \tag{P2-R''}\label{main_2}
  \begin{aligned}
\minimize_{\shiftmat,\nucmat}~~ &  {{\left\| \nucmat \right\|}_{*}}{{+\left\| (\upermat^{\top}\otimes\upermat^{\top})\vect(\shiftmat) \right\|}_{2}^{2}}\\&+\lambda{{\left\| (\uparmat^{\top}\otimes\upermat^{\top})\vect(\shiftmat) \right\|}_{2}^{2}}\\
\st   \;\;\;\;    	&\vect(\nucmat)=\exactpfeasmat\transpose\feasmat~\vect(\shiftmat),\\
&
\consmat\vect(\shiftmat)=\exactpfeasvec.
  \end{aligned}
\end{align}
\reqnomode

\end{myitemize}%

\cmt{ADMM}
\begin{myitemize}%

\myitem\cmt{ADMM algo}The ADMM method
in scaled form  \cite[Sec. 3.1.1]{boyd2011distributed} applied
to \eqref{main_2} reads as:
 \begin{subequations} \label{ADMM}
\begin{align}
&{\shiftmat^{k+1}}:=
  \underset{\shiftmat}{\text{argmin}}~\bar L_{\ADMMreg}(\shiftmat, {\nucmat^{k}},\dummymat_1^k,\dummymat_2^k)\label{ADMMa}\\
&{\nucmat^{k+1}}:=
  \underset{\nucmat}{\text{argmin}}~\bar L_{\ADMMreg}(\shiftmat^{k+1}, {\nucmat},\dummymat_1^k,\dummymat_2^k)\label{eq:app:yupdate}\\
&\lagfirst^{k+1}:=\lagfirst^{k}+\vect(\nucmat^{k+1})-\exactpfeasmat\transpose\feasmat\vect(\shiftmat^{k+1})\label{ADMM_con}\\
&\lagsec^{k+1}:=\lagsec^{k}+\consmat\vect(\shiftmat^{k+1})-\exactpfeasvec\label{q2},
\end{align}
\end{subequations}
where
\begin{align}
\bar L_{\ADMMreg}&(\shiftmat, {\nucmat},\lagfirst,\lagsec)\define    {{\left\| \nucmat   \right\|}_{*}} +{{\left\| (\upermat^{\top}\otimes\upermat^{\top})\vect(\shiftmat) \right\|}_{2}^{2}}\nonumber\\&+\lambda{{\left\| (\uparmat^{\top}\otimes\upermat^{\top})\vect(\shiftmat) \right\|}_{2}^{2}}\nonumber\\&+(\ADMMreg/2){{\left\| \vect(\nucmat)-\exactpfeasmat\transpose\feasmat\vect(\shiftmat)+\lagfirst \right\|}_{2}^{2}}\nonumber\\&+(\ADMMreg/2){{\left\| \consmat\vect(\shiftmat)-\exactpfeasvec+\lagsec \right\|}_{2}^{2}}\nonumber,
\end{align}
where is the augmented Lagrangian  with Lagrange multipliers  $\lagfirst$ and $\lagsec$.

\myitem\cmt{$\dummymat_1$ update}To express \eqref{ADMM_con} in a more
convenient form, observe first from the definitions of
$\exactpfeasmat$ and $\feasmat$ that
$\exactpfeasmat^{\top}\feasmat=\regnucpar(\consmat_{1}^{\top}\otimes\consmat_{1}^{\top})
\linfunmat_{\parallel} +\regnucper
(\consmat_{2}^{\top}\otimes\consmat_{2}^{\top})\linearmat_{\perp}
$. Second, note from the properties of the Kronecker product that
$\linfunmat_{\parallel}\vect(\shiftmat)=\linfunmat(\uparmat^{\top}\otimes\uparmat^{\top})\vect(\shiftmat)=
\linfunmat\vect(\uparmat^{\top}\shiftmat\uparmat)=\vect(\uparmat^{\top}\shiftmat\uparmat\otimes\identity_\subspacedim-\identity_\subspacedim\otimes\uparmat^{\top}\shiftmat\uparmat)$.
Similarly,  $\linearmat_{\perp}\vect(\shiftmat)=
\vect(\upermat^{\top}\shiftmat\upermat\otimes\identity_{\nodenum-\subspacedim}-\identity_{\nodenum-\subspacedim}\otimes\upermat^{\top}\shiftmat\upermat)$.
Consequently, 
\begin{align}\nonumber
&\exactpfeasmat^{\top}\feasmat\vect(\shiftmat)\\=&\regnucpar(\consmat_{1}^{\top}\otimes\consmat_{1}^{\top})\vect(\uparmat^{\top}\shiftmat\uparmat\otimes\identity-\identity\otimes\uparmat^{\top}\shiftmat\uparmat)\nonumber\\&+\regnucper(\consmat_{2}^{\top}\otimes\consmat_{2}^{\top})\vect(\upermat^{\top}\shiftmat\upermat\otimes\identity-\identity\otimes\upermat^{\top}\shiftmat\upermat)\nonumber\\
=&\regnucpar\vect(\consmat_{1}^{\top}(\uparmat^{\top}\shiftmat\uparmat\otimes\identity-\identity\otimes\uparmat^{\top}\shiftmat\uparmat)\consmat_{1})+\nonumber\\&\regnucper\vect(\consmat_{2}^{\top}(\upermat^{\top}\shiftmat\upermat\otimes\identity-\identity\otimes\upermat^{\top}\shiftmat\upermat)\consmat_{2})\nonumber
\\=&\vect[\regnucpar(\uparmat^{\top}\shiftmat\uparmat\otimes\identity-\identity\otimes\uparmat^{\top}\shiftmat\uparmat),\mathbf{0};\label{eq:app:cdvs}
  \\&\quad~~\mathbf{0},\regnucper(\upermat^{\top}\shiftmat\upermat\otimes\identity-\identity\otimes\upermat^{\top}\shiftmat\upermat)]\nonumber.
\end{align}
\balance
The update \eqref{ADMM_con} can therefore be expressed
upon defining $\dummymat_1^k\define\vect\inv(\lagfirst^{k})$
as
\begin{align}
&\dummymat_1^{k+1}:=\dummymat_1^{k}+\nucmat^{k+1}\label{eq:app:qud}\\&-
[\regnucpar(\uparmat^{\top}\shiftmat^{k+1}\uparmat\otimes\identity_\subspacedim-\identity_\subspacedim\otimes\uparmat^{\top}\shiftmat^{k+1}\uparmat),\mathbf{0};\nonumber\\&~~\quad\mathbf{0},\regnucper(\upermat^{\top}\shiftmat^{k+1}\upermat\otimes\identity_{\nodenum-\subspacedim}-\identity_{\nodenum-\subspacedim}\otimes\upermat^{\top}\shiftmat^{k+1}\upermat)]\nonumber.
\end{align}

\myitem\cmt{$\nucmat$ update}Furthermore, by using the proximal
operator of the nuclear norm (see  Appendix~\ref{sec:exactsolver}),
the update in \eqref{eq:app:yupdate} becomes 
\begin{align}
&\nucmat^{k+1}=\mathrm{prox}_{1/\ADMMreg}(\vect^{-1}(\exactpfeasmat\transpose\feasmat\vect(\shiftmat^{k})-\lagfirst^{k})) \nonumber\\
&=\mathrm{prox}_{1/\ADMMreg}([\regnucpar(\uparmat^{\top}\shiftmat^{k+1}\uparmat\otimes\identity_\subspacedim-\identity_\subspacedim\otimes\uparmat^{\top}\shiftmat^{k+1}\uparmat),\mathbf{0};\nonumber\\&\mathbf{0},\regnucper(\upermat^{\top}\shiftmat^{k+1}\upermat\otimes\identity_{\nodenum-\subspacedim}-\identity_{\nodenum-\subspacedim}\otimes\upermat^{\top}\shiftmat^{k+1}\upermat)]-\dummymat_1^{k}),\label{eq:app:yud}
\end{align}
where the second equality follows from  \eqref{eq:app:cdvs}. 

\myitem\cmt{$\nucmat$ and $\dummymat_1$ updates}If $\dummymat_1^0$ is
initialized as a block diagonal matrix with diagonal blocks of size
$\subspacedim \times \subspacedim$ and
$(\nodenum-\subspacedim)\times(\nodenum-\subspacedim)$, then it
follows from \eqref{eq:app:qud} and \eqref{eq:app:yud} that
$\nucmat^{k}$ and $\dummymat_1^{k}$ will remain block diagonal with
blocks of the same size. These blocks can be updated as
\begin{subequations}
\label{eq:updatesyyqq}
\begin{align}
  \nucmat_1^{k+1}=&\mathrm{prox}_{1/\ADMMreg}\Big(\regnucpar(\uparmat^{\top}\shiftmat^{k+1}\uparmat\otimes
  \identity_\subspacedim\\&\quad\quad\quad~-\identity_\subspacedim\otimes\uparmat^{\top}\shiftmat^{k+1}\uparmat)-\dummymat_{1,1}^{k}\Big)
  \nonumber\\
  \nucmat_2^{k+1}=&\mathrm{prox}_{1/\ADMMreg}\Big(\regnucper(\upermat^{\top}\shiftmat^{k+1}\upermat\otimes\identity_{\nodenum-\subspacedim}\\\nonumber
  &\quad\quad\quad~-\identity_{\nodenum-\subspacedim}\otimes\upermat^{\top}\shiftmat^{k+1}\upermat)-\dummymat_{1,2}^{k}\Big)\\
\dummymat_{1,1}^{k+1}=&\dummymat_{1,1}^{k}+\nucmat_1^{k+1}\\&-\nonumber
\regnucpar(\uparmat^{\top}\shiftmat^{k+1}\uparmat\otimes\identity-\identity\otimes\uparmat^{\top}\shiftmat^{k+1}\uparmat)\\
\dummymat_{1,2}^{k+1}=&\dummymat_{1,2}^{k}+\nucmat_2^{k+1}\\&-\nonumber
\regnucper(\upermat^{\top}\shiftmat^{k+1}\upermat\otimes\identity-\identity\otimes\upermat^{\top}\shiftmat^{k+1}\upermat).
\end{align}
\end{subequations}
\myitem\cmt{$\shiftmat$-update}The $\shiftmat$-update \eqref{ADMMa}
  can be obtained in closed form as
\begin{align}\label{closed_1}
\vect(\shiftmat^{k+1})=&\Big[\feasmat^{\top}\feasmat+\consmat^{\top}\consmat+(2\lambda/\ADMMreg)(\uparmat\uparmat^{\top}\otimes\upermat\upermat^{\top})\nonumber\\&+
(2/\ADMMreg)(\upermat\upermat^{\top}\otimes\upermat\upermat^{\top})\Big]^{-1}\\&
\Big[\feasmat\transpose\exactpfeasmat(\vect(\nucmat^{k})+\lagfirst^{k})+\consmat^{\top}(\exactpfeasvec-\lagsec^{k})\Big] \nonumber
\end{align}
\begin{myitemize}%
\myitem\cmt{orth}by using the orthogonality of  $\exactpfeasmat$.
\myitem\cmt{simplify DtC}To reduce the computational complexity
of \eqref{closed_1}, note that 
\begin{align}
&\feasmat^\top\exactpfeasmat\vect(\nucmat^{k})=
\regnucpar\linfunmat_{\parallel}^{\top}(\consmat_{1}\otimes\consmat_{1})\vect(\nucmat^{k})+\nonumber\\& \regnucper\linearmat_{\perp}^{\top}(\consmat_{2}\otimes\consmat_{2})\vect(\nucmat^{k})= \regnucpar\linfunmat_{\parallel}^{\top}\vect(\nucmat_1^{k})+\nonumber\\ &\regnucper\linearmat_{\perp}^{\top}\vect(\nucmat_2^{k})\nonumber
\end{align}
and
\begin{align}
\feasmat^\top\exactpfeasmat\lagfirst^{k}
&=
\regnucpar\linfunmat_{\parallel}^{\top}\vect(\dummymat_{1,1}^{k})+ \regnucper\linearmat_{\perp}^{\top}\vect(\dummymat_{1,2}^{k})\nonumber
\end{align}
to obtain
\begin{align}\label{adm1}
&\vect(\shiftmat^{k+1})=\Big[
\feasmat^{\top}\feasmat+\consmat^{\top}\consmat+(2\lambda/\ADMMreg)
(\uparmat\uparmat^{\top}\otimes\upermat\upermat^{\top})\nonumber\\&+(2/\ADMMreg)(\upermat\upermat^{\top}\otimes\upermat\upermat^{\top})\Big]^{-1}\Big[\regnucpar\linfunmat_{\parallel}^{\top}\vect(\nucmat_1^{k}+\dummymat_{1,1}^{k})\nonumber\\
&
+ \regnucper\linearmat_{\perp}^{\top}\vect(\nucmat_2^{k}+\dummymat_{1,2}^{k})
+\consmat^{\top}(\exactpfeasvec-\lagsec^k)\Big].
\end{align}

\end{myitemize}%

\cmt{summary}The overall procedure is summarized as Algorithm~\ref{approx-algorithm}.
\cmt{complexity}\arev{Its complexity is dominated by the inversion in \eqref{adm1}, which involves $O(\nodenum^6)$ arithmetic operations.  This complexity is much lower than the one of  general-purpose convex solvers and does not limit the values of $\nodenum$ to be used in practice given the intrinsic limitations of graph filters; see Sec.~\ref{sec:conclusions}.}


\end{myitemize}%

\begin{algorithm}[t]
\caption{Iterative solver for \eqref{prob:approxrelaxed}.}
\label{approx-algorithm}
\begin{algorithmic}[1]
\REQUIRE $\edgeset$, $\uparmat$, 
 $\ADMMreg$, $\lambda$, $\regnucpar$, $\regnucper$, $\eigenreg$.
\STATE{Obtain $\upermat$ s.t. $\upermat\transpose\upermat=\bm
I_{\nodenum-\subspacedim}$ and $\uparmat\transpose\upermat=\bm 0_{\subspacedim\times\nodenum}$.}
\STATE{Initialize
$\nucmat_{1}^0,~\nucmat_{2}^0,~ \dummymat_{1,1}^0,~\dummymat_{1,2}^0$,
$\lagsec^0$, $k=0$}

\WHILE{stopping\_criterion\_not\_met()}
\STATE{obtain ${\shiftmat}^{{k+1}}$ via \eqref{adm1}.}	
\STATE{ obtain $  \nucmat_1^{k+1}$, $  \nucmat_2^{k+1}$,
$\dummymat_{1,1}^{k+1}$, $\dummymat_{1,2}^{k+1}$ via \eqref{eq:updatesyyqq}.}	
\STATE{ obtain $\lagsec^{k+1}$ via \eqref{q2}.}
\STATE{ $k\leftarrow k+1$.}	
	\ENDWHILE
\RETURN	 $\shiftmat^k$	
\end{algorithmic}
\end{algorithm}